\documentclass[draftclsnofoot,12pt, onecolumn]{IEEEtran}
\usepackage{cite}
\usepackage{amsmath,amssymb,amsfonts,dsfont}
\usepackage{algorithm}
\usepackage{algorithmic}
\usepackage{graphicx}
\usepackage{subfigure}
\usepackage{textcomp}
\usepackage{xcolor}
\usepackage{booktabs}
\usepackage{breqn}
\usepackage{subeqnarray}
\usepackage{cases}
\usepackage{setspace}
\usepackage{amsthm}
\usepackage{bbm}
\usepackage{bm}
\usepackage{lipsum}
\usepackage{enumitem}
\usepackage{mathtools}
\allowdisplaybreaks[4]

\newtheorem{thm}{Theorem}

\newtheorem{lem}{Lemma}
\newtheorem{prop}{Proposition}
\theoremstyle{remark}

\newtheorem{rem}{Remark}
\newtheorem{ass}{Assumption}

\DeclareMathOperator*{\argmin}{arg\,min}

\def\BibTeX{{\rm B\kern-.05em{\sc i\kern-.025em b}\kern-.08em
    T\kern-.1667em\lower.7ex\hbox{E}\kern-.125emX}}

\begin{document}
\begin{titlepage}
\title{Knowledge-aided Federated Learning for Energy-limited Wireless Networks\\
}
\author{Zhixiong~Chen,~\IEEEmembership{Student Member,~IEEE},
Wenqiang~Yi,~\IEEEmembership{Member,~IEEE},\\
Yuanwei~Liu,~\IEEEmembership{Senior Member,~IEEE},
and Arumugam~Nallanathan,~\IEEEmembership{Fellow,~IEEE}
\thanks{Zhixiong Chen, Wenqiang Yi, Yuanwei Liu, and Arumugam Nallanathan are with the School of Electronic Engineering and Computer Science, Queen Mary University of London, London, U.K. (emails: \{zhixiong.chen, w.yi, yuanwei.liu, a.nallanathan\}@qmul.ac.uk).}
\thanks{Part of this work has been accepted to IEEE International Conference on Communications (ICC) 2023 \cite{chen2022kfl}.}
}

\maketitle
\vspace{-1.9cm}
\begin{abstract}
The conventional model aggregation-based federated learning (FL) approach requires all local models to have the same architecture, which fails to support practical scenarios with heterogeneous local models. Moreover, the frequent model exchange is costly for resource-limited wireless networks since modern deep neural networks usually have over a million parameters.
To tackle these challenges, we first propose a novel knowledge-aided FL (KFL) framework, which aggregates light high-level data features, namely knowledge, in the per-round learning process.
This framework allows devices to design their machine-learning models independently and reduces the communication overhead in the training process.
We then theoretically analyze the convergence bound of the proposed framework under a non-convex loss function setting, revealing that scheduling more data volume in each round helps to improve the learning performance.
In addition, large data volume should be scheduled in early rounds if the total scheduled data volume during the entire learning course is fixed.
Inspired by this, we define a new objective function, i.e., the weighted scheduled data sample volume, to transform the inexplicit global loss minimization problem into a tractable one for device scheduling, bandwidth allocation, and power control.
To deal with unknown time-varying wireless channels, we transform the considered problem into a deterministic problem for each round with the assistance of the Lyapunov optimization framework. Then, we derive the optimal bandwidth allocation and power control solution by convex optimization techniques.
We also develop an efficient online device scheduling algorithm to achieve an energy-learning trade-off in the learning process.
Experimental results on two typical datasets (i.e., MNIST and CIFAR-10) under highly heterogeneous local data distributions show that the proposed KFL is capable of reducing over 99\% communication overhead while achieving better learning performance than the conventional model aggregation-based algorithms. In addition, the proposed device scheduling algorithm converges faster than the benchmark scheduling schemes.
\end{abstract}
\vspace{-0.6cm}

\section{Introduction}
The increasing demands for intelligent services, such as augmented reality/virtual reality (AR/VR) and Internet-of-Things (IoT) applications, motivate the integration of machine learning in future wireless networks \cite{9687500}. Federated learning (FL) is one of the most promising distributed learning frameworks to reduce the communication traffic load of intelligent services, which enables devices to collaboratively train machine learning models by periodically exchanging model parameters between devices and the parameter server instead of raw user data \cite{9460016}. However, the model aggregation nature of conventional FL confronts the following two limitations for its implementation in wireless networks:
1) \emph{High Communication Overhead:} The uploading of model/gradient parameters is costly for devices since modern deep neural network (NN) architectures usually possess massive parameters. For instance, the widely used MobileNet \cite{Sandler_2018_CVPR}, a convolutional NN (CNN) for on-device image processing, has 6.9 million parameters, corresponding to 27.6 MB.
Training such a model requires devices to upload 27.6 MB of data per round. Considering hundreds of rounds and multiple devices, the communication overhead is heavy for wireless networks with limited spectrum and energy resources.
2) \emph{Heterogeneous Local Models:} In practical wireless networks, devices are usually equipped with different NNs in terms of architectures and model sizes due to their heterogeneous computing capabilities and storage resources \cite{9770094}. In this case, the traditional model aggregation-based FL approaches fail to coordinate devices to perform the learning process.
To break these two limitations, state-of-the-art studies focus on the designs of communication-efficient FL and heterogeneous FL.

\vspace{-0.5cm}
\subsection{Related Works}
To enable communication-efficient FL in resource-limited wireless networks, existing works mainly focused on device scheduling \cite{9237168, 9207871, 9760232, 9292468}, model quantization \cite{9277666, 9269459, 9712310, 9425020}, and model pruning \cite{9598845, 9409149, 9707474}.
Device scheduling methods select a small subset of devices to participate in the per-round training process, thus reducing the communication burden and mitigating the straggler effect when devices have random or heterogeneous computing speed.
The device selection and bandwidth allocation in \cite{9237168, 9207871} guaranteed long-term learning performance in bandwidth-limited wireless networks.
The probabilistic scheduling policy for FL proposed in \cite{9760232} effectively minimized the model uploading latency and improved convergence speed.
The co-design of learning and device selection in \cite{9292468} reduced convergence time in resource-constrained wireless networks. Although these device scheduling approaches efficiently alleviate communication burden, transmitting the entire model is arduous for devices with weak channels and limited energy.
To tackle this issue, the model quantization compresses devices' model updates before transmitting to the parameter server, thus reducing the transmitted data volume and communication overhead for devices \cite{9277666, 9269459}.
Specifically, the model quantization approach in \cite{9712310} enabled edge devices to adjust their quantization proportional according to their communication resources for balancing training accuracy and communication overhead.
The heterogeneous quantization method in \cite{9425020} allocated different aggregation weights to clients for efficiently improving convergence speed.
While the model quantization is demonstrably effective, it introduces additional noise during training, which ultimately degrading the trained model's performance.
The model pruning is able to simultaneously reduce communication and computation costs by removing less important weights from the original model.
The joint design of the pruning ratio and wireless resource allocation in \cite{9598845} significantly improved the convergence rate of FL.
In \cite{9409149}, the NN pruning was integrated into FL to improve learning speed and guarantee training latency.
A random model pruning approach was adopted in \cite{9707474} to generate several subnets from the global model to adapt the channel condition of different devices. It reduced both communication overhead and computation loads.
The above three approaches reduce communication overhead while degrading the final model's accuracy.
Besides these approaches, our previous work \cite{chen2022federated} enabled devices to train the feature extractor part of NNs collaboratively, while the predictor part for devices is localized for personalization. It reduced communication overhead and improved learning performance in heterogeneous data distribution scenarios.
However, these approaches still require heavy parameter transmission in the learning process.

To allow devices equipped with heterogeneous models in FL, knowledge distillation (KD)-based FL approaches were developed and attracted much attention. In practical wireless networks, devices usually possess different computation capabilities and communication resources. Thus, requiring all the local models to be of the same architecture in many application scenarios may be ineffective.
KD is a teacher-student paradigm which transfers the knowledge distilled from the teacher model to the student model \cite{temperaturesoftmax}.
Integrating KD into FL allows devices to independently design their models according to channel conditions and computation capabilities.
Specifically, the federated KD approach in \cite{9839214} effectively enabled federated training between heterogeneous models by aggregating local models' logits on a public dataset.
In \cite{9121290}, an auxiliary distillation dataset generated by mixing local training data was adopted to empower the FL process, effectively reducing convergence time.
In \cite{pmlr-v139-zhu21b}, a lightweight generator was deployed at the server to ensemble user information and broadcast to devices to regulate their local training process.
By deploying an unlabelled dataset on both the server and devices, a global model was trained using the averaged outputs of local models on this dataset as the supervision label \cite{8904164, NEURIPS2020_18df51b9}.
The adaptive mutual KD and dynamic gradient compression approach in \cite{wu2022communication} significantly reduced communication costs and achieved competitive results with centralized model learning.
The federated distillation method \cite{yu2020salvaging} regularized local models to mitigate overfitting during training by treating the global model as the teacher and the local models as the students.
Besides enabling devices to design their machine learning models independently, the KD-based FL substantially reduces the transmitted data volume in the wireless channels because output logits are required to upload in the learning process instead of heavy model/gradient parameters.
However, these KD-based FL approaches require an extra public dataset to align the student and teacher models' outputs, increasing the computation costs. Moreover, their performance may significantly degrade with the increase in the distribution divergence between the public and on-device datasets that are usually non-independent and identically distributed (non-IID).
\vspace{-0.4cm}
\subsection{Motivations and Contributions}
Although the communication-efficient FL in \cite{9237168, 9207871, 9760232, 9292468, 9277666, 9269459, 9712310, 9425020, 9598845, 9409149, 9707474} can reduce communication overhead, they degrade the final model accuracy and require heavy parameter transmission in the learning process. In addition, the KD-based FL in \cite{9839214, 9121290, pmlr-v139-zhu21b, 8904164, NEURIPS2020_18df51b9, wu2022communication, yu2020salvaging} allowed devices to ensemble heterogeneous local models. However, they rely on the public dataset, which may not be practical for many scenarios.
To break these limitations, this work aims to enable collaborative training for devices equipped with heterogeneous models in a communication-efficient way, avoiding the reliance on heavy model transmission and extra public datasets.
Inspired by the human experience in discriminating between different objects and its successful application in clustering analysis, the same class of objects or data usually have similar high-level features, while different objects have distinct features \cite{frades2010overview}. In addition, a general insight for modern deep learning models is that the lower layers (close to the input) are primarily responsible for feature extraction, while the upper layers (proximate to the output) focus on complex pattern recognition \cite{goodfellow2016deep}.
We aim to enable devices for collaborative training by aggregating their output of lower layers of the NNs, namely knowledge, in the per-round training process. This design effectively reduces the communication overhead since the dimensions of the knowledge are usually much smaller than that of the model.
The main contributions of this paper are summarized as follows:
\begin{itemize}
  \item We propose a novel KFL framework in which devices collaboratively train models by uploading their knowledge of different data classes to the edge server for aggregation. This design reduces the transmitted data volume in the wireless channels, allowing devices to design their machine-learning models  independently according to their computation capabilities and communication conditions.

  \item We theoretically analyze the convergence bound of the proposed KFL framework under the general non-convex loss function setting, which indicates that scheduling more data samples in each round is able to improve the learning performance. In addition, when the total number of scheduled data volume during the entire learning course is fixed, more data volume should be scheduled in the early rounds.
      Following the experimental investigation of temporal scheduling policies in \cite{9237168}, this work further theoretically analyzes how the temporal device scheduling patterns affect the final learning performance through the convergence analysis.

  \item We formulate a long-term device scheduling, bandwidth allocation, and power control problem under limited devices' energy budgets with the aid of the convergence bound. To deal with unpredicted time-varying wireless channels and enable online device scheduling, we first transform the original problem into a deterministic problem in each round with the assistance of the Lyapunov optimization framework. Then, we derive the optimal bandwidth allocation and power control through convex optimization techniques. Finally, we develop an efficient polynomial-time algorithm to solve the device scheduling policy with $\mathcal{O}(\sqrt{V}, 1/V)$ energy-learning trade-off guarantee, where $V$ is an algorithm-specific parameter.

  \item We experimentally verify the correctness of our theoretical results, i.e., more data samples should be scheduled in the early rounds when the total scheduled data volume in the entire learning course are fixed. Compared with benchmark FL algorithms, the proposed KFL framework saves 99\% communication overhead and boosts 2.1\% and 6.65\% accuracy on MNIST and CIFAR-10 datasets, respectively. In addition, The proposed online device scheduling algorithm achieves a faster convergence speed than benchmark scheduling approaches.
\end{itemize}
\vspace{-0.5cm}
\subsection{Organization and Notations}
The rest of this paper is organized as follows: In Section \ref{sec:system_model}, we introduce the proposed KFL system and learning cost, then formulate the global loss minimization problem.
The convergence analysis and problem transformation are illustrated in Section \ref{sec:convergence_ana_transform}.
The joint device scheduling, bandwidth allocation, and power control algorithm are developed in \ref{sec:alg_design}.
Section \ref{sec:simulation} verifies the effectiveness of the proposed scheme by simulation. The conclusion is drawn in Section \ref{sec:conclusion}.
For convenience, we use ``$\buildrel \Delta \over = $'' to denote ``is defined to be equal to'', $\left|\cdot\right|$ denote the size operation of a set, $\nabla(\cdot)$ denote gradient operator, $\left\langle {\cdot,\cdot} \right\rangle $ denote inner product operator, and ``$\left\|  \cdot  \right\|$'' denote the $\ell_2$ norm throughout this paper. The main notations used in this paper are summarized in Table \ref{tab:notation}.

\begin{table}[ht]\scriptsize
\vspace{-0.5cm}
\caption{Notation Summary}
\vspace{-0.2cm}
\label{tab:notation}
\begin{tabular}{p{1.8cm}|p{6.2cm}|p{1.8cm}|p{6.2cm}}
\hline
Notation & Definition &Notation & Definition \\
\hline
$\mathcal{K}$; $K$;   & Set of devices; size of $\mathcal{K}$
& $\mathcal{C}$; $C$; & Set of classes; size of $\mathcal{C}$ \\
$\mathcal{D}_k$; $D_k$ & Local dataset of device $k$; size of $\mathcal{D}_k$
& $\mathcal{D}$; $D$ & Overall dataset in the system; size of $\mathcal{D}$ \\
$\mathcal{D}_{k,c}$; $D_{k,c}$ & Local dataset of $c$ class; size of $\mathcal{D}_{k,c}$
& $\bm{w}_k$; $\bm{u}_k$; $\bm{v}_k$; $\bm{W}$ & Local model; local feature extractor; local predictor of device $k$; all local models \\
$F_k(\bm{u}_k,\bm{v}_k)$; $F(\bm{W})$ & Local empirical loss function of device $k$; global empirical loss function
& $\eta_u$; $\eta_v$ & Learning rate for feature extractor and predictor \\
$L_k(\bm{u}_k)$ & Local knowledge loss function
& $\lambda$ & Knowledge loss weight \\
$\bm{\Omega}_{k,c}$; $\bm{\Omega}_k$ & Device $k$'s knowledge about class $c$; device $k$'s knowledge for all classes
&$\bm{\Omega}_{c}$; $\bm{\Omega}$ & Global knowledge about class $c$; global knowledge about all classes \\
$\bm{S}_t$; $\tau$ & Scheduling policy in round $t$, i.e., the set of scheduled devices; local iteration number
& $f_k$; $C_{k}$ & CPU frequency of device $k$; Computation workload of one data sample at device $k$ \\
$p_{k,t}$; $p_{k,\max}$ & Transmit power of device $k$ in round $t$; maximum transmit power of device $k$
& $C_{k}$; $Q$ & Computation workload of one data sample at device $k$; Data size of local knowledge\\
$B$; $\bm{\theta}_t$ & Wireless bandwidth;  the proportion of $B$ allocated to devices in round $t$
& $E_k$; $\mathcal{T}_{\max}$ & Total energy budget of device $k$; Maximum completion time for each round\\
\hline
\end{tabular}
\vspace{-1.2cm}
\end{table}

\vspace{-0.2cm}
\begin{figure*}
\centering
\subfigure[]{\label{fig:sysmodel}
\includegraphics[width=0.45\linewidth]{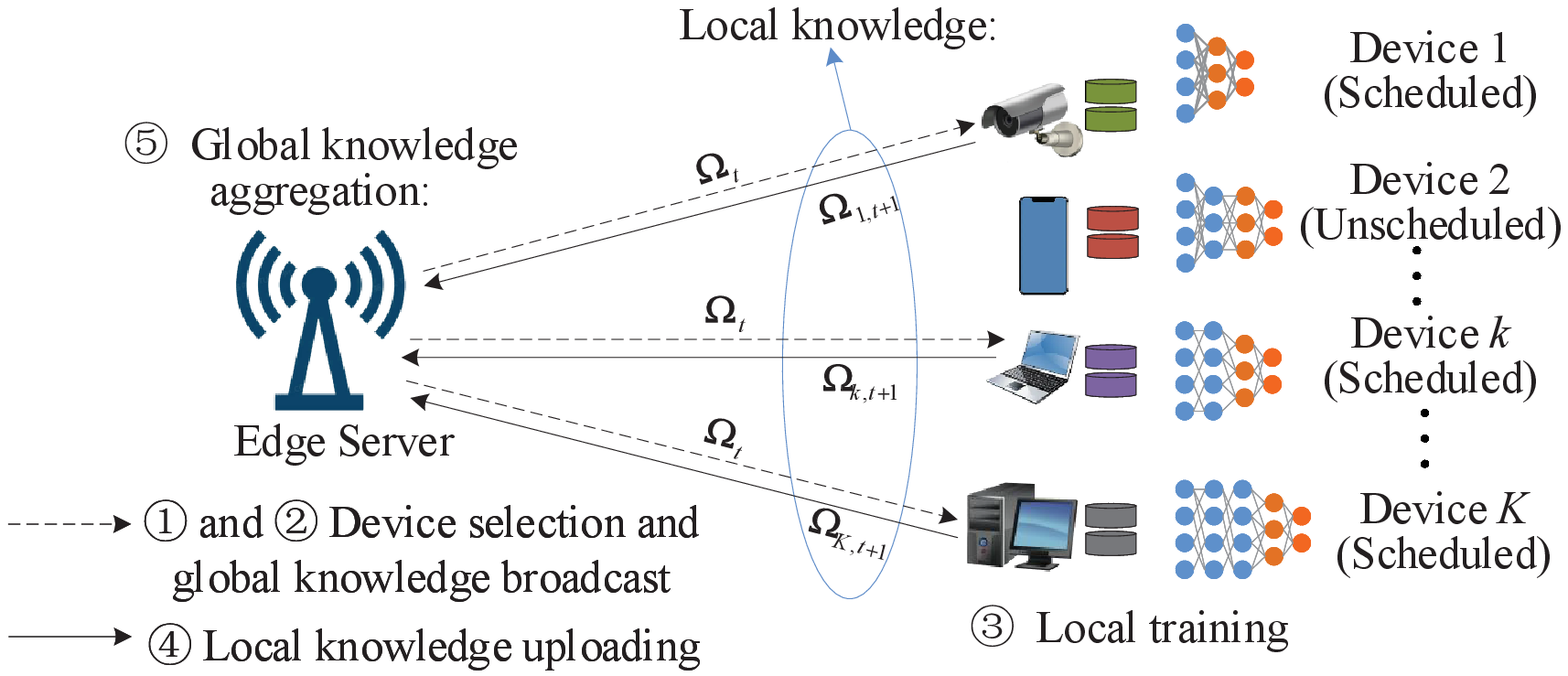}}
\hspace{0.01\linewidth}
\subfigure[]{\label{fig:local_training}
\includegraphics[width=0.3\linewidth]{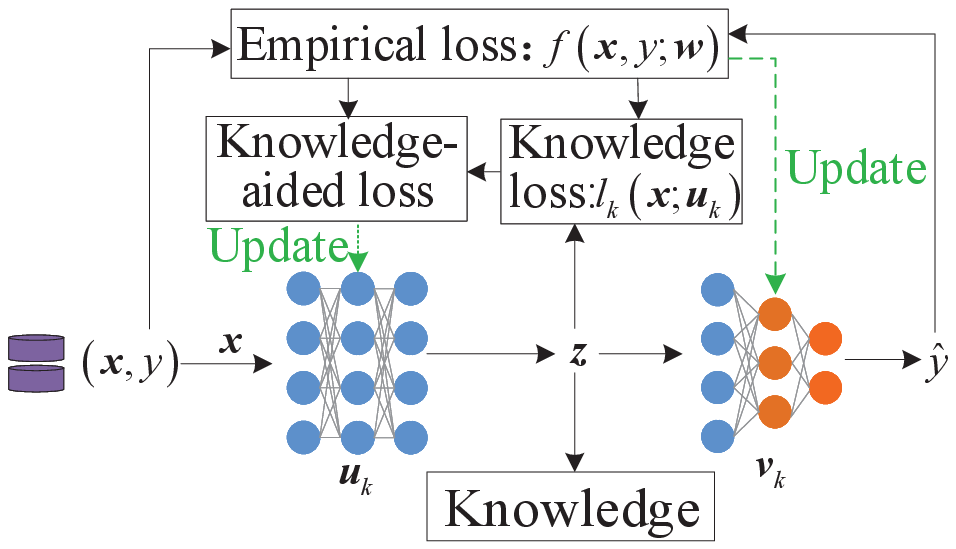}}
\vspace{-0.3cm}
\caption{The illustrated KFL over wireless networks: (a) Federated learning with knowledge aggregation mechanism, where devices have different local models; and (b) Local training process with the proposed knowledge-aided loss.}
\label{fig:sys_model}
\vspace{-1cm}
\end{figure*}

\section{System Model and Learning Mechanism}\label{sec:system_model}
In the considered KFL system, as shown in Fig. \ref{fig:sys_model}, an edge server coordinates $K$ different devices to train machine learning models for classification or recognition tasks. Unlike the conventional FL that requires all devices' models to be of the same architecture, the KFL in this work allows devices to be equipped with heterogeneous models. The devices are indexed by $\mathcal{K} = \left\{ 1,2, \cdots ,K \right\}$. For the dataset at devices, the number of data classes in the classification or recognition task is $C$, indexed by $\mathcal{C} = \left\{ 1,2, \cdots, C \right\}$.
Each device $k$ ($k \in \mathcal{K}$) has a local dataset $\mathcal{D}_k$ with $D_k = \left| \mathcal{D}_k \right|$ data samples, in which the data samples belong to $c$-th class ($c \in \mathcal{C}$) is denoted as $\mathcal{D}_{k,c}$ with $D_{k,c} = \left| \mathcal{D}_{k,c} \right|$ data samples. Thus, $\mathcal{D}_k =  \cup \left\{ \mathcal{D}_{k,c} \right\}_{c = 1}^C$. Without loss of generality, we assume there is no overlapping between datasets from different devices, i.e., $\mathcal{D}_k \cap \mathcal{D}_h = \emptyset$, $\forall k, h \in \mathcal{K}$. Thus, the entire dataset, $\mathcal{D} =  \cup \left\{ \mathcal{D}_k \right\}_{k = 1}^K$, is with total number of samples $D = \sum\nolimits_{k = 1}^K D_k$. For ease of presentation, we use $\mathcal{D}_c$ to represent all data samples belonging to class $c$ in $\mathcal{D}$. That is, $\mathcal{D}_c = \cup \left\{ \mathcal{D}_{k,c} \right\}_{k = 1}^K$ with $D_c = \sum\nolimits_{k = 1}^K D_{k,c}$ data samples.
\vspace{-0.5cm}
\subsection{Knowledge-aided Loss Function for Local Training}
Let $\zeta=(\bm{x}, y)$ denote a data sample in $\mathcal{D}$, where $\bm{x} \in \mathbb{R}^d$ is the $d$-dimensional input feature vector, $y \in \mathbb{R}$ is the corresponding ground-truth label.
Let $\bm{z} \in \mathbb{R}^p$ be the latent feature vector. As shown in Fig. \ref{fig:local_training}, the machine learning model parameterized by $\bm{w}=[\bm{u},\bm{v}]$ consists of two components: a feature extractor $h: \bm{x} \rightarrow \bm{z}$ parameterized by $\bm{u}$, and a label predictor $g: \bm{z} \rightarrow \hat{y}$ parameterized by $\bm{v}$.
Before discussing the knowledge-aided loss function, we introduce two fundamental loss functions, i.e., empirical loss and knowledge loss.
The empirical loss supervises the local models' training to minimize the prediction error, while the knowledge loss achieves knowledge sharing among devices.
\begin{enumerate}[fullwidth,itemindent=1em,label=\arabic*)]
  \item \textbf{Empirical loss function for local model update}: Let $f(\bm{x},y;\bm{w})$ denote the sample-wise empirical loss function, which quantifies the error between the ground-truth label, $y$, and the predicted output, $\hat{y}$, based on model $\bm{w}$. Thus, the local empirical loss function at device $k$, which measures the model error on its local dataset $\mathcal{D}_k$, is defined as
      \begin{align}
        F_k(\bm{w}_k) = F_k(\bm{u}_k,\bm{v}_k) \buildrel \Delta \over = \frac{1}{D_k}\sum\nolimits_{(\bm{x},y) \in \mathcal{D}_k} f(\bm{x},y;\bm{w}_k),
      \end{align}
      where $\bm{w}_k$ denotes the machine learning model of device $k$; $\bm{u}_k$ and $\bm{v}_k$ correspond to its feature extractor and label predictor parts, respectively. For ease of presentation, we use $\bm{W} = (\bm{w}_1,\bm{w}_2, \cdots,\bm{w}_K)$ to denote all the devices' models throughout this paper. The global loss function associated with all distributed local datasets is given by
      \begin{align}\label{eq:glo_EM_loss}
        F(\bm{W}) = F(\bm{w}_1,\bm{w}_2, \cdots,\bm{w}_K) \buildrel \Delta \over = \frac{1}{D}\sum\nolimits_{k = 1}^K  D_k F_k(\bm{w}_k).
      \end{align}
      The federated learning process is done by solving the following problem:
        \begin{align}\label{eq:FL_opt_prob}
        \min\limits_{\bm{W}=(\bm{w}_1,\bm{w}_2, \cdots,\bm{w}_K)} F(\bm{W}).
        \end{align}
      To preserve the data privacy of devices, the devices collaboratively learn $\bm{W}$ without transmitting the raw training data.
      Note that the conventional FL algorithms, e.g., FedAvg \cite{pmlr-v54-mcmahan17a}, aim to find an optimal shared global model $\bm{w}^*=\bm{w}_1^*=\cdots=\bm{w}_K^*$ to minimize the global loss $F(\bm{W})$.
      However, this work aims to develop a personalized FL algorithm which trains personalized models for each device to solve the problem \eqref{eq:FL_opt_prob}, where different local models are used to fit user-specific data and capture the common knowledge distilled from data of other devices.

  \item \textbf{Knowledge loss function for local feature extractor update}: When devices are equipped with heterogeneous models, the conventional FL algorithms fail to coordinate devices to train models collaboratively. To tackle this issue, we introduce the knowledge loss function to regularize devices' feature extractors in the training process, achieving knowledge sharing between devices. It is worth mentioning that the knowledge of different devices and classes has the same dimensionality that equals the dimension of feature extractors' output, i.e., $p$.
    Let $\bm{\Omega}_{k,c}$ denote device $k$'s knowledge about data class $c$, which is defined as the average output of its feature extractor based on the data samples in $\mathcal{D}_{k,c}$, that is
    \begin{align}
    \bm{\Omega}_{k,c} = \frac{1}{D_{k,c}}\sum\nolimits_{(\bm{x},y) \in \mathcal{D}_{k,c}} h_k(\bm{x};\bm{u}_{k,t}),
    \end{align}
    where $h_k(\cdot)$ denote the feature extractor of device $k$.
    Let $\bm{\Omega}_{c}$ denote the global knowledge about class $c$ that aggregates all devices' knowledge of class $c$, i.e.,
    \begin{align}\label{eq:c_class_know}
    \bm{\Omega}_c = \frac{1}{D_c}\sum\nolimits_{k = 1}^K D_{k,c}\bm{\Omega}_{k,c}.
    \end{align}
    We use $\bm{\Omega} = (\bm{\Omega}_{1},\bm{\Omega}_{2}, \cdots ,\bm{\Omega}_{C} )$ to denote the aggregated global knowledge.
    For each data sample $(\bm{x}, y) \in \mathcal{D}_{k,c}$ ($\forall k \in \mathcal{K}, c \in \mathcal{C}$), we define the knowledge loss of device $k$'s feature extractor as $l_k(\bm{x};\bm{u}_k) = \frac{1}{2}{\left\| h_k(\bm{x};\bm{u}_k) - \bm{\Omega}_c \right\|^2}$, which quantifies the difference between the extracted feature of device $k$ on data sample $(\bm{x}, y)$ and the global feature of class $c$. Thus, the knowledge loss of device $k$ is
    \begin{align}\label{eq:knowledge_loss}
    L_k(\bm{u}_k) = \frac{1}{D_k}\sum\nolimits_{c = 1}^C  \sum\nolimits_{(\bm{x},y) \in \mathcal{D}_{k,c}} \frac{1}{2}{\left\| h_k(\bm{x};\bm{u}_{k,t}) - \bm{\Omega}_c \right\|^2},
    \end{align}

    which measures the difference between local knowledge and global knowledge. According to \eqref{eq:knowledge_loss}, devices only learn the knowledge of their local data types instead of all the data types. However, it fits devices' local models to their specific data and improves the learning performance on heterogeneous local data scenarios. In addition, devices can use global knowledge to regularize the local training process when new data classes are generated and rapidly adapt their local models to these new class data.
\end{enumerate}

In this work, we define a \textbf{knowledge-aided loss function} based on the empirical and knowledge loss functions, i.e., $F_k(\bm{u}_k, \bm{v}_k) + \lambda L_k(\bm{u}_k)$, to guide the feature extractor training for device $k$ ($\forall k \in \mathcal{K}$), where $\lambda$ is a hyperparameter to balance the empirical loss and knowledge loss for device $k$. For the label predictor, we still use the conventional empirical loss function.

\subsection{Knowledge-aided Federated Learning Mechanism}
The conventional FL approaches rely on aggregating devices' model/gradient parameters in each round, which induces remarkable communication overhead for wireless networks and requires all the local models to be of the same architecture.
To tackle these issues, we propose a novel KFL algorithm to enable collaborative training between heterogeneous local models. Specifically, devices upload their lightweight \emph{knowledge} to the server for aggregation in the per-round training process instead of the heavy model/gradient parameters. The learning process repeats the following steps until the devices' models converge, as shown in Fig. \ref{fig:sysmodel}.
\begin{enumerate}
  \item \textbf{Device selection}: The edge server selects a subset of devices from $\mathcal{K}$ to participate in the training process in the current round. Let $\alpha_{k,t}\in \{0,1\}$ denote the scheduling indicator of device $k$ in round $t$, where $\alpha_{k,t}=1$ indicates that device $k$ is scheduled in round $t$, $\alpha_{k,t}=0$ otherwise. Thus, the scheduled device set in round $t$ is $\bm{S}_t = \{k: \alpha_{k,t}=1, \forall k \in \mathcal{K}\}$.

  \item \textbf{Knowledge broadcast}: In each round $t$, the edge server broadcasts the latest global knowledge, i.e., $\bm{\Omega}_t = (\bm{\Omega}_{1,t},\bm{\Omega}_{2,t}, \cdots ,\bm{\Omega}_{C,t} )$, to all scheduled devices to regularize their local training process, where $\bm{\Omega}_{c,t}$ is the $c$-th class knowledge in round $t$ that is computed in \eqref{eq:c_class_know}.

  \item \textbf{Local training}: All scheduled devices update their local models after receiving the global knowledge, $\bm{\Omega}_t$, by performing $\tau$ steps gradient descent on its local dataset, as shown in Fig. \ref{fig:local_training}. For device $k$, its local feature extractor in $t$-th round is updated as
      \begin{align}\label{eq:feature_update}
      \bm{u}_{k,t,l+1} = \bm{u}_{k,t,l} - \eta_u \Big{(}\nabla_u F_k(\bm{u}_{k,t,l},\bm{v}_{k,t,l}) + \lambda \nabla L_k(\bm{u}_{k,t,l})\Big{)}, \forall l \in \{0, 1, \cdots, \tau-1\},
      \end{align}
      and its predictor is updated by
      \begin{align}\label{predictor_update}
      \bm{v}_{k,t,l + 1} = \bm{v}_{k,t,l} - \eta_v \nabla_v F_k(\bm{u}_{k,t,l},\bm{v}_{k,t,l}), \forall l \in \{0, 1, \cdots, \tau-1\},
      \end{align}
      where $\eta_u$ and $\eta_v$ are the learning rate of feature extractor and predictor, respectively, $\lambda$ is a hyperparameter to balance the empirical loss and knowledge loss for devices $k$.

  \item \textbf{Knowledge computing}: After finishing the local iterations, all scheduled devices compute their knowledge for each class $c$ ($c \in \mathcal{C}$) as $\bm{\Omega}_{k,c,t+1} = \frac{1}{D_{k,c}}\sum\nolimits_{(\bm{x},y) \in \mathcal{D}_{k,c}} h_k(\bm{u}_{k,t+1};\bm{x})$.
      The knowledge of device $k$ for all classes is denoted by $\bm{\Omega}_{k,t+1} = (\bm{\Omega}_{k,1,t+1},\bm{\Omega}_{k,2,t+1}, \cdots,\bm{\Omega}_{k,C,t+1})$.

  \item \textbf{Knowledge aggregation}: After finishing the local knowledge computing, all scheduled devices upload their knowledge to the edge server through wireless channels for aggregation. Specifically, the edge server computes the global shared knowledge of $c$-th class as
      \begin{align}\label{eq:know_agg}
      \bm{\Omega}_{c,t+1} = \frac{\sum\nolimits_{k \in \bm{S}_t} D_{k,c} \bm{\Omega}_{k,c,t+1}}{\sum\nolimits_{k \in \bm{S}_t} D_{k,c}}.
      \end{align}
      The aggregated global knowledge in round $(t+1)$ is $\bm{\Omega}_{t+1} = (\bm{\Omega}_{1,t+1},\bm{\Omega}_{2,t+1}, \cdots ,\bm{\Omega}_{C,t+1} )$.
\end{enumerate}

To better illustrate the proposed KFL, we summarize the detailed steps of its training process in Algorithm \ref{alg:KFL}.
It is worth mentioning that the proposed KFL requires devices to upload the knowledge to the edge server for aggregation instead of the entire local models. Devices' knowledge is generated by averaging the output of their local feature extractor on the data samples from the same class, and the process is irreversible \cite{8835269}. Thus, KFL is more beneficial for privacy preservation than the model aggregation-based FL algorithms exchanging local models between devices and the edge server. The reason is that the local models are updated according to the devices' private data, whose pattern is encoded into the model parameters. Therefore, if a corresponding decoder could be constructed, the private data or statistics would be recovered inversely \cite{ hitaj2017deep}.
\begin{algorithm}
\floatname{algorithm}{Algorithm}
\algsetup{linenosize=} \small
\caption{Knowledge-aided Federated Learning Algorithm}
\label{alg:KFL}
\begin{algorithmic}[1]
\STATE \textbf{Initialization:} $t=0$, training round $T$, and each device initials its local model $\bm{w}_{k,t}$;
\STATE \textbf{Server side:}
\FOR{$t=0, 1, \cdots, T-1$}
    \STATE Select a subset of devices ($\bm{S}_t$) and broadcasts the latest global knowledge, i.e., $\bm{\Omega}_t$, to them.
    \IF{Receive the knowledge from the selected devices}
        \STATE Aggregate the global knowledge according to \eqref{eq:know_agg}.
    \ENDIF
\ENDFOR

\STATE \textbf{Device side:}
\IF{Device $k$ is scheduled}
    \STATE Receive the global knowledge, $\bm{\Omega}_t$, from the edge server;
    \FOR{$l=0, 1, \cdots, \tau-1$}
        \STATE Update the local feature extractor, $\bm{u}_{k,t,l+1}$, based on \eqref{eq:feature_update};
        \STATE Update the local predictor, $\bm{v}_{k,t,l+1}$, based on \eqref{predictor_update};
    \ENDFOR
    \STATE Compute their knowledge for each class $c$ ($c \in \mathcal{C}$) as $\bm{\Omega}_{k,c,t+1} = \frac{1}{D_{k,c}}\sum\nolimits_{(\bm{x},y) \in \mathcal{D}_{k,c}} h_k(\bm{u}_{k,t+1};\bm{x})$.
    \STATE Upload the local knowledge $\bm{\Omega}_{k,t+1} = (\bm{\Omega}_{k,1,t+1},\bm{\Omega}_{k,2,t+1}, \cdots,\bm{\Omega}_{k,C,t+1})$ to the edge server.
\ENDIF
\end{algorithmic}
\end{algorithm}

\subsection{Knowledge-aided Federated Learning Cost Model}
In the following, we characterize the learning cost model in each KFL round, including computation cost and communication cost.
\begin{enumerate}[fullwidth,itemindent=1em,label=\arabic*)]
  \item \textbf{Computation Cost}: We consider the central processing unit (CPU) adopted to perform training on each device. Denote the CPU clock frequency of device $k$ by $f_k$ (cycles per second). The number of float-point operations (FLOPs) per cycle is represented by $n_k$.
      Let $C_k$ denote the required number of FLOPs to process one data sample at device $k$. Consequently, the local training latency of device $k$ is given by
      \begin{align}
      \mathcal{T}_k^{\rm{L}} = {\tau D_k C_k}/{(f_k n_k)}.
      \end{align}
      The corresponding energy consumption of device $k$ is
      \begin{align}
      E_{k}^{\rm{L}} = \kappa \tau D_k C_k f_k^2 /n_k,
      \end{align}
      where $\kappa$ is the power coefficient, depending on the chip architecture.
  \item \textbf{Communication Cost}: We consider that the frequency division multiple access is employed for devices to upload their knowledge. The total available wireless bandwidth is $B$Hz. Let $p_{k,t}$ denote the transmit power of device $k$, its maximum value is $p_{k,\max}$. The channel gain between device $k$ and the edge server is represented by $h_{k,t}$, which considers the path loss and Rayleigh fading. In addition, the channel remains unchangeable within one round but varies independently over rounds. Let $\theta_{k,t}\in [0,1]$ denote the proportion of the overall bandwidth allocated to device $k$ in round $t$, and $\bm{\theta}_t=(\theta_{1,t},\theta_{2,t}, \cdots, \theta_{K,t})$. The uplink rate of device $k$ can be described as $r_{k,t}= \theta_{k,t}B\log_2(1 + \frac{p_{k,t}h_{k,t}}{\theta_{k,t}B N_0})$, where $N_0$ is the power density of noise.
      Note that the proposed KFL requires that the knowledge of different devices and classes has the same dimensionality. Thus, the number of parameters in the knowledge of different devices is the same, denoted as $Q$.
      Each parameter is quantized by $q$ bits.
      Thus, the local knowledge uploading latency of device $k$ is
      \begin{align}
      \mathcal{T}_{k,t}^{\rm{U}} = \frac{Qq}{{{r_{k,t}}}} = \frac{Qq}{\theta_{k,t}B\log_2 \big{(} 1 + \frac{p_{k,t}h_{k,t}}{\theta_{k,t}BN_0} \big{)}}.
      \end{align}
      The corresponding energy consumption is
      \begin{align}\label{eq:trans_energycon}
      E_{k,t}^{\rm{U}} = p_{k,t} \mathcal{T}_{k,t}^{\rm{U}} = \frac{\theta_{k,t} B \mathcal{T}_{k,t}^{\rm{U}}N_0}{h_{k,t}}\Big{(} 2^{\frac{Qq}{\theta_{k,t} B \mathcal{T}_{k,t}^{\rm{U}}}} - 1 \Big{)}.
      \end{align}
\end{enumerate}

According to above modes, the energy consumption of device $k$ in round $t$ is $E_{k,t}= E_{k,t}^{\rm{L}} + E_{k,t}^{\rm{U}}$.
Note that we ignore the global knowledge broadcasting and aggregation latency in the above discussion because the broadcasting process occupies the entire bandwidth. The edge server has large transmit power, so the broadcasting latency is negligible. Moreover, the edge server is usually computationally powerful, and the global knowledge aggregation latency can be ignored compared to the above computation and communication latencies.

\subsection{Problem Formulation}
In this work, we aim to improve the learning performance by minimizing the global loss after $T$ rounds, i.e., $F(\bm{W}_T)$, under the energy budget constraint of devices, where $\bm{W}_T$ denote the local models in $T$-th round. Towards this end, we jointly optimize the device scheduling, bandwidth allocation, and power control policies. The optimization problem is given by
\begin{align}
\mathcal{P}:~~~~~~~&\min_{\left\{ \bm{S}_t,\bm{\theta}_t, \bm{p}_t \right\}_{t = 0}^{T-1}}~ F(\bm{W}_T)\label{prob:P}\\
\text{s.~t.~~}& \sum\nolimits_{t = 0}^{T - 1} {{E_{k,t}}}  \le {E_k}, \forall k \in \mathcal{K}, \label{cons:P_1}\tag{\theequation a}\\
&\mathcal{T}_{k,t}^{\rm{L}} + \mathcal{T}_{k,t}^{\rm{U}} \le \mathcal{T}_{\max}, \forall k \in \mathcal{K}, \forall t, \label{cons:P_2}\tag{\theequation b}\\
&\sum\nolimits_{k = 1}^K \theta_{k,t} \le 1, \forall t,\label{cons:P_3}\tag{\theequation c}\\
&0 \le \theta_{k,t} \le 1, \forall k \in \mathcal{K}, \forall t, \label{cons:P_4}\tag{\theequation d}\\
&\alpha_{k,t} \in \left\{ {0,1} \right\},\forall k \in \mathcal{K}, \forall t, \label{cons:P_5}\tag{\theequation e} \\
&0 \le p_k \le p_{k,\max}, \forall k \in \mathcal{K}.\label{cons:P_6}\tag{\theequation f}
\end{align}
In problem $\mathcal{P}$, \eqref{cons:P_1} imposes restrictions on the energy consumption of each device $k$ cannot exceed its budget $E_k$. \eqref{cons:P_2} stipulates that the completion time of each round cannot exceed its maximum allowable delay.
\eqref{cons:P_3} indicates that the wireless bandwidth allocated to all devices cannot exceed the total available bandwidth resource. \eqref{cons:P_4} restricts the wireless bandwidth resource allocated to each device. \eqref{cons:P_5} indicates which devices are scheduled in each round.

Solving problem $\mathcal{P}$ requires the explicit form about how device scheduling policy affects the final global loss function. Since it is almost impossible to find an exact analytical expression of $F(\bm{W}_T)$ with respect to $\bm{S}_t$ ($t \in \{0, 1, \cdots, T-1\}$), we turn to find an upper bound of $F(\bm{W}_T)$ and minimize it for the global loss minimization in Section \ref{subsec:conver_ana}.
Moreover, the optimal solution to problem $\mathcal{P}$ requires the system state information of all rounds at the beginning of training. However, such information is unavailable in the practical systems due to the unpredictable time-varying channel condition. To enable online device scheduling, the device scheduling decision should be made at the beginning of each round with only the current state. To this end, we transform the long-term decision problem into a deterministic one with the assistance of the Lyapunov optimization approach in Section \ref{subsec:prob_trans}.

\section{Convergence Analysis and Problem Formulation}\label{sec:convergence_ana_transform}
In this section, we theoretically analyze the convergence bound of the proposed KFL under a non-convex loss function setting.
The convergence bound reveals that the scheduled data volume in each round and different learning rounds significantly affect the learning performance. Motivated by this, we define a new metric, i.e., the weighted scheduled data volume, to guide the device scheduling design.
Then, we transfer the original problem to maximize this metric for minimizing the gap between the global loss function and the optimal loss.
To enable the online dynamic device scheduling under long-term energy budgets constraint, we further transform the problem into a deterministic problem in each round with the assistance of the Lyapunov optimization approach.

\subsection{Convergence Analysis}\label{subsec:conver_ana}
In this subsection, we investigate the convergence behavior of the proposed KFL algorithm. To facilitate the analysis, we make the following assumptions on each local loss function $F_k(\cdot)$.
\begin{ass}\label{assump:one}
All empirical loss functions $F_k(\bm{u}_k,\bm{v}_k)$ ($k \in \mathcal{K}$) are continuously differentiable with respect to $\bm{u}_k$ and $\bm{v}_k$, and there exist constants $L_u$, $L_v$, $L_{uv}$, and $L_{vu}$ such that for each $F_k(\bm{u}_k,\bm{v}_k)$:
\begin{itemize}
  \item $\nabla_{\bm{u}}F_k(\bm{u}_k,{\bm{v}_k})$ is  $L_u$-Lipschitz continuous with $\bm{u}_k$ and $L_{uv}$-Lipschitz continuous with $\bm{v}_k$, that is,
    \begin{align}
    \left\| \nabla_{\bm{u}}F_k(\bm{u}_k, {\bm{v}_k}) \!-\! \nabla_{\bm{u}}F_k(\bm{u}_k', \bm{v}_k) \right\| \le L_u \left\|\bm{u}_k \!-\! \bm{u}_k'\right\|,
    \end{align}
    and
    \begin{align}
    \left\| {\nabla_{\bm{u}}F_k(\bm{u}_k,{\bm{v}_k}) \!-\! \nabla_{\bm{u}}F_k(\bm{u}_k,\bm{v'}_k)} \right\| \!\le\! L_{uv} \left\| \bm{v}_k \!-\! \bm{v'}_k \right\|.
    \end{align}
  \item $\nabla_{\bm{v}}F_k(\bm{u}_k,{\bm{v}_k})$ is  $L_v$-Lipschitz continuous with $\bm{v}_k$ and $L_{vu}$-Lipschitz continuous with $\bm{u}$.
\end{itemize}
\end{ass}

\begin{ass}\label{assump:two}
The squared norm of gradients is uniformly bounded, i.e., $\| \nabla_{\bm{u}}F_k(\bm{u}_{k,t},\bm{v}_{k,t}) \|^2 \le G_1^2$ and $\| \nabla_{\bm{v}}F_k(\bm{u}_{k,t},\bm{v}_{k,t}) \|^2 \le G_2^2$.
\end{ass}

\begin{ass}\label{assump:three}
For each local feature extractor $h_k(\cdot)$ ($\forall k \in \mathcal{K}$), its gradient norm is bounded by $\vartheta^2$, i.e., $\| \nabla h_k(\bm{u}_k) \|^2 \le \vartheta^2$, and the squared norm of its output vector is bounded by $\| h_k(\bm{u}_k;x) \|^2 \le \varsigma^2$.
\end{ass}

Assumption \ref{assump:one} is satisfied by most deep NNs.
The modern NNs are usually composed of multiple layers. Based on \cite{abbasnejad2018deep}, a deep NN defined by a composition of functions is a Lipschitz NN if the functions in all layers are Lipschitz. It has been proved in \cite{abbasnejad2018deep, NEURIPS2018_d54e99a6} that the convolution layer, linear layer, and some nonlinear activation functions (e.g., Sigmoid and tanh) are Lipschitz functions. Thus, most deep NNs have Lipschitz continuous gradients.
For a Lipschitz NN in which all layers are Lipschitz functions, both the feature extractor and predictor composed of Lipschitz layers are Lipschitz functions.
Thus, Assumption \ref{assump:one} is satisfied by assuming the whole NN to be Lipschitz continuous. In addition, according to Proposition 1 in \cite{abbasnejad2018deep}, one can derive that $F_k(\bm{u}_k,{\bm{v}_k})$ is $(L_u \times L_v)$-smooth based on Assumption \ref{assump:one}.
Assumption \ref{assump:two} is widely used in the existing convergence analysis works, e.g., \cite{9425020,9598845,9409149,chen2022federated}. Assumption \ref{assump:three} is inherently satisfied by Assumption \ref{assump:two} since the gradient of a NN is a function of its output vector.
To begin with, we first derive a key lemma to assist our analysis as follows:
\vspace{-0.6cm}
\begin{lem}\label{lem:lip_F}
Let Assumption $\rm{\ref{assump:one}}$ holds, we have
\begin{multline}
F_k(\bm{u}_k',\bm{v}_k') - F_k(\bm{u}_k,\bm{v}_k)
\le \left\langle {\nabla_{\bm{u}}F_k(\bm{u}_k,\bm{v}_k),\bm{u}_k' - \bm{u}_k} \right\rangle  + \frac{1+\chi}{2} L_u \left\| \bm{u}_k' - \bm{u}_k \right\|^2 \\
+ \left\langle {\nabla_{\bm{v}}F_k(\bm{u}_k,\bm{v}_k),\bm{v}_k' - \bm{v}_k} \right\rangle  + \frac{1+\chi}{2} L_v \left\| \bm{v}_k' - \bm{v}_k \right\|^2,
\end{multline}
where $\chi = \max \left\{ L_{uv}, L_{vu} \right\}/\sqrt{L_u L_v}$, which measures the relative cross-sensitivity of $\nabla_{\bm{u}}F_k(\bm{u}_k,\bm{v}_k)$ with respect to $\bm{v}_k$ and $\nabla_{\bm{v}}F_k(\bm{u}_k,\bm{v}_k)$ with respect to $\bm{u}_k$.
\end{lem}
\begin{proof}
Please see Appendix \ref{app:one}.
\end{proof}

Lemma \ref{lem:lip_F} reveals the gradient relationships of a NN between its feature extractor and label predictor part. According to Lemma \ref{lem:lip_F}, we derive the one-round convergence bound of any device $k$ ($k\in\mathcal{K}$) in Lemma \ref{lem:oneR_conver}, in which devices utilize the proposed knowledge-aided loss to update their local models.
\vspace{-0.3cm}
\begin{lem}\label{lem:oneR_conver}
Let Assumption $\rm{\ref{assump:one}}$, $\rm{\ref{assump:two}}$, and $\rm{\ref{assump:three}}$ hold. The learning rates satisfy $\eta_u \le \frac{1}{4\tau(1 + \chi)L_u}$ and $\eta_v \le \frac{1}{2\tau(1 + \chi)L_v}$, the one-round convergence bound of device $k$ ($k\in \mathcal{K}$) is given by
\begin{multline}\label{eq:lemma_eq}
F_k(\bm{u}_{k,t + 1},\bm{v}_{k,t + 1}) - F_k(\bm{u}_{k,t},\bm{v}_{k,t})
\le \Big{(}2(1 + \chi)L_u \eta_u^2 \tau^2 - \frac{1}{2}\eta_u\tau \Big{)} \left\| \nabla_{\bm{u}}F_k(\bm{u}_{k,t},\bm{v}_{k,t}) \right\|^2 \\
+ \Big{(}(1 + \chi)L_v \eta_v^2{\tau^2} - \frac{1}{2}\eta_v\tau \Big{)}{\left\|\nabla_{\bm{v}}F_k(\bm{u}_{k,t},\bm{v}_{k,t}) \right\|^2}
+ A_1
+ \frac{5}{4}\eta_u \lambda^2\sum\nolimits_{l = 0}^{\tau - 1}\left\|\nabla L_k(\bm{u}_{k,t,l}) \right\|^2 \\
+ 2\eta_u^2 \lambda^2(3\eta_u L_u^2 + 2\eta_v \chi^2 L_u L_v)\sum\nolimits_{l = 0}^{\tau - 1} (\tau - l)\left\| \nabla L_k(\bm{u}_{k,t,l}) \right\|^2,
\end{multline}
where $A_1 = \tau (\tau + 1)(2\tau + 1)\left(\eta_u^3G_1^2L_u^2 + \frac{1}{3}\eta_v^3 G_2^2 L_v^2 + (\frac{2}{3}\eta_u G_1^2 + \frac{1}{2} \eta_v G_2^2)\eta_u \eta_v \chi^2 L_u L_v \right)$.
\end{lem}
\begin{proof}
Please see Appendix \ref{app:two}.
\end{proof}

Based on Lemma \ref{lem:oneR_conver}, we further analyze the convergence behaviour of the proposed KFL algorithm after $T$ rounds in Theorem \ref{thm:one}, which takes into account the knowledge aggregation between devices.
\begin{thm}\label{thm:one}
Let Assumption $\rm{\ref{assump:one}}$, $\rm{\ref{assump:two}}$, and $\rm{\ref{assump:three}}$ hold, $\eta_u \le \frac{1}{4\tau(1 + \chi)L_{\bm{u}}}$ and $\eta_v \le \frac{1}{2\tau(1 + \chi)L_v}$, the gap between the global loss function after $T$ rounds and the optimal loss is bounded by
\begin{multline}\label{eq:conv_bound}
F(\bm{W}_T) - F(\bm{W}^*) \le {A_3^T}( F(\bm{W}_0) - F(\bm{W}^*) )\\
+ \frac{1 - A_3^T}{1 - A_3}(A_1 + A_2)
+{A_2}\frac{CK}{D}\sum\nolimits_{t = 0}^{T - 2} A_3^{T - 2 - t}\sum\nolimits_{k = 1}^K  \sum\nolimits_{c = 1}^C \frac{D_{k,c}^2}{D_c^2 D_k}\sum\nolimits_{k = 1}^K D_{k,c}^2 \\
- {A_2}\frac{1}{DK(T - 1)}\frac{1}{\max_{1\le k \le K} D_k}{\left(\sum\nolimits_{t = 0}^{T - 2} A_3^{T - 2 - t}\sum\nolimits_{k = 1}^K \alpha_{k,t} D_k \right)^2}
\end{multline}
where $A_2 = 10\eta_u\lambda^2\tau \vartheta^2\varsigma^2 + 8\eta_u^2\lambda^2 \vartheta^2\varsigma^2\left( 3\eta_u L_u^2 + 2\eta_v\chi^2L_uL_v \right)\tau (\tau + 1)$, $A_3 = 1 + (4L_u^2\eta_u^2 + 2L_v^2\eta_v^2)(1 + \chi)\tau^2 - (\eta_u L_u+\eta_v L_v)\tau $.
\end{thm}
\begin{proof}
Please see Appendix \ref{app:three}.
\end{proof}

Theorem \ref{thm:one} reveals how the device scheduling policy affects the convergence bound of KFL without characterizing the impact of non-IID degrees on the convergence bound.
In general, the non-IID degree is characterized by the difference between the optimal global loss and the weighted summation of optimal local losses \cite{lion2020}. However, the proposed KFL is a personalized FL algorithm which trains a personalized model for each device. Thus, one cannot characterize the impacts of non-IID degree on the convergence bound in this way due to $F(\bm{W}^*)- \frac{1}{D}\sum_{k=1}^K D_k F_k(\bm{w}_k^*)=0$. However, how to characterize non-IID degrees' effects on the convergence bound of personalized FL algorithms is a promising research direction, which will be studied in our future works.

According to Theorem \ref{thm:one}, the gap between the global loss after $T$ rounds and the optimal loss is bounded by four terms, 1) the gap in the initial round, 2) two terms related to hyperparameters of the learning system, 3) the scheduled data volume in all rounds.
It is noted that $A_3\le 1$ due to $\eta_u \le \frac{1}{4\tau(1 + \chi)L_u}$ and $\eta_v \le \frac{1}{2\tau(1 + \chi)L_v}$. As $T$ increases, $A_3^T$ approaches to 0. Hence, the first term converges to 0, and the second and the third terms converge to a constant.
The first three terms decided by the system hyperparameters and initial models of devices are not related to the device scheduling policies.
The last term is an explicit form related to device scheduling. For the last term, we have the following remark:
\begin{rem}\label{rem:one}
Increasing the scheduled data samples in each round is able to narrow the gap between global loss and optimal loss.
In addition, as $t$ increases, $A_3^{T-1-t}$ also increases due to $A_3 < 1$. This indicates that more devices should be scheduled in early rounds when the total number of scheduled devices in the learning process is fixed.
\end{rem}

Note that, it has been experimentally observed in \cite{9237168} that scheduling more devices in the later rounds is beneficial for the learning performance of the federated averaging algorithm.
However, the proposed KFL that only aggregates devices' knowledge in each round achieves better learning performance when scheduling more devices in the earlier rounds, which is verified by the theoretical analysis in Remark \ref{rem:one} and experimental results in Section \ref{sec:simulation}.

\subsection{Problem Transformation via Lyapunov Optimization Framework}\label{subsec:prob_trans}
According to Theorem \ref{thm:one}, the gap between the global loss and the optimal loss can be narrowed by minimizing the last term on the right-hand-side (RHS) of \eqref{eq:conv_bound}. However, it is tractable to directly minimize this term since it involves some unknown parameters, e.g., the Lipschitz constant $L_u$ and $L_v$. Based on \cite{NEURIPS2018_d54e99a6}, the exact computation of the Lipschitz constant of deep learning architectures is intractable, even for two-layer NNs.
Inspired by Remark \ref{rem:one}, to enable tractable device scheduling design, we introduce a variable $\gamma_t$ ($t=0,1,\cdots, T-1$) as the weight of scheduled data samples in round $t$ to capture the varying significance of scheduling devices in different rounds. Based on this, we define the weighted scheduled data volume as $\sum\nolimits_{t = 0}^{T - 1} \gamma_t\sum\nolimits_{k = 1}^K \alpha_{k,t} D_k$ and maximize it for the global loss minimization.
Thus, we transform problem $\mathcal{P}$ as the following problem:
\begin{align}
\widehat{\mathcal{P}}:~~~~~~~&\max_{\left\{ \bm{S}_t,\bm{\theta}_t, \bm{p}_t \right\}_{t = 0}^{T-1}}~ \sum\nolimits_{t = 0}^{T - 1} \gamma_t\sum\nolimits_{k = 1}^K \alpha_{k,t} D_k\\
\text{s.~t.~~}& \eqref{cons:P_1}, \eqref{cons:P_2}, \eqref{cons:P_3}, \eqref{cons:P_4}, \eqref{cons:P_5}. \notag
\end{align}
Problem $\widehat{\mathcal{P}}$ involves multi-dimension discrete and continuous variables is a typical mixed-integer programming problem, which is generally NP-Hard.
In addition, solving the optimal solution of problem $\widehat{\mathcal{P}}$ offline requires optimally dividing the energy of all devices in each round due to the long-term energy constraints, which is intractable.
The most critical challenge of directly solving problem $\widehat{\mathcal{P}}$ is that it requires channel information of all devices over all rounds at the beginning of the FL process,  which may unfeasible in practical systems. To enable the online dynamic device scheduling, we utilize the Lyapunov optimization framework to deal with the correlations among rounds.
To this end, we construct a virtual queue $q_k(t)$ for each device $k$ ($k\in\mathcal{K}$), which evolves as

\begin{align}\label{eq:lyapunov_queue}
q_k(t + 1) = \max \Big\{ q_k(t) + \alpha_{k,t}E_{k,t} - \frac{E_k}{T},0 \Big\},
\end{align}
with the initial value $q_k(t)=0$ for all devices. Inspired by the drift-plus-penalty algorithm in \cite{neely2010stochastic}, we transform problem $\widehat{\mathcal{P}}$ as the following problem to enable online device scheduling
\begin{align}
\widetilde{\mathcal{P}}:~~~~~~~&\min_{\left\{ \bm{S}_t,\bm{\theta}_t, \bm{p}_t \right\}_{t = 0}^{T-1}}~ -V\gamma_t\sum\nolimits_{k = 1}^K \alpha_{k,t}D_k + \sum\nolimits_{k = 1}^K q_k(t)\alpha_{k,t}E_{k,t} \label{eq:lya_obj}\\
\text{s.~t.~~}& \eqref{cons:P_2}, \eqref{cons:P_3}, \eqref{cons:P_4}, \eqref{cons:P_5}. \notag
\end{align}
In problem $\widetilde{\mathcal{P}}$, $V \ge 0$ is a weight factor that balances the energy consumption of devices and learning performance. A large $V$ emphasises the learning performance improvement by sacrificing the devices' energy and vice versa.
In addition, from the objective function \eqref{eq:lya_obj}, the unscheduled devices in the former rounds have smaller $q_k(t)$. These devices are encouraged to participate in the current round of training for minimizing \eqref{eq:lya_obj}. Thus, problem $\widetilde{\mathcal{P}}$ contributes to a fair scheduling scheme between devices.

\section{Online Device Scheduling and Wireless Resource Allocation}\label{sec:alg_design}
In this section, we propose an energy-aware device scheduling, bandwidth allocation, and power control algorithm that solves problem $\widetilde{\mathcal{P}}$ in an online fashion. We first derive the optimal bandwidth allocation and power control policies using convex optimization techniques. Then, we propose a polynomial-time algorithm to solve the device scheduling decision with a $\mathcal{O}(\sqrt{V}, 1/V)$ energy-learning trade-off guarantee, where $V$ is an algorithm-related parameter.

\subsection{Optimal Power Control and Bandwidth Allocation}
For any given scheduled device set $\bm{S}_t \in \mathcal{K}$, we decompose the bandwidth allocation and power control problem from $\widetilde{\mathcal{P}}$ as follows:
\begin{align}
\mathcal{P}_1:~~~~~~~&\min_{\left\{\bm{\theta}_t, \bm{p}_t \right\}}~ \sum\nolimits_{k \in \bm{S}_t} q_k(t)E_{k,t} \label{eq:obj_power_band}\\
\text{s.~t.~~}& \eqref{cons:P_2}, \eqref{cons:P_3}, \eqref{cons:P_4}. \notag
\end{align}
For problem $\mathcal{P}_1$, we have the following proposition:
\begin{prop}\label{prop:one}
The optimal solution of problem $\mathcal{P}_1$ satisfies $\mathcal{T}_{k,t}^{\rm{U}} = \mathcal{T}_{\max}-\mathcal{T}_k^{\rm{L}}$, and the optimal transmit power of device $k$ satisfies
\begin{align}\label{eq:power_expression}
p_{k,t} = \frac{\theta_{k,t}B N_0}{h_{k,t}}\Big{(} 2^{\frac{Qq}{(\mathcal{T}_{\max} - \mathcal{T}_k^{\rm{L}})\theta_{k,t}B}} - 1 \Big{)}.
\end{align}

Sketch of proof: By proving that the first-order derivatives of the objective function \eqref{eq:obj_power_band} are great than 0, \eqref{eq:obj_power_band} is an non-increasing function
with respect to the communication time $\mathcal{T}_{k,t}^{\rm{U}}$. Thus, the optimal completion time of device $k$ is $\mathcal{T}_{k,t}^{\rm{U}} = \mathcal{T}_{\max}-\mathcal{T}_k^{\rm{L}}$. Based on \eqref{eq:trans_energycon}, the proposition is proved. We have the detailed proof in Section I of the technical report \cite{chen2022proof}.
\end{prop}

According to Proposition \ref{prop:one}, we substitute \eqref{eq:power_expression} into problem $\mathcal{P}_1$, the optimal bandwidth allocation problem can be formulated as
\begin{align}
\mathcal{P}_2:~~~~~~~&\min_{\bm{\theta}_t}~ \sum\nolimits_{k \in \bm{S}_t} \frac{\theta_{k,t}B  N_0q_k(t)(\mathcal{T}_{\max} - \mathcal{T}_k^{\rm{L}})}{h_{k,t}} \mathcal{I}(\theta_{k,t}) \label{eq:band_obj} \\[-0.2cm]
\text{s.~t.~~}& \eqref{cons:P_3}, \eqref{cons:P_4}, \notag \\[-0.2cm]
& \theta_{k,t}B\log \Big{(}1 + \frac{p_{k,\max}h_{k,t}}{\theta_{k,t}BN_0} \Big{)} \ge \frac{Qq}{(\mathcal{T}_{\max} - \mathcal{T}_k^{\rm{L}})}, \label{eq:minimal_band_cons}\tag{\theequation a}
\end{align}
where
\begin{align}
\mathcal{I}(\theta_{k,t}) = \exp \Big{(} \frac{Qq\ln 2}{(\mathcal{T}_{\max} - \mathcal{T}_k^{\rm{L}}){\theta_{k,t}}B} \Big{)} - 1.
\end{align}
For problem $\mathcal{P}_2$, we obtain its optimal solution by using the following lemma.
\begin{lem}\label{lem:band}
The optimal bandwidth allocation of problem $\mathcal{P}_2$ satisfies
\begin{align}\label{eq:optimal_theta_max}
\theta_{k,t}^* = \max \left\{ \theta_{k,t}(\mu),\theta_{k,t}^{\min} \right\},
\end{align}
where
\begin{align}\label{eq:optimal_theta}
\theta_{k,t}(\mu) = \frac{Qq\ln 2}{(\mathcal{T}_{\max} - \mathcal{T}_k^{\rm{L}})B\left(\mathcal{W}\left( \frac{\mu h_{k,t}}{eBN_0q_k(t)(\mathcal{T}_{\max} - \mathcal{T}_k^{\rm{L}})} - \frac{1}{e} \right) + 1 \right)},
\end{align}
and $\theta_{k,t}^{\min}$ satisfies constraint \eqref{cons:P_3}, $\mu$ is the Lagrange multiplier which satisfies $\sum\nolimits_{k = 1}^K \theta_{k,t}(\mu^*) = 1$. $\mathcal{W}(\cdot)$ is the principal branch of the Lambert function, defined as the solution of $\mathcal{W}(x) e^{\mathcal{W}(x)}=x$, in which $e$ is the Euler's number.

Sketch of proof: The problem $\mathcal{P}_2$ is a convex problem. By solving the KKT conditions of $\mathcal{P}_2$, the lemma is proved. We have the detailed proof in Section II of the technical report \cite{chen2022proof}.
\end{lem}

Although Lemma \ref{lem:band} provides the optimal condition of bandwidth allocation, there is still an unknown variable $\mu$. Below we develop a binary search method to solve the optimal $\mu$.
Since the Lagrange multiplier $\mu \ge 0$, we have $\frac{\mu h_{k,t}}{ e B N_0 q_k(t)\left(\mathcal{T}_{\max} - \mathcal{T}_k^{\rm{L}} \right)} - \frac{1}{e} \ge  - \frac{1}{e}$.
Moreover, $\mathcal{W}(x)$ is a monotonically increasing function when $x \ge -\frac{1}{e}$. Thus, $\theta_{k,t}(\mu)$ is a monotonically decreasing function with respect to $\mu$. To deploy the binary search method, we derive the lower and upper bound of $\mu$. Since $\mu \ge 0$, the lower bound of $\mu$ is $\mu_{\rm{lb}}=0$. For the upper bound, we have $\mathop {\max}\nolimits_{\bm{S}_t} \left\{ \theta _{k,t}(\mu) \right\} \ge \frac{1}{\left| \bm{S}_t \right|}$.
Let $\varphi_k = \frac{Qq\left| \bm{S}_t \right|\ln 2}{\left(\mathcal{T}_{\max} - \mathcal{T}_k^{\rm{L}}\right)B}$. Based on the definition of Lambert function, we have
\begin{align}\label{eq:lagrange_up}
\mu \le \mu_{\rm{ub}}  = \mathop {\max }\limits_{k \in \bm{S}_t} \left\{ \frac{B N_0 q_k(t)(\mathcal{T}_{\max} - \mathcal{T}_k^{\rm{L}})\left((\varphi_k - 1)e^{\varphi_k} + 1 \right)}{h_{k,t}} \right\}.
\end{align}
Based on the lower bound $\mu_{\rm{lb}}$ and upper bound $\mu_{\rm{ub}}$, the optimal Lagrange multiplier can be obtained by the binary search method.
For clarity, we summarize the detailed steps for solving the optimal bandwidth allocation policy in Algorithm \ref{alg:band_allo}.
The binary search method halves the search region at every iteration and terminate when the given precision (i.e., $\varepsilon$) requirement is satisfied. Thus, the time complexity of this method is $\mathcal{O}\left( \log_2\frac{\mu_{\rm{ub}} - \mu_{\rm{lb}}} \varepsilon \right)$.
\begin{algorithm}
\floatname{algorithm}{Algorithm}
\algsetup{linenosize=} \small
\caption{Optimal Wireless Bandwidth Allocation}
\label{alg:band_allo}
\begin{algorithmic}[1]
\STATE Initialize $\bm{S}_t$, the precision requirement $\varepsilon >0$.
\STATE Initialize the upper bound of Lagrange multiplier $\mu_{\rm{ub}}$ based on \eqref{eq:lagrange_up}, set the lower bound to $\mu_{\rm{lb}}=0$.

\REPEAT
    \STATE Set $\mu = (\mu_{\rm{lb}}+ \mu_{\rm{ub}})/2$.
    \STATE For each device $k \in \bm{S}_t$, compute the required bandwidth allocation ratio $\theta_{k,t}(\mu)$ based on \eqref{eq:optimal_theta}.
    \STATE Compute the summation of required bandwidth allocation ratio $\sum_{k\in \bm{S}_t} \theta_{k,t}(\mu)$.
    \IF{$\sum_{k\in \bm{S}_t} \theta_{k,t}(\mu)>1$}
        \STATE Halve the searching region by setting $\mu_{\rm{lb}}=\mu$ and $\mu_{\rm{ub}}= \mu_{\rm{ub}}$.
    \ELSIF{$0<\sum_{k\in \bm{S}_t} \theta_{k,t}(\mu)<1-\varepsilon$}
        \STATE Halve the searching region by setting $\mu_{\rm{lb}}=\mu_{\rm{lb}}$ and $\mu_{\rm{ub}}= \mu$.
    \ELSE
        \STATE Break the circulation.
    \ENDIF
\UNTIL{$|\mu_{\rm{ub}} - \mu_{\rm{lb}} | < \varepsilon $}
\STATE Substituting $\mu$ into \eqref{eq:optimal_theta} for get $\theta_{k,t}(\mu)$, then compute the optimal bandwidth allocation policy based on \eqref{eq:optimal_theta_max}.
\STATE Substitute the optimal bandwidth allocation policy into \eqref{eq:power_expression} for obtaining the optimal power control policy.
\RETURN The optimal bandwith allocation policy $\bm{\theta}_t$, the optimal power control policy $\bm{p}_t$.
\end{algorithmic}
\end{algorithm}

\subsection{Device Scheduling}
Based on the above analysis, the optimal bandwidth allocation and power control policy for any device scheduling set $\bm{S}_t$ can be obtained by using Algorithm \ref{alg:band_allo}. For device scheduling design, an intuitive method is to compute the objective function value for all possible device scheduling decisions, and select the one with minimal objective function as the final scheduling decision.
However, this intuitive method is infeasible in its implementation since there are $\sum\nolimits_{n = 0}^K {C_K^n}  = 2^K$ possible scheduling decisions, inducing an exponential time complexity with $\mathcal{O}\left(2^K  \log_2\frac{\mu_{\rm{ub}} - \mu_{\rm{lb}}} \varepsilon \right)$. In the following part, we develop an efficient algorithm to solve the device scheduling policy.

According to the objective function \eqref{eq:lya_obj}, it is desirable to select devices with small $q_k(t)$ and $E_{k,t}$, as well as large data samples. The small $E_{k,t}$ is achieved by strong channels and low computation energy consumption.
To identify these devices, we first allocate equal bandwidth to all devices (i.e., $\theta=1/K$), and then substitute $\mathcal{T}_{k,t}^{\rm{U}}=\mathcal{T}_{\max}-\mathcal{T}_k^{\rm{L}}$ into \eqref{eq:trans_energycon} to compute the estimated energy consumption $\bar E_{k,t} = E_{k,t}^{\rm{L}} + E_{k,t}^{\rm{U}}$.
Based on the estimated energy consumption for all devices, we sort devices based on $\Delta_{k,t} = - V\gamma_t D_{k,c} + q_k(t)\bar{E}_{k,t}$ ($\forall k \in \mathcal{K}$) in the ascending order. Denote $\widetilde{\mathcal{K}}$ as the sorted device set.
Many sorting algorithms, such as Heapsort or Mergesort, can be used, with a worst-case complexity $\mathcal{O}(K\log K)$.
Then, we solve the device scheduling policy by incrementally adding devices into the selection set $\bm{S}$ from the sorted device set $\widetilde{\mathcal{K}}$.
For each possible device scheduling set $\bm{S}$, we perform Algorithm \ref{alg:band_allo} to obtain the optimal wireless bandwidth allocation $\bm{\theta}_t^*(\bm{S})$, power control decisions $\bm{p}_t^*(\bm{S})$, as well as the optimal energy consumption $E_t^*(\bm{S})$.
Substituting $E_t^*(\bm{S})$ into \eqref{eq:lya_obj}, the drift-plus-penalty value of device scheduling set $\bm{S}$ can be obtained, denoted as $\mathcal{Y}(\bm{S})$. Let $\mathcal{H}$ denote the set of all possible device scheduling set $\bm{S}$. Finally, we obtain the optimal device scheduling policy through comparing the drift-plus-penalty value of all possible device scheduling set $\bm{S} \in \mathcal{H}$, i.e., $\bm{S}_t^* = \mathop {\arg \min }\nolimits_{\bm{S} \in \mathcal{H}} \mathcal{Y}(\bm{S})$.
Note that, the energy consumption of devices with $q_k(t)=0$ does not affect the objective function value, the minimal required bandwidth should be allocated to them for saving more bandwidth for other users with $q_k(t)>0$.
For clarity, we summarize the detail steps of device scheduling algorithm in Algorithm \ref{alg:two}, which obtains the device scheduling policy by performing at most $K$ times Algorithm \ref{alg:band_allo} and has a polynomial time complexity $\mathcal{O}\left( K\log_2\frac{\mu_{\rm{ub}} - \mu _{\rm{lb}}}{\varepsilon }\right)$.

For Algorithm \ref{alg:two}, we analyze its performance by comparing it with its optimal offline counterpart which is in fact problem $\widehat{\mathcal{P}}$. The offline algorithm has the channel information of all rounds. Let $\alpha_{k,t}^*$ be the offline optimal device scheduling decision obtained by solving the problem $\widehat{\mathcal{P}}$ with pre-known device information. The performance guarantee of the proposed device scheduling algorithm is shown in Proposition \ref{prop:ly_performance}.

\begin{algorithm}
\floatname{algorithm}{Algorithm}
\small
\caption{Energy-aware online Device scheduling}
\label{alg:two}
\begin{spacing}{1}
\begin{algorithmic}[1]
\STATE Input the virtual queue length $q_k(t)$ ($k \in \mathcal{K}$) and $\gamma_t$, initialize $V$.
\STATE Substituting $\theta_{k,t}=1/K$ and $\mathcal{T}_{k,t}^{\rm{U}}=\mathcal{T}_{\max}-\mathcal{T}_k^{\rm{L}}$ into \eqref{eq:trans_energycon} to compute the estimated energy consumption of device $k$ ($\forall k \in \mathcal{K}$), i.e., $\bar E_{k,t} = E_{k,t}^{\rm{L}} + E_{k,t}^{\rm{U}}$.
\STATE Sort devices based on $\Delta_{k,t} \!=\! - V \gamma_t D_{k,c} \!+\! q_k(t) \bar{E}_{k,t}$ in the ascending order to obtain the sorted device set $\widetilde{\mathcal{K}}$.
\FOR{$k = \widetilde{\mathcal{K}}(1), \widetilde{\mathcal{K}}(2), \cdots, \widetilde{\mathcal{K}}(K) $}
    \STATE Update $\bm{S}=\bm{S} \cup \{k\}$
    \STATE Solve the optimal bandwidth allocation and power control policy by Algorithm \ref{alg:band_allo}, i.e., $\bm{\theta}_t^*(\bm{S})$ and $\bm{p}_t^*(\bm{S})$.
    \STATE Compute the drift-plus-penalty of $\bm{S}$, i.e., $\mathcal{Y}(\bm{S})= -V\gamma_t\sum\nolimits_{k \in \bm{S}} D_k + \sum\nolimits_{k \in \bm{S}} q_k(t)E_{k,t} $
    \IF{$-V D_k + q_k(t) E_{k,t}>0$}
        \STATE Break the circulation
    \ELSE
        \STATE Add $\bm{S}$ into $\mathcal{H}$, i.e., $\mathcal{H}=\mathcal{H} \cup \bm{S}$
    \ENDIF
\ENDFOR
\STATE Find the optimal device scheduling set $\bm{S}_t^* = \argmin_{\bm{S} \in \mathcal{H}}\mathcal{Y}(\bm{S})$.
\RETURN The device scheduling set $\bm{S}_t^*$, wireless bandwidth allocation $\bm{\theta}_t^*(\bm{S}_t^*)$, and power control policy $\bm{p}_t^*(\bm{S}_t^*)$.
\end{algorithmic}
\end{spacing}
\end{algorithm}

\begin{prop}\label{prop:ly_performance}
Compared to the offline optimal solution, the cumulative loss of Algorithm \ref{alg:two} is bounded by
\begin{align}\label{eq:datavolume_bound}
\sum\nolimits_{t = 0}^{T - 1}  {\gamma_t}\sum\nolimits_{k = 1}^K \alpha_{k,t} {D_{k,c}} \ge  - \frac{T\zeta_0}{V} - \frac{T(T - 1)}{2V}\sum\nolimits_{k = 1}^K \zeta_k^2 + \sum\nolimits_{t = 0}^{T - 1}  {\gamma_t}\sum\nolimits_{k = 1}^K \alpha_{k,t}^*{D_{k,c}},
\end{align}
and the total energy consumption of Algorithm \ref{alg:two} is bounded by
\begin{align}\label{eq:energy_bound}
\sum\nolimits_{t = 0}^{T - 1} \sum\nolimits_{k = 1}^K \alpha_{k,t}E_{k,t} \le \sum\nolimits_{k = 1}^K  E_k + \sqrt{2K\left( T \zeta_0 + V\sum\nolimits_{t = 0}^{T - 1} \gamma_t D \right)},
\end{align}
where $\zeta_0 = \frac{1}{2}\sum\nolimits_{k = 1}^K  \zeta_k^2$ and $\zeta_k = {\max_t}\left\{ {\left| \alpha_{k,t}E_{k,t} - \frac{E_k}{T} \right|} \right\}$.
\end{prop}
\begin{proof}
Please see Appendix \ref{app:five}
\end{proof}

Proposition \ref{prop:ly_performance} characterizes the performance of the proposed device scheduling algorithm, which shows that 1) the energy constraints of devices are approximately satisfied with the $\mathcal{O}(\sqrt V)$-bounded factor, and 2) the proposed device algorithm is $\mathcal{O}(1/V)$-optimal with respect to the performance of its optimal offline counterpart solution.
Thus, the proposed device scheduling algorithm demonstrates an $\mathcal{O}(\sqrt{V}, 1/V)$ energy-learning trade-off. The worst-case performance of Algorithm \ref{alg:two} can be improved by reducing the upper bound of energy usage bias $\zeta_0$.
In addition, adjusting the weight parameter $V$ is able to achieve the balance between the learning performance and energy consumption of devices.
Specifically, with larger $V$, more emphasis is put on the scheduled data samples to improve the learning performance while more energy is consumed at devices, and vice versa.
In practical systems, one should carefully select the value of $V$ to optimize the learning performance with energy limits and use the energy in a balanced manner to avoid large $\zeta_k$.

\section{Numerical Results}\label{sec:simulation}
In this section, we verify the effectiveness of the proposed KFL algorithm.
If not specified, we consider $K = 100$ devices randomly distributed in a cell with a radius of 500m, and the server is located at the centre of the cell.
The total bandwidth is set to $B = 5$MHz. Similar to \cite{9862981}, the channel gain is modelled as $h_{k,t} = h_0 \rho_k(t) (d_0/d_k)^v$, where $h_0=-30$dBm is the path loss constant; $d_k$ is the distance between device $k$ and the edge server; $d_0=1$m is the reference distance; $\rho_k(t) \sim \text{Exp}(1)$ is exponentially distributed with unit mean, which represents the small-scale fading channel power gain from the device $k$ to the edge server in round $t$;  $d_0/d_k$ represents the large-scale path loss with $v=2$ being the path loss exponent.
The channel noise power spectral density $N_0$ is set to $-174$ dBm.
For all devices in the system, we set their maximal transmit power to $p_{k,\max}=30$dBm, and their CPU frequency are randomly selected from \{0.85, 1.12, 1.2, 1.3\}GHz \cite{9562538, 9509427}.

We evaluate the proposed KFL algorithm on two image classification tasks using MNIST and CIFAR-10 datasets, both of them have 10 classes of data samples.
For the MNIST dataset, we train five-layer multilayer perceptron (MLP) models with the following architecture: four fully connected layers with 784, 512, $d_{\text{MLP}}$, 64 units, each of these layers is activated by the ReLU function; and a 10-unit softmax output layer.
For the CIFAR-10 dataset, we train five-layer CNN models with the following structure: two $5\times5$ convolution layers followed by a $2\times2$ max-pooling layer, in which the first convolution layer possesses 6 channels and the second layer with 16 channels; three fully connected layers with 400, $d_{\text{CNN}}$, and 64 units, respectively; and a 10-unit softmax output layer. The ReLU function activates each convolution or fully connected layer.
When devices are equipped with homogeneous models, we set $d_{\text{MLP}}=256$ and $d_{\text{CNN}}=128$.
The number of FLOPs and parameters of these machine learning models can be estimated using the method in \cite{goodfellow2016deep}. Specifically, the MLP with $d_{\text{MLP}}=256$ possesses 553406 parameters which equal the FLOPs required to process one data sample. The CNN with $d_{\text{CNN}}=128$ has 63106 parameters and requires 1245834 FLOPs to process one data sample.
When devices have heterogeneous models, the value of $d_{\text{MLP}}$ and $d_{\text{CNN}}$ for all devices are randomly selected from $ \{128,192,256,320,384\}$.
For both MLP and CNN, the learning rates, i.e., $\eta_u$ and $\eta_v$, are set to 0.05, a momentum of 0.9 is adopted, and the number of local iterations is set to $\tau=5$, and cross-entropy is adopted as the loss function.
In addition, we first classify the training data samples according to their labels, then split each class of data samples into $mK/10$ shards, and finally randomly distribute
two shards of data samples to each device. That is, each client has a data distribution corresponding to at most $m$ classes.
The data distributions among devices are more skewed for smaller $m$. Due to page limits, we only show results based on $m=2$ in the following experiments. The results based on $m=3$ are presented in Section III of our technical report \cite{chen2022proof}, showing similar results with $m=2$.
For the MNIST dataset, we set the energy budgets $E_k=0.1 \times T$ J ($\forall k \in \mathcal{K}$) and $T_{\max}=1$s. For CIFAR-10, we set $E_k=0.5 \times T$ J ($\forall k \in \mathcal{K}$) and $T_{\max}=2$s.
In the simulations, each device first computes the number of correct predicted data samples on its test dataset by its local model.
Note that the deployed trained models on devices are the same in FedAvg, while each device has a personalized local model in the proposed KFL.
Then, the test accuracy is computed as the total number of correct predicted test data samples on all devices divided by the total number of test data samples on all devices.

In the following sections, we verify the theoretical results in Remark \ref{rem:one} on MNIST and CIFAR-10 datasets by comparing the learning performance of the following three temporal device scheduling patterns.
1) Uniform Scheduling: Ten devices are randomly scheduled in each round to participate in the learning process.
2) Ascend Scheduling: The number of scheduled devices increases from 1 to 20, with an average number of 10 devices scheduled in each round.
3) Descend Scheduling: The number of scheduled devices decreases from 20 to 1, with an average number of 10 devices per round.

\subsection{Performance Evaluation with Homogeneous Models}
We evaluate the performance of the proposed KFL algorithm by comparing it with the following benchmarks. Note that, devices are equipped with homogeneous local models, and we do not consider the energy and bandwidth limitation in this subsection.
1) FedAvg \cite{pmlr-v54-mcmahan17a}: In each round, the scheduled devices upload their model parameters to the edge server for aggregation.
2) FedRep \cite{pmlr-v139-collins21a}: The scheduled devices sequentially train the feature extractor and label predictor parts of their models in each round. Particularly, the scheduled devices only upload feature extractor part of their models to the edge server for aggregation.
3) APFL \cite{deng2020adaptive}: In each round, each scheduled device trains its own local model and the received global model from the edge server. Then APFL incorporates the devices' locally trained model and the updated global model to achieve a device-specific model.
It is worth mentioning that the proposed KFL algorithm requires fewer parameters transmission in each round than the benchmarks and thus reduces the communication cost.
Specifically, for the MNIST dataset, the proposed algorithm only requires devices to upload the knowledge of 10 classes, including $64\times 10=640$ parameters, accounting for 0.12\% of the transmitted parameters by FedAvg or by FedRep. For the CIFAR-10 dataset, the KFL algorithm only requires devices to upload 640 parameters in each round, comprising 1.01\% of the total model parameters.

Fig. \ref{fig:acc_sameM} shows the learning performance of the proposed KFL algorithm and two benchmarks on MNIST and CIFAR-10 datasets. From Fig. \ref{fig:mnist_sameM_S2}, compared to the state-of-art FedRep, it is observed that the proposed algorithm achieves 0.96\% accuracy improvement when 50 devices participate in each round learning process and obtains a 2.1\% accuracy gain when scheduling 10 devices in each round. In addition, the proposed algorithm converges faster than the benchmarks.
A similar experiment conducted on the CIFAR-10 dataset is shown in Fig. \ref{fig:cifar10_sameM_S2}. Similar to the results on the MINIST dataset, the proposed algorithm outperforms the benchmarks. Specifically, it improves 6.65\% accuracy when 10 devices are scheduled in each round. Although the proposed algorithm has similar accuracy to the FedRep when scheduling 50 devices in each round, it converges faster than the latter.
\begin{figure*}
\centering
\subfigure[]{\label{fig:mnist_sameM_S2}
\includegraphics[width=0.32\linewidth]{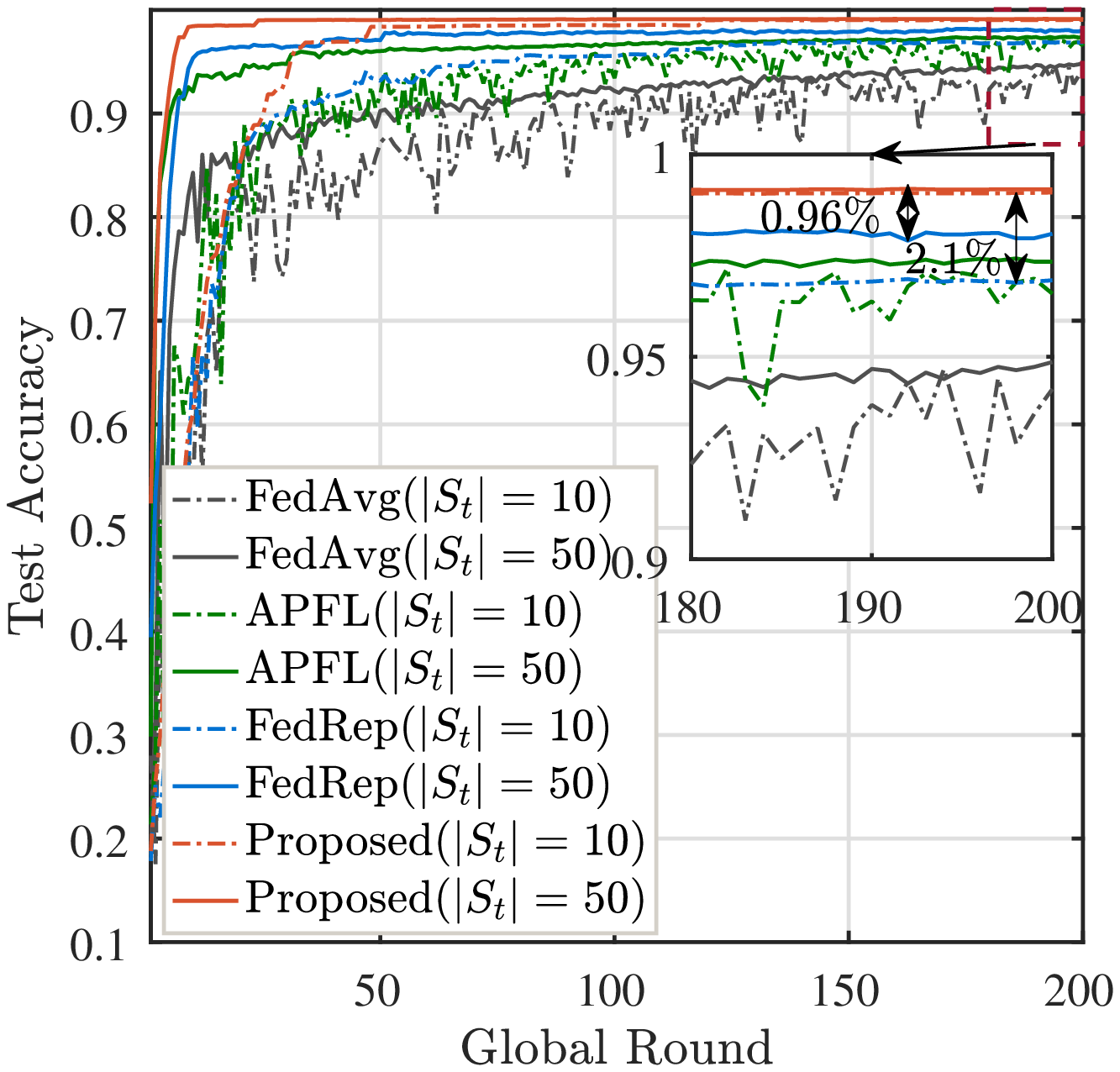}}
\hspace{-0.02\linewidth}
\subfigure[]{\label{fig:cifar10_sameM_S2}
\includegraphics[width=0.32\linewidth]{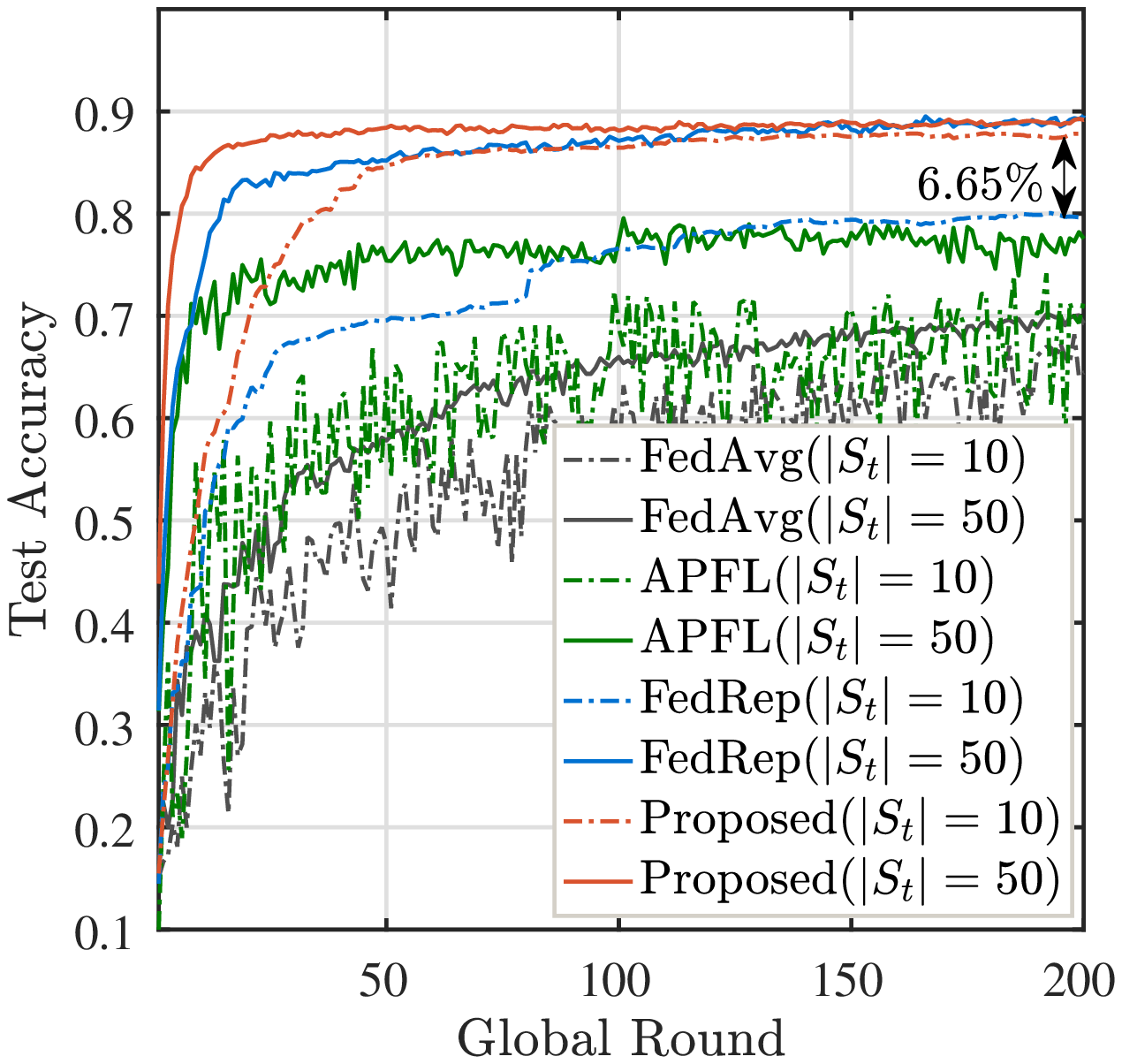}}
\hspace{-0.02\linewidth}
\subfigure[]{\label{fig:schedule_ORsameM_S2}
\includegraphics[width=0.32\linewidth]{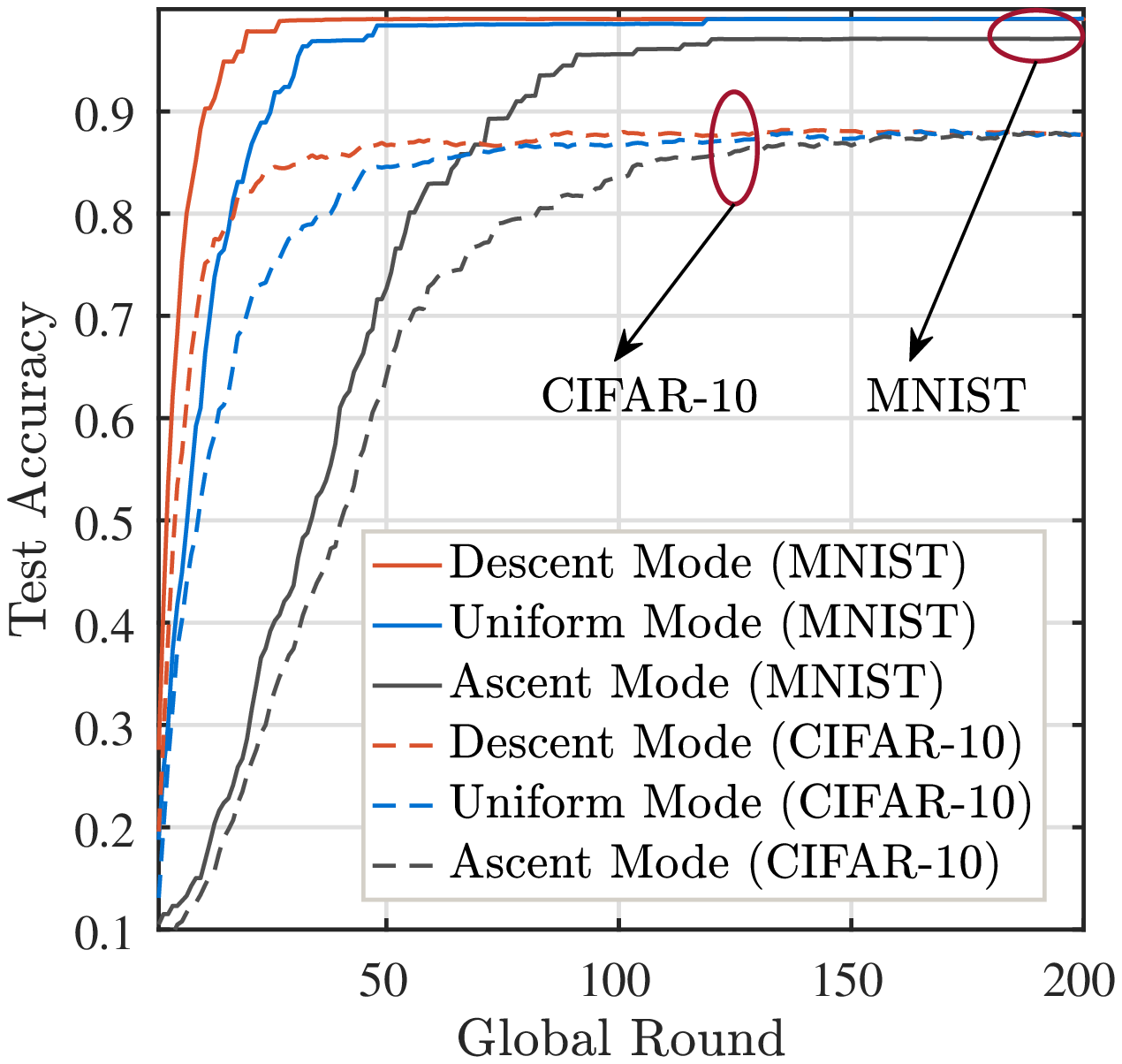}}
\vspace{-0.4cm}
\caption{Comparison of learning performance under homogeneous models (a) different algorithms on the MNIST dataset; (b) different algorithms on the CIFAR-10 dataset; (c) different scheduling patterns on MNIST and CIFAR-10 datasets.}
\label{fig:acc_sameM}
\end{figure*}

Fig. \ref{fig:schedule_ORsameM_S2} presents the test accuracy of the proposed KFL algorithm with different device scheduling patterns on MNIST and CIFAR-10 datasets. It is observed that the descend scheduling pattern converges faster than the other two scheduling patterns on these two datasets. This experimental result verifies the theoretical results in Remark \ref{rem:one}, which indicates that more scheduled data volume should bias to the early rounds if the entire scheduled data volume are fixed.

\subsection{Performance Evaluation with Heterogeneous Models}
In this subsection, we verify the effectiveness of the proposed KFL algorithm by comparing it with the FedKD \cite{li2019fedmd}, which is a knowledge distillation-based FL algorithm.
Note that in this subsection, devices are equipped with heterogeneous local models, and do not consider the energy and bandwidth limitations.
Since the knowledge distillation process requires aggregating devices' model output logits on an additional proxy dataset, we sample 50 data samples from each class (for both MNIST and CIFAR-10) to construct the proxy dataset with 500 data samples.
Note that as stated in the experimental setting, our proposed algorithm only requires devices to upload 640 parameters in each round, reducing 87\% of transmission costs compared with the knowledge distillation-based algorithm because the latter requires devices to upload $500 \times 10 = 5000$ parameters in each round.

Fig. \ref{fig:acc_diffM} presents the test accuracy of the proposed KFL algorithm and the knowledge distillation-based FL algorithm under heterogeneous devices' models. Fig. \ref{fig:mnist_diffM_S2} shows the results of the MNIST dataset. Compared to the FedKD algorithm, the proposed KFL algorithm achieves a slight test accuracy improvement, i.e., 0.61\% when 10 devices and 0.59\% when 50 devices participate in the per-round training. In addition, the proposed KFL algorithm convergences faster than the FedKD algorithm. Fig. \ref{fig:cifar10_diffM_S2} presents the results of the CIFAR-10 dataset which is more complex than the MNIST dataset. The proposed KFL algorithm obtains a more distinct learning performance gain on the CIFAR-10 dataset, i.e., compared to FedKD, improving 4\% and 9.35\% accuracy when 10 and 50 devices participate in per-round training, respectively. In fact, the learning performance of FedKD or other knowledge distillation-based FL algorithms heavily relies on the quality of the proxy dataset. In practical applications, the additional proxy dataset may not always be available, and its quality is usually not very high. Thus, the proposed KFL algorithm is flexible for practical scenarios.
\begin{figure*}
\centering
\subfigure[]{\label{fig:mnist_diffM_S2}
\includegraphics[width=0.32\linewidth]{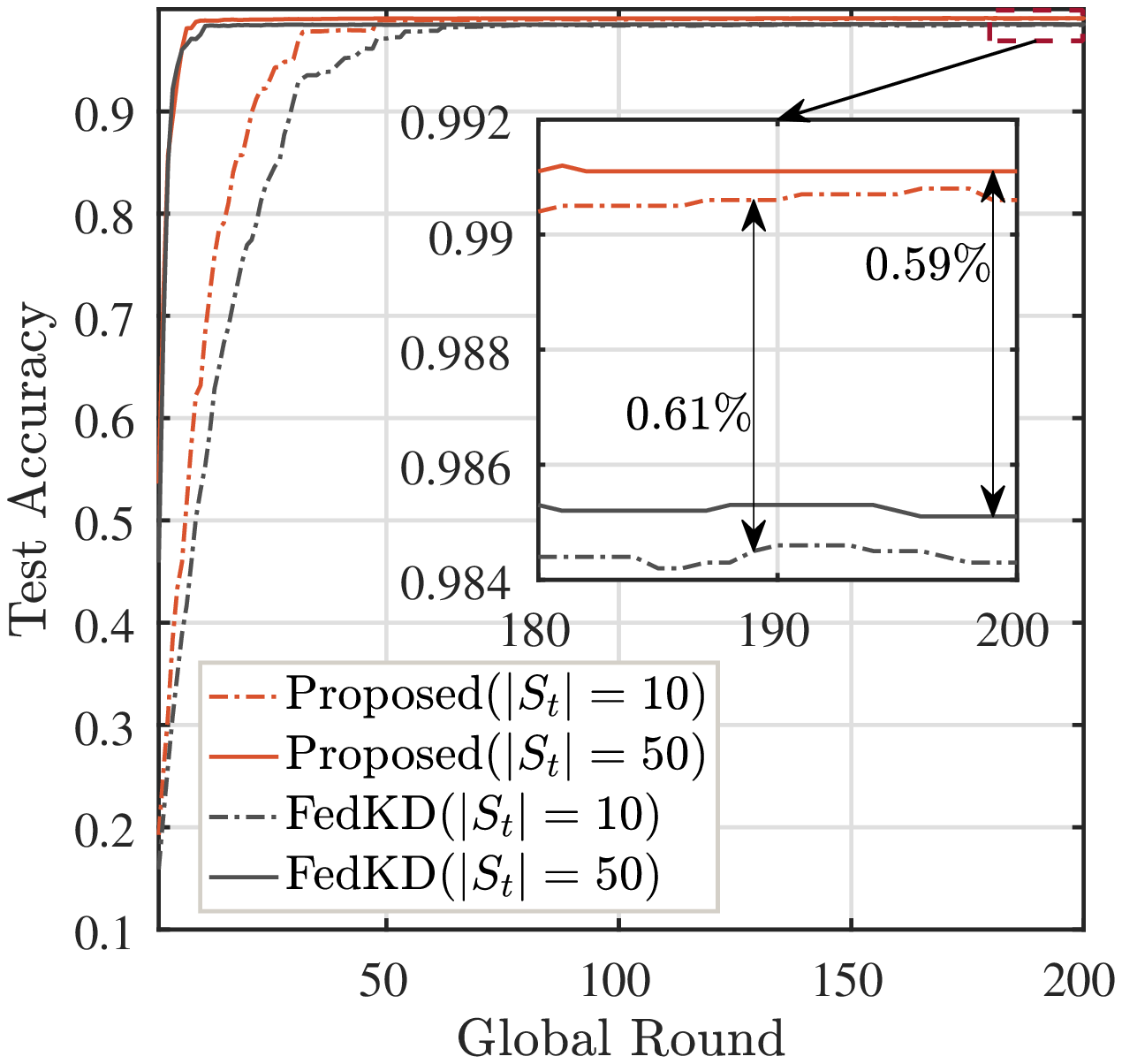}}
\hspace{-0.02\linewidth}
\subfigure[]{\label{fig:cifar10_diffM_S2}
\includegraphics[width=0.32\linewidth]{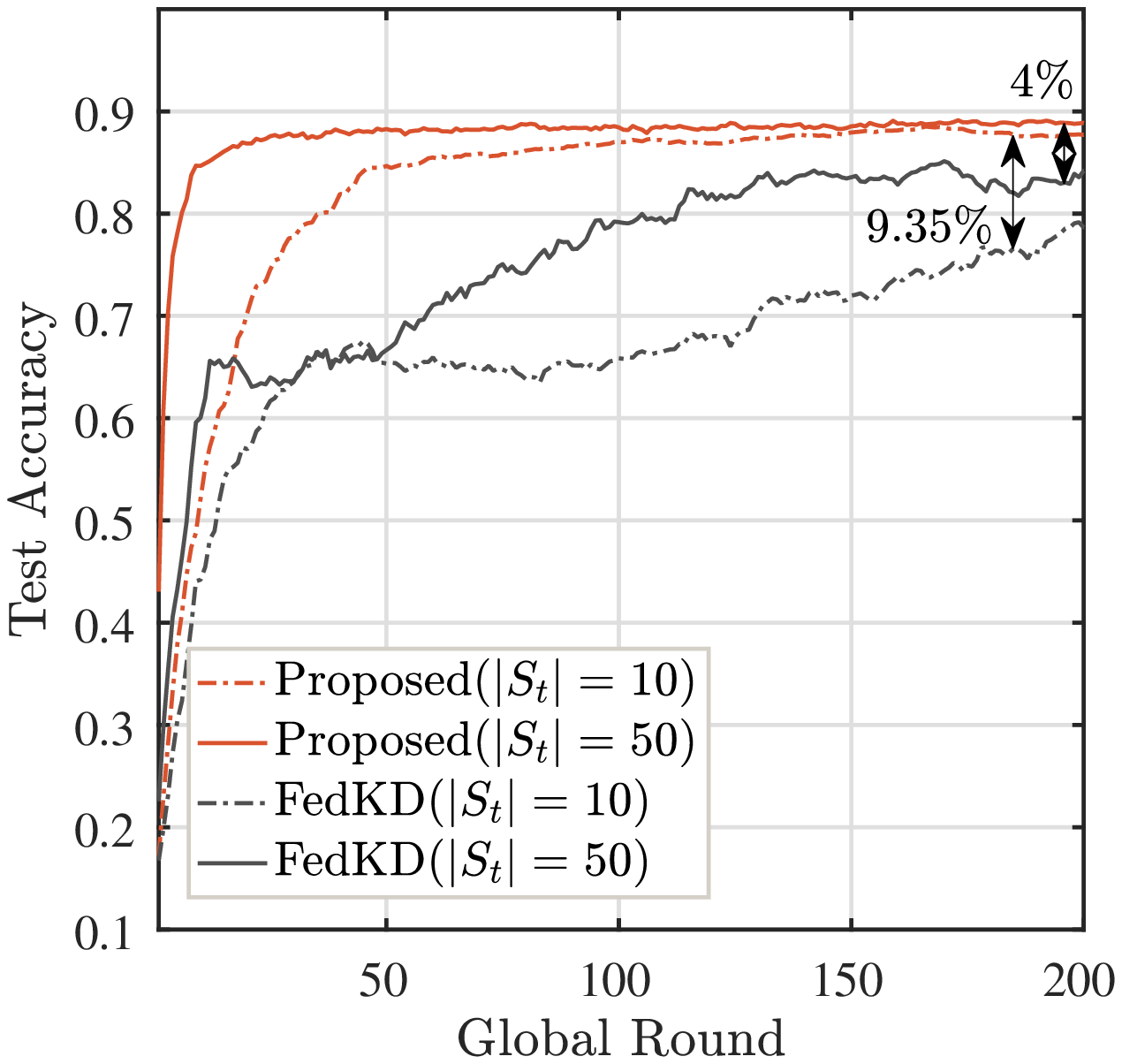}}
\hspace{-0.02\linewidth}
\subfigure[]{\label{fig:schedule_ORdiffM_S2}
\includegraphics[width=0.32\linewidth]{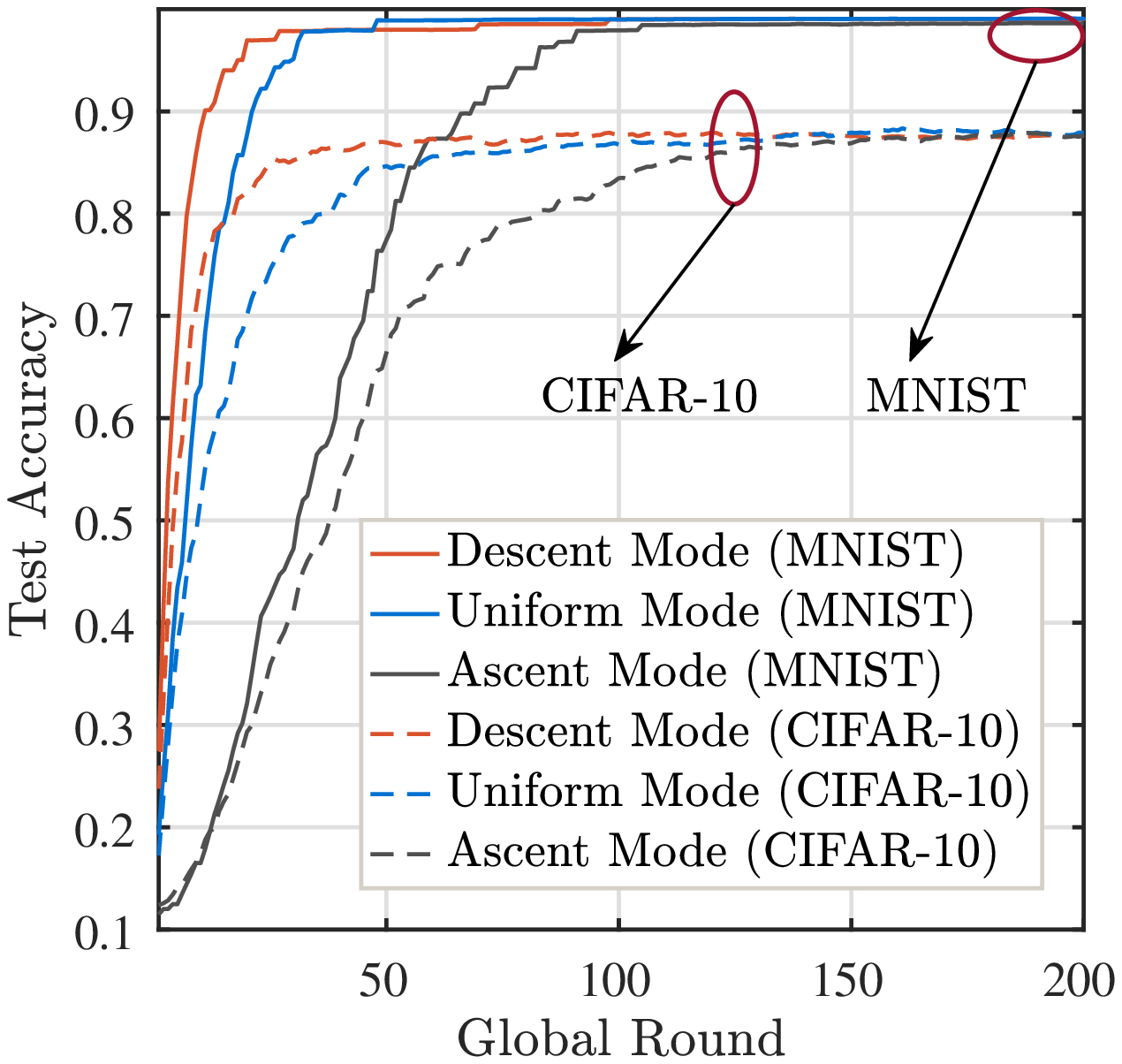}}
\caption{Comparison of learning performance under heterogeneous models (a) different algorithms on the MNIST dataset; (b) different algorithms on the CIFAR-10 dataset; (c) different scheduling patterns on MNIST and CIFAR-10 datasets.}
\label{fig:acc_diffM}
\end{figure*}
Fig. \ref{fig:schedule_ORdiffM_S2} shows a similar result to the experiment under homogeneous models, indicating that more data sample volume should be scheduled in the earlier rounds when the total scheduled data volume in the entire learning course are fixed.

\subsection{Performance of the Proposed Device Scheduling Algorithm}
This subsection evaluates the proposed device scheduling algorithm in the wireless network by comparing it with two benchmarks. Note that devices in this subsection are equipped with heterogeneous models.
1) Round Robin Scheduling Policy \cite{8851249}: In each round, the round robin policy selects a set of devices with the size of 5  (for both MNIST and CIFAR-10 dataset) that have sufficient energy to support its current local training and knowledge uploading to participate in the training process. This policy contributes a fairness scheduling among devices.
The size of the scheduled device set of the round robin is determined by the maximum overall scheduled devices of other scheduling algorithms divided by the number of rounds.
2) Myopic Scheduling Policy \cite{9207871}: For each device $k$, the available energy in round $t$ is given by the remaining energy divided by the remaining number of rounds, i.e.,
$\frac{E_k - \sum\nolimits_{t' = 0}^{t - 1} \alpha_{k,t'}E_{k,t'}}{T - t + 1}$.
Note that, in this subsection, devices are equipped with heterogeneous models.
For the proposed algorithm, we set $\gamma_t = \frac{1}{t}$ ($\forall t \in \{0, 1, \cdots, T-1\}$).
\begin{figure*}
\centering
\subfigure[]{\label{fig:mnistE_accdiffM__S2}
\includegraphics[width=0.25\linewidth]{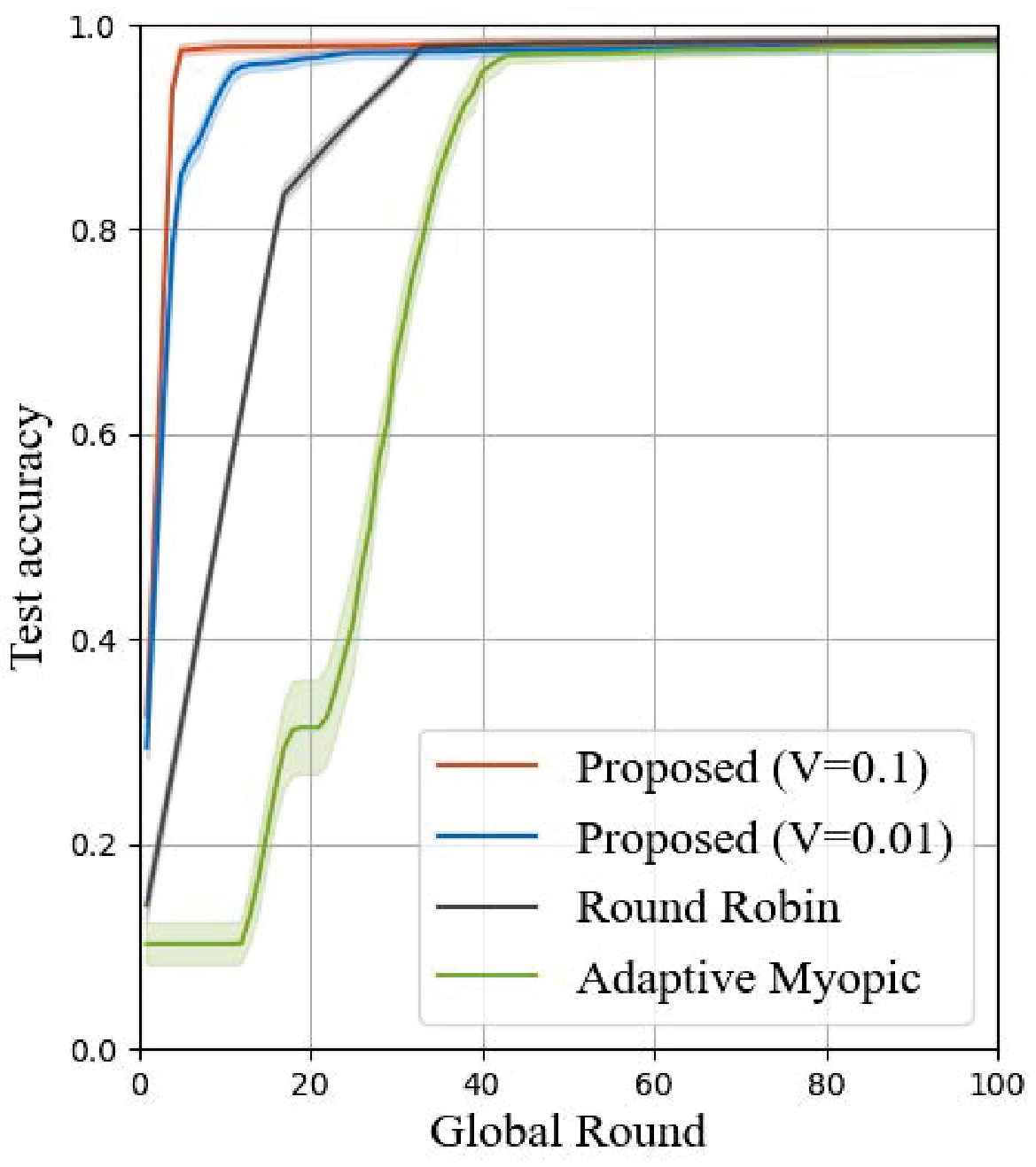}}
\hspace{-0.03\linewidth}
\subfigure[]{\label{fig:mnistE_enerdiffsameM__S2}
\includegraphics[width=0.25\linewidth]{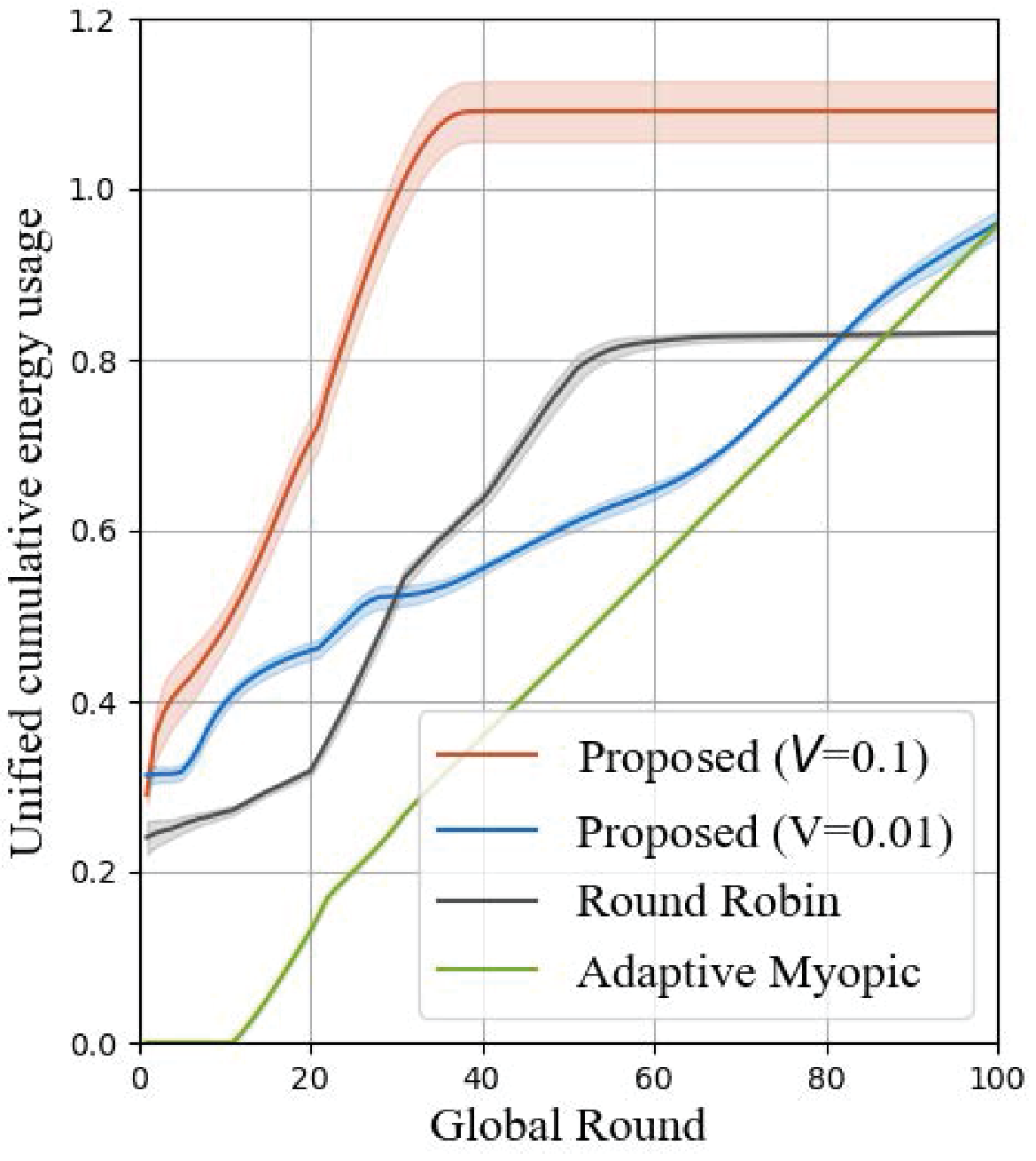}}
\hspace{-0.03\linewidth}
\subfigure[]{\label{fig:cifarE_accdiffM_S2}
\includegraphics[width=0.25\linewidth]{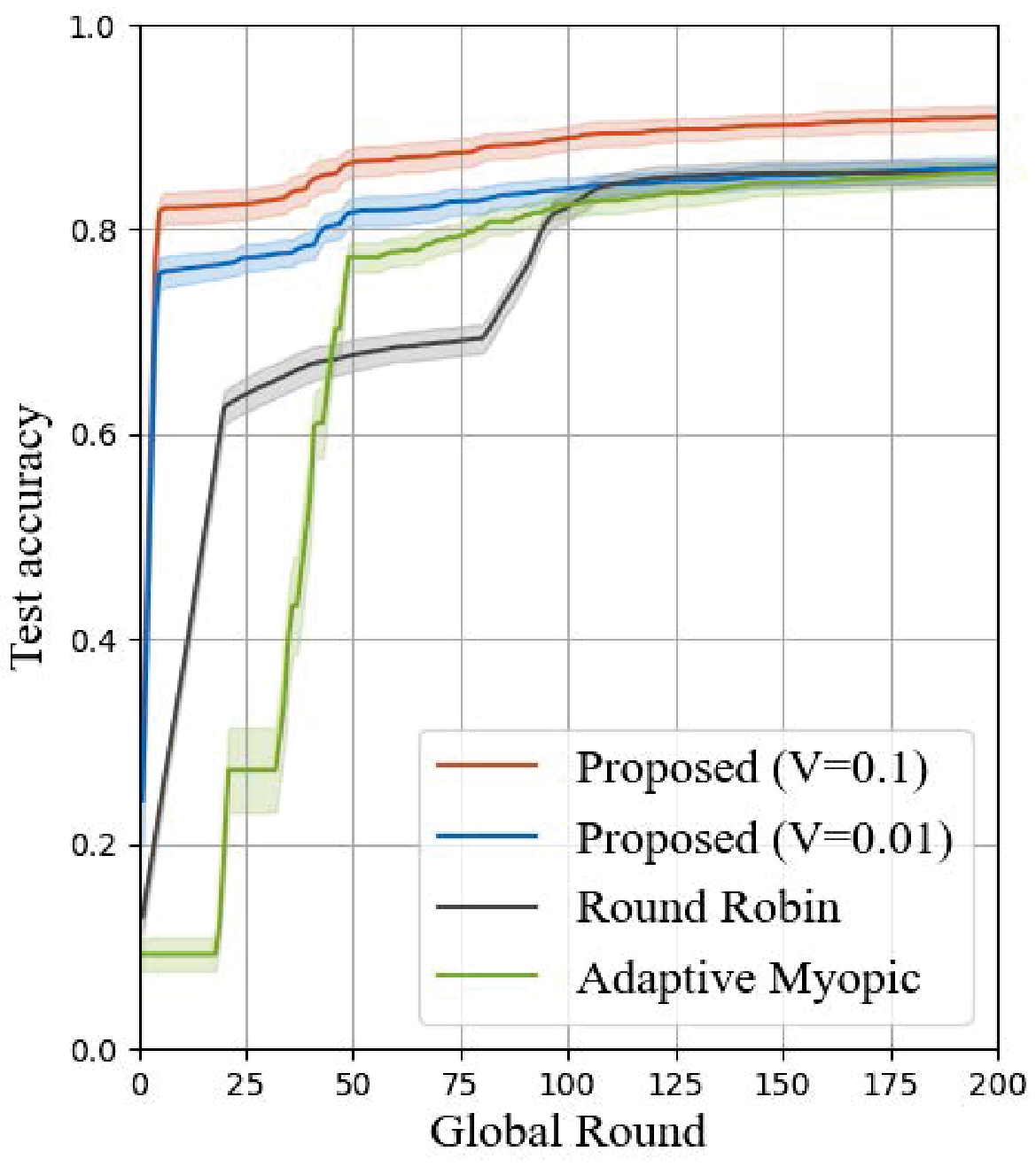}}
\hspace{-0.03\linewidth}
\subfigure[]{\label{fig:cifarE_enerdiffM_S2}
\includegraphics[width=0.25\linewidth]{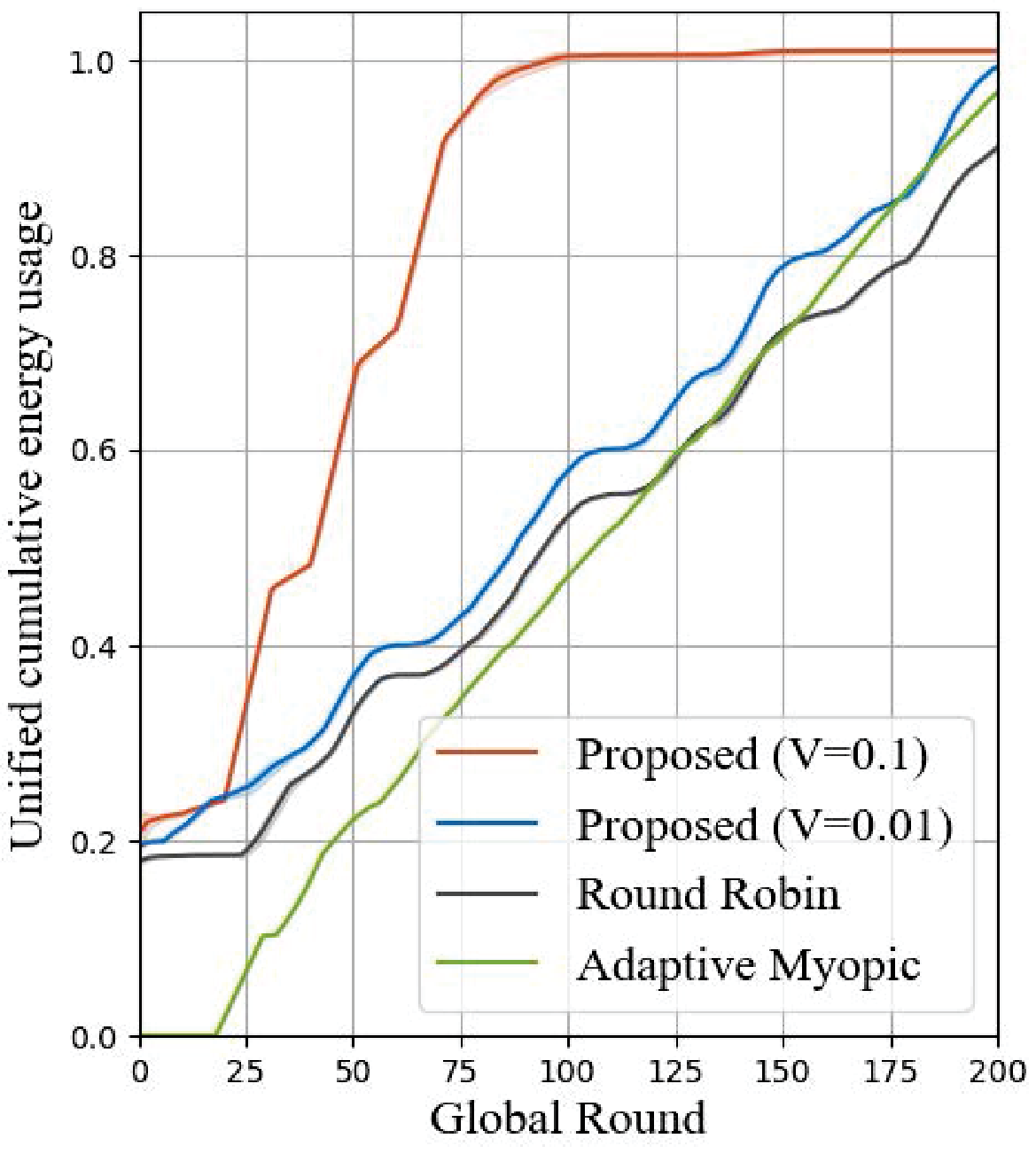}}
\caption{Comparison of learning performance in different device scheduling algorithms on the MNIST and CIFAR-10 datasets.}
\label{fig:scheduleE}
\end{figure*}
Fig. \ref{fig:mnistE_accdiffM__S2} and Fig. \ref{fig:mnistE_enerdiffsameM__S2} compare the test accuracy and cumulative energy usage of the scheduling algorithms on the MNIST dataset. It is observed from Fig. \ref{fig:mnistE_accdiffM__S2} that the proposed algorithm obtains a faster convergence speed and higher test accuracy than the benchmarks. From Fig. \ref{fig:mnistE_enerdiffsameM__S2}, we can see that the proposed scheduling algorithm with $V=0.01$ and $V=0.1$ have higher energy usage in the beginning 30 rounds than benchmarks. This induces a faster convergence speed of the proposed algorithm. Particularly, the proposed algorithm with $V=0.01$ has the same energy usage as the Adaptive Myopic algorithm. Both satisfy the energy constraints of devices (at the end of the training process, the unified energy usage is smaller than 1). However, the proposed algorithm with $V=0.01$ achieves better learning performance. This performance gain comes from the proposed algorithm enabling devices to use energy more flexibly, thus improving the training performance.

Similar comparison is made on CIFAR-10 dataset in Fig. \ref{fig:cifarE_accdiffM_S2} and Fig. \ref{fig:cifarE_enerdiffM_S2}. It is also observed that the proposed online device scheduling algorithm outperforms the baselines in accuracy and convergence speed. From Fig. \ref{fig:cifarE_enerdiffM_S2}, we can see that the proposed algorithm enables devices to consume more energy in the earlier rounds compared to the baselines, which indicates that the proposed algorithm schedules more data samples in the early rounds. Thus, based on Remark 1, the proposed algorithm obtains better learning performance than the baselines. Particularly, the proposed algorithm with $V=0.1$ enables devices to exhaust their energy in the former 100 rounds and achieve the best learning performance. The round robin algorithm enables devices to consume energy uniformly throughout the process. While for the Adaptive Myopic algorithm, the energy consumption at the former rounds exceeds the budget, and thus no devices are scheduled.
In fact, the proposed algorithm schedules devices in the descend scheduling pattern, while Adaptive Myopic schedules devices in the ascend scheduling pattern and Round Robin schedules devices in the uniform scheduling pattern. Thus, the result in Fig. 4(a) and \ref{fig:cifarE_enerdiffM_S2} also verified the correctness of our theoretical analysis in Remark 1, i.e., more data samples should be scheduled in the early rounds under restricted resources budgets.

\section{Conclusion}\label{sec:conclusion}
In this work, we have developed a novel KFL framework which aggregates devices' knowledge to enable collaborative training between devices. The benefits of this framework are three folds: 1) Allowing devices with heterogeneous models to train machine learning models collaboratively. 2) Significantly reducing the communication overhead of devices compared to conventional model aggregation-based FL approaches. 3) Mitigating the impact of non-IID data distribution among devices on learning performance.
Experimental results show that compared to conventional model aggregation-based FL algorithms,  the proposed KFL framework is able to reduce 99\% communication load while boosting 2.1\% and 6.65\% accuracy on MNIST and CIFAR-10 datasets, respectively.
In addition, we have theoretically and experimentally revealed that more scheduled data samples should be biased to the early rounds if the scheduled data samples of the entire learning process are fixed.
With this insight, we have developed an efficient online device scheduling and resource allocation algorithm to improve learning performance under devices' limited energy budgets.
Experimental results show that the proposed online device scheduling algorithm converges faster than the benchmark device scheduling algorithms.
In the future work, we will optimize the local models' design according to the devices' computing capabilities and datasets for further improving the learning performance of KFL.

\appendix
\subsection{Proof of Lemma \rm{\ref{lem:lip_F}}}\label{app:one}
Using $L_u$ smooth of $F_k(\cdot, \bm{v}_k)$ and $L_v$-smooth of $F(\bm{u}_k, \cdot)$, we have
\begin{align}\label{eq:local_Lu}
F_k(\bm{u}_k',\bm{v}_k') - F_k\left( {\bm{u}_k,\bm{v}_k'} \right)
\le \left\langle \nabla_{\bm{u}} F_k\left(\bm{u}_k,\bm{v}_k' \right),\bm{u}_k' - \bm{u}_k \right\rangle  + \frac{L_u}{2} \left\| \bm{u}_k' - \bm{u}_k \right\|^2,
\end{align}
and
\begin{align}\label{eq:local_Lv}
F_k\left( {\bm{u}_k,\bm{v}_k'} \right) - F_k(\bm{u}_k,\bm{v}_k)
 \le \left\langle \nabla_{\bm{v}} F_k(\bm{u}_k,\bm{v}_k),\bm{v}_k' - \bm{v}_k \right\rangle  + \frac{L_v}{2} \left\| \bm{v}_k' - \bm{v}_k \right\|^2.
\end{align}
Summarizing \eqref{eq:local_Lu} and \eqref{eq:local_Lv}, we have
\begin{multline}\label{eq:local_L_temp}
F_k(\bm{u}_k',\bm{v}_k') - F_k(\bm{u}_k,\bm{v}_k)
\le \left\langle {\nabla_{\bm{u}} F_k\left(\bm{u}_k,\bm{v}_k' \right),\bm{u}_k' - \bm{u}_k} \right\rangle  + \frac{L_u}{2} \left\| \bm{u}_k' - \bm{u}_k \right\|^2\\[-0.3cm]
+ \left\langle \nabla_{\bm{v}}F_k(\bm{u}_k,\bm{v}_k),\bm{v}_k' - \bm{v}_k \right\rangle  + \frac{L_v}{2} \left\| \bm{v}_k' - \bm{v}_k \right\|^2.
\end{multline}
We now focus on bounding $\left\langle \nabla_{\bm{u}} F_k\left(\bm{u}_k,\bm{v}_k' \right),\bm{u}_k' - \bm{u}_k \right\rangle$ as follows:
\begin{align}\label{local_Lu_inner}
&\left\langle \nabla_{\bm{u}}F_k(\bm{u}_k,\bm{v}_k'),\bm{u}_k' - \bm{u}_k \right\rangle \notag\\
&\overset{(a)}= \left\langle \nabla_{\bm{u}}F_k(\bm{u}_k,\bm{v}_k),\bm{u}_k' - \bm{u}_k \right\rangle
+ \left\langle \nabla_{\bm{u}}F_k\left( {\bm{u}_k,\bm{v}_k'} \right) - \nabla_{\bm{u}}F_k(\bm{u}_k,\bm{v}_k),\bm{u}_k' - \bm{u}_k \right\rangle \notag\\[-0.3cm]
&\overset{(b)}\le \left\langle {\nabla_{\bm{u}}F_k(\bm{u}_k,\bm{v}_k),\bm{u}_k' - \bm{u}_k} \right\rangle
+ \left\| {\nabla_{\bm{u}}F_k\left(\bm{u}_k,\bm{v}_k' \right) - \nabla_{\bm{u}}F_k(\bm{u}_k,\bm{v}_k)} \right\|\left\| \bm{u}_k' - \bm{u}_k \right\| \notag\\[-0.3cm]
&\overset{(c)}\le \left\langle \nabla_{\bm{u}}F_k(\bm{u}_k,\bm{v}_k),\bm{u}_k' - \bm{u}_k \right\rangle  + L_{uv} \left\| \bm{v}_k' - \bm{v}_k \right\|\left\| {\bm{u}_k' - \bm{u}_k} \right\| \notag\\
&\overset{(c)}\le \left\langle \nabla_{\bm{u}}F_k(\bm{u}_k,\bm{v}_k),\bm{u}_k' - \bm{u}_k \right\rangle  + \frac{1}{2}\chi {L_v}{\left\| \bm{v}_k' - \bm{v}_k \right\|^2} + \frac{1}{2}\chi {L_u}{\left\| \bm{u}_k' - \bm{u}_k \right\|^2},
\end{align}
where (a) is derived by adding and substracting $\nabla_{\bm{u}}F_k(\bm{u}_k,\bm{v}_k)$ into $\nabla_{\bm{u}}F_k(\bm{u}_k,\bm{v}_k')$, (b) follows the Cauchy-Schwarz inequality, (c) comes from Assumption \ref{assump:one}, (d) is due to the definition of $\chi$.
Substituting \eqref{local_Lu_inner} into \eqref{eq:local_L_temp}, the proof completes.

\subsection{Proof of Lemma \rm{\ref{lem:oneR_conver}}}\label{app:two}
According to Lemma \ref{lem:lip_F}, we have
\begin{multline}\label{eq:local_L_t}
F_k(\bm{u}_{k,t+1},\bm{v}_{k,t+1}) - F_k(\bm{u}_{k,t},\bm{v}_{k,t})
\le \left\langle \nabla_{\bm{u}} F_k(\bm{u}_{k,t},\bm{v}_{k,t}),\bm{u}_{k,t+1} - \bm{u}_{k,t} \right\rangle  + \frac{1+\chi}{2}{L_u}{\left\| {\bm{u}_{k,t+1} - \bm{u}_{k,t}} \right\|^2} \\
+ \left\langle \nabla_{\bm{v}}F_k(\bm{u}_{k,t},\bm{v}_{k,t}),\bm{v}_{k,t+1} - \bm{v}_{k,t} \right\rangle  + \frac{1+\chi}{2}{L_v}{\left\| {\bm{v}_{k,t+1} - \bm{v}_{k,t}} \right\|^2}.
\end{multline}
Below we focus on bounding the four terms on the right-hand side (RHS) of \eqref{eq:local_L_t}.
Firstly, we bound $\left\langle {\nabla_{\bm{u}}F_k(\bm{u}_{k,t},\bm{v}_{k,t}),\bm{u}_{k,t+1} - \bm{u}_{k,t}} \right\rangle$ as follows:
\begin{align}\label{eq:local_L_t_u1}
&\left\langle \nabla_{\bm{u}}F_k(\bm{u}_{k,t},\bm{v}_{k,t}),\bm{u}_{k,t+1} - \bm{u}_{k,t} \right\rangle
= - {\eta_u}\sum\limits_{l = 0}^{\tau - 1} \left\langle {\nabla_{\bm{u}}F_k(\bm{u}_{k,t},\bm{v}_{k,t}),\nabla_{\bm{u}}F_k(\bm{u}_{k,t,l},\bm{v}_{k,t,l}) + \lambda \nabla L_k(\bm{u}_{k,t,l})} \right\rangle \notag\\
&\overset{(a)}\le - \frac{\eta_u\tau}{2} \left\| \nabla_{\bm{u}}F_k(\bm{u}_{k,t},\bm{v}_{k,t}) \right\|^2
+ \frac{\eta_u}{2}\sum\nolimits_{l = 0}^{\tau-1} \left\| \nabla_{\bm{u}}F_k(\bm{u}_{k,t,l},\bm{v}_{k,t,l}) \!-\! \nabla_{\bm{u}}F_k(\bm{u}_{k,t},\bm{v}_{k,t}) + \lambda \nabla L_k(\bm{u}_{k,t,l}) \right\|^2 \notag\\
&\overset{(b)}\le - \frac{\eta_u\tau}{2} {\left\| {\nabla_{\bm{u}}F_k(\bm{u}_{k,t},\bm{v}_{k,t})} \right\|^2}
+ {\eta_u}\lambda^2\sum\nolimits_{l = 0}^{\tau - 1} {\left\| {\nabla L_k(\bm{u}_{k,t,l})} \right\|^2}  \notag\\
&~~~~~~~~~~~~~~~~~~~~~~~~~~~~~~~~~~~~~~~~+ {\eta_u}\sum\nolimits_{l = 0}^{\tau - 1} {\left\| \nabla_{\bm{u}}F_k(\bm{u}_{k,t,l},\bm{v}_{k,t,l}) - \nabla_{\bm{u}}F_k(\bm{u}_{k,t},\bm{v}_{k,t}) \right\|^2},
\end{align}
where (a) is derived by adding and substracting $\nabla_{\bm{u}}F_k(\bm{u}_{k,t},\bm{v}_{k,t})$ into $\nabla_{\bm{u}}F_k(\bm{u}_{k,t,l},\bm{v}_{k,t,l})$ and using the triangle inequality, (b) follows the triangle inequality.
For the second term on the RHS of \eqref{eq:local_L_t}, we bound $\left\| {\bm{u}_{k,t+1} - \bm{u}_{k,t}} \right\|^2$ as
\begin{align}\label{eq:local_L_t_u2}
\left\| {\bm{u}_{k,t+1} - \bm{u}_{k,t}} \right\|^2
&= \eta_u^2{\left\| {\sum\nolimits_{l = 0}^{\tau - 1}  \left( {\nabla_{\bm{u}}F_k(\bm{u}_{k,t,l},\bm{v}_{k,t,l}) + \lambda \nabla L_k(\bm{u}_{k,t,l})} \right)} \right\|^2} \notag\\
&\overset{(a)}\le \eta_u^2\tau \sum\nolimits_{l = 0}^{\tau - 1}  {\left\| {\nabla_{\bm{u}}F_k(\bm{u}_{k,t,l},\bm{v}_{k,t,l}) + \lambda \nabla L_k(\bm{u}_{k,t,l})} \right\|^2} \notag\\
&\overset{(b)}\le 2\eta_u^2\tau \sum\nolimits_{l = 0}^{\tau - 1}  {\left\| {\nabla_{\bm{u}}F_k(\bm{u}_{k,t,l},\bm{v}_{k,t,l})} \right\|^2}
+ 2\eta_u^2\tau \sum\nolimits_{l = 0}^{\tau - 1}  {\left\| {\lambda \nabla L_k(\bm{u}_{k,t,l})} \right\|^2} \notag\\
&\overset{(c)}\le 4\eta_u^2{\tau^2}{\left\| {\nabla_{\bm{u}}F_k(\bm{u}_{k,t},\bm{v}_{k,t})} \right\|^2}
+ 2\eta_u^2\tau \lambda^2\sum\nolimits_{l = 0}^{\tau - 1}  {\left\| {\nabla L_k(\bm{u}_{k,t,l})} \right\|^2} \notag\\
&+ 4 \eta_u^2\tau \sum\nolimits_{l = 0}^{\tau - 1}  {\left\| \nabla_{\bm{u}}F_k(\bm{u}_{k,t,l},\bm{v}_{k,t,l}) - \nabla_{\bm{u}}F_k(\bm{u}_{k,t},\bm{v}_{k,t}) \right\|^2},
\end{align}
where (a) is due to Jensen's inequality, (b) follows the triangle inequality, (c) is derived by adding and substracting $\nabla_{\bm{u}}F_k(\bm{u}_{k,t},\bm{v}_{k,t})$ into $\nabla_{\bm{u}}F_k(\bm{u}_{k,t,l},\bm{v}_{k,t,l})$.
We now focus on bounding $\sum\nolimits_{l = 0}^{\tau - 1} \| \nabla_{\bm{u}}F_k(\bm{u}_{k,t,l},\bm{v}_{k,t,l}) \!-\! \nabla_{\bm{u}}F_k(\bm{u}_{k,t},\bm{v}_{k,t}) \|^2$ which appears in both \eqref{eq:local_L_t_u1} and \eqref{eq:local_L_t_u2} as
\begin{align}\label{local_L_t_u_innerterm}
&\sum\nolimits_{l = 0}^{\tau - 1} {\left\| \nabla_{\bm{u}}F_k(\bm{u}_{k,t,l},\bm{v}_{k,t,l}) - \nabla_{\bm{u}}F_k(\bm{u}_{k,t},\bm{v}_{k,t}) \right\|^2} \notag\\
&\overset{(a)}\le 2\sum\limits_{l = 0}^{\tau  - 1}  \left( \left\| \nabla_{\bm{u}} F_k(\bm{u}_{k,t,l},\bm{v}_{k,t,l}) - \nabla_{\bm{u}} F_k(\bm{u}_{k,t},\bm{v}_{k,t,l}) \right\|^2 + \left\| \nabla_{\bm{u}} F_k(\bm{u}_{k,t},\bm{v}_{k,t,l}) - \nabla_{\bm{u}} F_k( \bm{u}_{k,t},\bm{v}_{k,t}) \right\|^2 \right) \notag\\
&\overset{(b)}\le 2\sum\nolimits_{l = 0}^{\tau - 1}  L_u^2{\left\| \bm{u}_{k,t,l}-\bm{u}_{k,t} \right\|^2} + 2\sum\nolimits_{l = 0}^{\tau - 1}  {\chi^2}L_uL_v{\left\| \bm{v}_{k,t,l}-\bm{v}_{k,t} \right\|^2}.
\end{align}
where (a) derived by adding and substracting $\nabla_{\bm{u}} F_k(\bm{u}_{k,t},\bm{v}_{k,t,l})$ and using the triangle inequality, (b) follows Assumption \ref{assump:one} and the definition of $\chi$.

For the last two terms on the RHS of \eqref{eq:local_L_t}, we have
\begin{align}\label{eq:local_L_t_v}
&\left\langle \nabla_{\bm{v}}F_k(\bm{u}_{k,t},\bm{v}_{k,t}),\bm{v}_{k,t+1} - \bm{v}_{k,t} \right\rangle  + \frac{1+\chi}{2}{L_v}{\left\| {\bm{v}_{k,t+1} - \bm{v}_{k,t}} \right\|^2} \notag\\[-0.2cm]
&=  - {\eta_v}\sum\nolimits_{l = 0}^{\tau - 1}  \left\langle {\nabla_{\bm{v}}F_k(\bm{u}_{k,t},\bm{v}_{k,t}),\nabla_{\bm{v}}F_k(\bm{u}_{k,t,l},\bm{v}_{k,t,l})} \right\rangle
+ \frac{1+\chi}{2}{L_v}\eta_v^2{\left\| {\sum\nolimits_{l = 0}^{\tau - 1}  \nabla_{\bm{v}}F_k(\bm{u}_{k,t,l},\bm{v}_{k,t,l})} \right\|^2} \notag\\
&\overset{(a)}\le  - {\eta_v}\sum\nolimits_{l = 0}^{\tau - 1}  \left\langle {\nabla_{\bm{v}}F_k(\bm{u}_{k,t},\bm{v}_{k,t}),\nabla_{\bm{v}}F_k(\bm{u}_{k,t,l},\bm{v}_{k,t,l})} \right\rangle
+ \frac{1+\chi}{2}{L_v}\eta_v^2\tau \sum\nolimits_{l = 0}^{\tau - 1}  {\left\| {\nabla_{\bm{v}}F_k(\bm{u}_{k,t,l},\bm{v}_{k,t,l})} \right\|^2} \notag\\
&\overset{(b)}\le ( (1+\chi)L_v\eta_v^2 \tau^2 - \frac{1}{2}\eta_v\tau){\left\| {\nabla_{\bm{v}}F_k(\bm{u}_{k,t},\bm{v}_{k,t})} \right\|^2} \notag\\
&+ ((1+\chi)L_v\eta_v^2\tau  + \frac{1}{2}\eta_v)\sum\nolimits_{l = 0}^{\tau - 1}  {\left\| \nabla_{\bm{v}}F_k(\bm{u}_{k,t,l},\bm{v}_{k,t,l}) - \nabla_{\bm{v}}F_k(\bm{u}_{k,t},\bm{v}_{k,t}) \right\|^2},
\end{align}
where (a) is due to Jensen's inequality, (b) is derived by adding and substracting $\nabla_{\bm{v}}F_k(\bm{u}_{k,t},\bm{v}_{k,t})$ into $\nabla_{\bm{v}}F_k(\bm{u}_{k,t,l},\bm{v}_{k,t,l})$ and using the triangle inequality.
In \eqref{eq:local_L_t_v}, we bound $\sum\nolimits_{l = 0}^{\tau - 1} \| - \nabla_{\bm{v}}F_k(\bm{u}_{k,t},\bm{v}_{k,t}) + \nabla_{\bm{v}}F_k(\bm{u}_{k,t,l},\bm{v}_{k,t,l}) \|^2$ as
\begin{align}\label{eq:local_L_t_v_innerterm}
&\sum\nolimits_{l = 0}^{\tau - 1}  {\left\| \nabla_{\bm{v}}F_k(\bm{u}_{k,t,l},\bm{v}_{k,t,l}) - \nabla_{\bm{v}}F_k(\bm{u}_{k,t},\bm{v}_{k,t})  \right\|^2} \notag\\
&\le 2\sum\limits_{l = 0}^{\tau - 1} \left\|\nabla_{\bm{v}}F_k(\bm{u}_{k,t,l},\bm{v}_{k,t,l}) - \nabla_{\bm{v}}F_k(\bm{u}_{k,t},\bm{v}_{k,t,l}) \right\|^2
+ 2\sum\limits_{l = 0}^{\tau - 1} \left\|\nabla_{\bm{v}}F_k(\bm{u}_{k,t},\bm{v}_{k,t,l}) - \nabla_{\bm{v}}F_k(\bm{u}_{k,t},\bm{v}_{k,t}) \right\|^2  \notag\\
&\le 2\sum\nolimits_{l = 0}^{\tau - 1} {\chi^2} L_uL_v{\left\| \bm{u}_{k,t,l}-\bm{u}_{k,t} \right\|^2} + 2\sum\nolimits_{l = 0}^{\tau - 1} L_v^2{\left\| \bm{v}_{k,t,l}-\bm{v}_{k,t} \right\|^2}.
\end{align}

Substituting \eqref{eq:local_L_t_u1}, \eqref{eq:local_L_t_u2}, \eqref{local_L_t_u_innerterm}, \eqref{eq:local_L_t_v}, and \eqref{eq:local_L_t_v_innerterm} into \eqref{eq:local_L_t}, and the learning rates satisfy $\eta_u \le \frac{1}{4\tau(1 + \chi)L_u}$ and $\eta_v \le \frac{1}{2\tau(1 + \chi)L_v}$, we have
\begin{multline}\label{eq:local_F_intermidiate}
F_k(\bm{u}_{k,t + 1},\bm{v}_{k,t + 1}) - F_k(\bm{u}_{k,t},\bm{v}_{k,t})
\le \Big{(}2(1 + \chi)L_u\eta_u^2{\tau^2} - \frac{1}{2}{\eta_u}\tau \Big{)} \left\| \nabla_{\bm{u}} F_k(\bm{u}_{k,t},\bm{v}_{k,t}) \right\|^2\\
+ \Big{(}(1 + \chi)L_v\eta_v^2{\tau^2} - \frac{1}{2}{\eta_v}\tau \Big{)}\left\| {{\nabla_{\bm{v}}}F_k(\bm{u}_{k,t},\bm{v}_{k,t})} \right\|^2
+ (3\eta_uL_u^2 + 2\eta_v \chi^2 L_uL_v) \sum\nolimits_{l = 0}^{\tau - 1} {\left\| \bm{u}_{k,t,l} - \bm{u}_{k,t} \right\|^2} \\
+ (3\eta_u\chi^2 L_uL_v + 2{\eta_v}L_v^2)\sum\nolimits_{l = 0}^{\tau - 1}  {\left\| \bm{v}_{k,t,l} - \bm{v}_{k,t} \right\|^2}
+ \frac{5}{4}{\eta_u}\lambda^2\sum\nolimits_{l = 0}^{\tau - 1} {\left\| {\nabla {L_k}(\bm{u}_{k,t,l})} \right\|^2}.
\end{multline}

Below we focus on bounding two terms in \eqref{eq:local_F_intermidiate}, i.e., $\sum\nolimits_{l = 0}^{\tau - 1} {\left\| \bm{v}_{k,t,l} - \bm{v}_{k,t} \right\|^2}$ and $\sum\nolimits_{l = 0}^{\tau - 1}  {\left\| \bm{u}_{k,t,l} - \bm{u}_{k,t} \right\|^2}$. Firstly, for $\sum\nolimits_{l = 0}^{\tau - 1} {\left\| \bm{v}_{k,t,l} - \bm{v}_{k,t} \right\|^2}$, we have
\begin{multline}\label{eq:local_u_oneR}
\sum\nolimits_{l = 0}^{\tau - 1}  {\left\| \bm{v}_{k,t,l} - \bm{v}_{k,t} \right\|^2} = \sum\nolimits_{l = 0}^{\tau - 1} \eta_v^2 \left\| \sum\nolimits_{n = 0}^{l - 1}  \nabla_{\bm{v}}F_k(\bm{u}_{k,t,n},\bm{v}_{k,t,n}) \right\|^2 \\
\overset{(a)}\le \sum\nolimits_{l = 0}^{\tau - 1} \eta_v^2l\sum\nolimits_{n = 0}^{l - 1} {\left\| \nabla_{\bm{v}}F_k(\bm{u}_{k,t,n},\bm{v}_{k,t,n}) \right\|^2}
\overset{(b)}\le \eta_v^2G_2^2\frac{\tau (\tau  + 1)(2\tau + 1)}{6},
\end{multline}
where (a) comes from the Jensen's inequality, (b) follows the bounded gradient assumption in Assumption \ref{assump:two}.
For $\sum\nolimits_{l = 0}^{\tau - 1}  {\left\| \bm{u}_{k,t,l} - \bm{u}_{k,t} \right\|^2}$, we have
\begin{multline}\label{eq:local_v_oneR}
\sum\nolimits_{l = 0}^{\tau - 1} {\left\| \bm{u}_{k,t,l} - \bm{u}_{k,t} \right\|^2}
= \sum\nolimits_{l = 0}^{\tau - 1} \eta_u^2{\left\| \sum\nolimits_{n = 0}^{l - 1} \left(\nabla_{\bm{u}}F_k(\bm{u}_{k,t,n},\bm{v}_{k,t,n}) + \lambda \nabla{L_k}(\bm{u}_{k,t,n}) \right) \right\|^2} \\
\overset{(a)}\le \sum\nolimits_{l = 0}^{\tau - 1} \eta_u^2l\sum\nolimits_{n = 0}^{l - 1} {\left\|{\nabla_{\bm{u}}}F_k(\bm{u}_{k,t,n},\bm{v}_{k,t,n}) + \lambda \nabla{L_k}(\bm{u}_{k,t,n}) \right\|^2} \\
\overset{(b)}\le \sum\nolimits_{l = 0}^{\tau - 1} \eta_u^2l\sum\nolimits_{n = 0}^{l - 1} 2{\left\| {\nabla_{\bm{u}}}F_k(\bm{u}_{k,t,n},\bm{v}_{k,t,n}) \right\|^2}  + \sum\nolimits_{l = 0}^{\tau - 1} \eta_u^2l\sum\nolimits_{n = 0}^{l - 1} 2\left\| {\lambda \nabla{L_k}(\bm{u}_{k,t,n})} \right\|^2 \\
\overset{(c)}\le \frac{\tau (\tau + 1)(2\tau + 1)}{3}\eta_u^2G_1^2 + 2\eta_u^2\lambda^2\sum\nolimits_{l = 0}^{\tau - 1} l\sum\nolimits_{n = 0}^{l - 1} {\left\|\nabla{L_k}(\bm{u}_{k,t,n}) \right\|^2},
\end{multline}
where (a) is due to the Jensen's inequality, (b) follows the triangle inequlity, (c) is due to Assumption \ref{assump:two}.
Substituting \eqref{eq:local_u_oneR} and \eqref{eq:local_v_oneR} into \eqref{eq:local_F_intermidiate}, the proof is completed.

\subsection{Proof of Theorem \rm{\ref{thm:one}} }\label{app:three}
By substituting \eqref{eq:lemma_eq} into \eqref{eq:glo_EM_loss}, we have the one-round convergence bounded of the global loss as follows:
\begin{multline}\label{eq:diff_F_oneR}
F(\bm{W}_{t + 1}) - F(\bm{W}_t)
\le \sum\nolimits_{k = 1}^K \frac{D_k}{D}(2(1 + \chi)L_u\eta_u^2 \tau^2 - \frac{1}{2}\eta_u\tau) \left\| \nabla_{\bm{u}} F_k(\bm{u}_{k,t},\bm{v}_{k,t}) \right\|^2 \\
+ \sum\limits_{k = 1}^K \frac{D_k}{D}((1 + \chi)L_v\eta_v^2{\tau^2} - \frac{1}{2}\eta_v\tau)\left\|\nabla_v F_k(\bm{u}_{k,t},\bm{v}_{k,t}) \right\|^2
+ \frac{5}{4}\eta_u \lambda^2\sum\limits_{k = 1}^K \frac{D_k}{D}\sum\limits_{l = 0}^{\tau - 1} \left\| \nabla {L_k}(\bm{u}_{k,t,l}) \right\|^2 \\
+ A_1
+ 2\eta_u^2 \lambda^2(3\eta_uL_u^2 + 2\eta_v\chi^2 L_uL_v)\sum\nolimits_{k = 1}^K \frac{D_k}{D}\sum\nolimits_{l = 0}^{\tau - 1} (\tau - l)\left\| \nabla L_k(\bm{u}_{k,t,l}) \right\|^2,
\end{multline}
Below we bound $\left\| \nabla {L_k}(\bm{u}_{k,t,l}) \right\|^2$. For ease of proof, we introduce an auxiliary variable ${\bar{\bm{\Omega}}}_{c,t}= \frac{\sum\nolimits_{k \in \mathcal{K}} D_{k,c} \bm{\Omega}_{k,c,t}}{\sum\nolimits_{k \in \mathcal{K}} D_{k,c}}$, which aggregates all devices's knowledge about class $c$ ($\forall c \in \mathcal{C}$).
\begin{multline}\label{eq:grad_prop_loss}
\left\| \nabla {L_k}(\bm{u}_{k,t,l}) \right\|^2 = \left\| \frac{1}{D_k}\sum\nolimits_{c = 1}^C \sum\nolimits_{(\bm{x},y) \in \mathcal{D}_{k,c}} \left\| h_k(\bm{u}_{k,t,l};\bm{x}) - \bm{\Omega}_{c,t} \right\|\nabla h_k(\bm{u}_{k,t,l};\bm{x}) \right\|^2 \\
\overset{(a)}\le \frac{1}{D_k^2}C\sum\nolimits_{c = 1}^C D_{k,c}\sum\nolimits_{(\bm{x},y) \in \mathcal{D}_{k,c}} {\left\| h_k(\bm{u}_{k,t,l};x) - \bm{\Omega}_{c,t} \right\|^2}{\left\| \nabla h_k(\bm{u}_{k,t,l};\bm{x}) \right\|^2} \\
\overset{(b)}\le \frac{1}{D_k^2}C\sum\nolimits_{c = 1}^C D_{k,c}\sum\nolimits_{(\bm{x},y) \in \mathcal{D}_{k,c}} \left\| h_k(\bm{u}_{k,t,l};\bm{x}) - \bm{\Omega}_{c,t} \right\|^2 \vartheta^2 \\
\overset{(c)}\le 2\frac{1}{D_k^2}\vartheta^2C\sum\limits_{c = 1}^C {D_{k,c}}\sum\limits_{(\bm{x},y) \in \mathcal{D}_{k,c}}  {\left\| h_k(\bm{u}_{k,t,l};\bm{x}) - {\bar{\bm{\Omega}}}_{c,t} \right\|^2}
+ 2\frac{1}{D_k^2}\vartheta^2C\sum\limits_{c = 1}^C D_{k,c}^2  \left\| {\bar{\bm{\Omega}}}_{c,t} - \bm{\Omega}_{c,t} \right\|^2,
\end{multline}
where (a) follows Jensen's inequality, (b) is due to Assumption \ref{assump:three}, (c) derived by adding and substracting ${\bar{\bm{\Omega}}}_{c,t}$ into $\bm{\Omega}_{c,t}$ and using the triangle inequality.

Below we focus on bounding the two terms on the RHS of \eqref{eq:grad_prop_loss}, where the first term is bounded as
\begin{align}\label{eq:diff_local_fullProp}
&2\frac{1}{{D_k^2}}\vartheta^2C\sum\nolimits_{c = 1}^C D_{k,c}\sum\nolimits_{(\bm{x},y) \in \mathcal{D}_{k,c}} \left\| h_k(\bm{u}_{k,t,l};\bm{x}) - {\bar{\bm{\Omega}}}_{c,t} \right\|^2 \notag\\
&= 2\frac{1}{D_k^2}\vartheta^2C\sum\nolimits_{c = 1}^C  D_{k,c}\sum\nolimits_{(\bm{x},y) \in \mathcal{D}_{k,c}} \Big\| \frac{1}{D_c}\sum\nolimits_{h = 1}^K \sum\nolimits_{(\bm{x}_1,y_1) \in \mathcal{D}_{n,c}} \left(h_k(\bm{u}_{k,t,l};\bm{x}) - h_n(\bm{u}_{h,t};\bm{x}_1) \right) \Big\|^2 \notag\\
&\le 8\vartheta^2 \varsigma^2,
\end{align}
where the inequality is due to Jensen's inequality and Assumption \ref{assump:three}.
For the second term on the RHS of \eqref{eq:grad_prop_loss}, we have
\begin{align}\label{eq:diff_full_part_prop}
&\left\| \bar{\bm{\Omega}}_{c,t} - \bm{\Omega}_{c,t} \right\|^2
= \Big\| \frac{\sum\nolimits_{k = 1}^K \sum\nolimits_{(\bm{x}_1,y_1) \in \mathcal{D}_{k,c}} h_k(\bm{u}_{k,t};\bm{x}_1)}{D_c} - \frac{\sum\nolimits_{k = 1}^K \alpha _{k,t-1}\sum\nolimits_{(\bm{x}_1,y_1) \in \mathcal{D}_{k,c}} h_k(\bm{u}_{k,t};\bm{x}_1)}{\sum\nolimits_{k = 1}^K  \alpha _{k,t-1}D_{k,c}} \Big\|^2 \notag\\
&= \Bigg\| \frac{\sum\nolimits_{k = 1}^K \! (1 \!-\! \alpha_{k,t-1})\!\!\!\!\!\sum\limits_{(\bm{x}_1,y_1) \in \mathcal{D}_{k,c}} \!\!\!\!\!\!h_k(\bm{u}_{k,t};\bm{x}_1)}{D_c} \!\!-\! \frac{(D_c \!-\!\! \sum\nolimits_{k = 1}^K \! \alpha_{k,t-1}D_{k,c})\sum\nolimits_{k = 1}^K \! \alpha_{k,t-1}~\sum\limits_{\mathclap{(\bm{x}_1,y_1) \in \mathcal{D}_{k,c}}}~ h_k(\bm{u}_{k,t};\bm{x}_1)}{D_c\sum\nolimits_{k = 1}^K \alpha_{k,t-1}D_{k,c}} \Bigg\|^2 \notag\\
&\le \Bigg{(}\frac{\sum\nolimits_{k = 1}^K \! (1 \!-\! \alpha_{k,t-1}) \sum\limits_{\mathclap{(\bm{x}_1,y_1) \in \mathcal{D}_{k,c}}}  \left\| h_k(\bm{u}_{k,t};\bm{x}_1) \right\| }{D_c} \!+\! \frac{(D_c \!-\! \! \sum\nolimits_{k = 1}^K \! \alpha_{k,t - 1}D_{k,c})\sum\nolimits_{k = 1}^K \!\! \alpha_{k,t - 1}\sum\limits_{\mathclap{(\bm{x}_1,y_1) \in \mathcal{D}_{k,c}}} \left\| h_k(\bm{u}_{k,t};\bm{x}_1) \right\|}{D_c\sum\nolimits_{k = 1}^K \alpha_{k,t-1}\mathcal{D}_{k,c}} \Bigg{)}^2 \notag\\
&\overset{(a)}\le 4\varsigma^2\Big{(} \frac{D_c - \sum\nolimits_{k = 1}^K \alpha_{k,t - 1}D_{k,c}}{D_c} \Big{)}^2,
\end{align}
where (a) is due to Assumption \ref{assump:three}.
Substituting \eqref{eq:grad_prop_loss}, \eqref{eq:diff_local_fullProp}, and \eqref{eq:diff_full_part_prop} into \eqref{eq:diff_F_oneR}, then substracting $F(\bm{W}^*)$ into both $F(\bm{W}_{t + 1})$ and $F(\bm{W}_t)$, we have
\begin{multline}\label{eq:global_op_diff1}
F(\bm{W}_{t + 1}) \!-\! F(\bm{W}^*)
\le F(\bm{W}_t) \!-\! F(\bm{W}^*) \!+\! \Big{(}2(1 + \chi)L_u\eta_u^2{\tau^2} - \frac{1}{2}{\eta_u}\tau  \Big{)} \sum\limits_{k = 1}^K \frac{D_k}{D}{\left\| {{\nabla_{\bm{u}}}F_k(\bm{u}_{k,t},\bm{v}_{k,t})} \right\|^2} \\
+ \Big{(}(1 + \chi)L_v\eta_v^2{\tau ^2} - \frac{1}{2}{\eta_v}\tau \Big{)}\sum\nolimits_{k = 1}^K \frac{D_k}{D}{\left\| {{\nabla_{\bm{v}}}F_k(\bm{u}_{k,t},\bm{v}_{k,t})} \right\|^2} \\
+A_1 + A_2
+ A_2 \sum\nolimits_{k = 1}^K \frac{1}{D}\frac{1}{D_k}C\sum\nolimits_{c = 1}^C  D_{k,c}^2 \Big{(}\frac{D_c - \sum\nolimits_{k = 1}^K \alpha _{k,t - 1}D_{k,c}}{D_c} \Big{)}^2,
\end{multline}
where $A_2 = 10\eta_u\lambda^2\tau \vartheta^2\varsigma^2 + 8\eta_u^2\lambda^2 \vartheta^2\varsigma^2\left( 3\eta_uL_u^2 + 2\eta_v\chi^2L_uL_v \right)\tau (\tau + 1)$.

By using the L-smooth of loss functions, we have
\begin{align}\label{eq:L_S_u_op}
\left\| \nabla_{\bm{u}} F_k(\bm{u}_{k,t},\bm{v}_{k,t})\right\|^2 \le 2L_u\left( F_k(\bm{u}_{k,t},\bm{v}_{k,t}) - F_k(\bm{u}_k^*,\bm{v}_k^*) \right),
\end{align}
and
\begin{align}\label{eq:L_S_v_op}
\left\| \nabla_{\bm{v}}F_k(\bm{u}_{k,t},\bm{v}_{k,t}) \right\|^2 \le 2L_v\left(F_k(\bm{u}_{k,t},\bm{v}_{k,t}) - F_k(\bm{u}_k^*,\bm{v}_k^*) \right).
\end{align}
Substituting \eqref{eq:L_S_u_op} and \eqref{eq:L_S_v_op} into \eqref{eq:global_op_diff1}, we have
\begin{multline}
F(\bm{W}_{t + 1}) - F(\bm{W}^*)
\le A_3(F(\bm{W}_t) - F(\bm{W}^*))
+ A_1
+ A_2 \\
 + A_2 \sum\nolimits_{k = 1}^K \frac{1}{D}\frac{1}{D_k}C\sum\nolimits_{c = 1}^C D_{k,c}^2 \Big{(}\frac{D_c - \sum\nolimits_{k = 1}^K \alpha_{k,t - 1}D_{k,c}}{D_c} \Big{)}^2,
\end{multline}
where $A_3 = 1 + (4L_u^2\eta_u^2 + 2L_v^2\eta_v^2)(1 + \chi)\tau^2 - (\eta_u L_u+\eta_v L_v)\tau $.
By telescoping the above inequality, we have
\begin{multline}\label{eq:oneR_temp}
F(\bm{W}_T) - F(\bm{W}^*) \le {A_3^T}( F(\bm{W}_0) - F(\bm{W}^*) ) + \frac{1 - A_3^T}{1 - A_3}(A_1 + A_2) \\
+ A_2\sum\nolimits_{t = 1}^{T - 1} A_3^{T - 1 - t}\sum\nolimits_{k = 1}^K \frac{1}{D}\frac{1}{D_k}C\sum\nolimits_{c = 1}^C \frac{D_{k,c}^2}{D_c^2} \Big{(} D_c - \sum\nolimits_{k = 1}^K \alpha_{k,t - 1}D_{k,c} \Big{)}^2.
\end{multline}
Below we bounding the last term on the RHS of \eqref{eq:oneR_temp} as
\begin{multline}\label{eq:oneR_im_t1}
A_2\sum\nolimits_{t = 1}^{T - 1} {A_3^{T - 1 - t}}\sum\nolimits_{k = 1}^K \frac{1}{D}\frac{1}{D_k}C\sum\nolimits_{c = 1}^C \frac{D_{k,c}^2}{D_c^2} \Big{(} D_c - \sum\nolimits_{k = 1}^K \alpha_{k,t-1}D_{k,c} \Big{)}^2 \\
= A_2\sum\nolimits_{t = 0}^{T - 2} {} A_3^{T - 2 - t}\sum\nolimits_{k = 1}^K \frac{1}{D}\frac{1}{D_k}C\sum\nolimits_{c = 1}^C \frac{D_{k,c}^2}{D_c^2}{\left(D_c - \sum\nolimits_{k = 1}^K \alpha_{k,t}D_{k,c} \right)^2} \\
\overset{(a)}\le \frac{A_2 KC}{D}\sum\nolimits_{t = 0}^{T - 2} A_3^{T - 2 - t}\sum\nolimits_{k = 1}^K \sum\nolimits_{c = 1}^C \frac{D_{k,c}^2}{D_k D_c^2}\sum\nolimits_{k = 1}^K (1 - \alpha_{k,t})D_{k,c}^2 \\
\!= \!\frac{A_2 CK}{D}\!\sum\limits_{t = 0}^{T - 2} \!A_3^{T - 2 - t}\!\sum\limits_{k = 1}^K \!\sum\limits_{c = 1}^C \!\frac{D_{k,c}^2}{D_k D_c^2}\!\sum\limits_{k = 1}^K \!D_{k,c}^2
- \frac{A_2 K}{D}\!\sum\limits_{t = 0}^{T - 2} \!A_3^{T - 2 - t}\!\sum\limits_{k = 1}^K \!C\!\sum\limits_{c = 1}^C \!\frac{D_{k,c}^2}{D_k D_c^2}\!\sum\limits_{k = 1}^K\! {\alpha _{k,t}}D_{k,c}^2,
\end{multline}
where (a) is due to Jensen's inequality and $(1 - \alpha_{k,t})^2 = 1 - \alpha_{k,t}$.
For the last term on the RHS of \eqref{eq:oneR_im_t1}, we have
\begin{multline}\label{eq:oneR_im_t2}
\frac{A_2K}{D}\sum\nolimits_{t = 0}^{T - 2} A_3^{T - 2 - t}\sum\nolimits_{k = 1}^K C\sum\nolimits_{c = 1}^C \frac{D_{k,c}^2}{D_k D_c^2}\sum\nolimits_{k = 1}^K {\alpha _{k,t}}D_{k,c}^2 \\
\overset{(a)}\ge A_2\frac{1}{D K (T - 1)}\left(\sum\nolimits_{t = 0}^{T - 2} A_3^{T - 2 - t}\sum\nolimits_{c = 1}^C \sum\nolimits_{k = 1}^K \frac{D_{k,c}}{\sqrt {D_k} D_c}\sum\nolimits_{k = 1}^K {\alpha _{k,t}}D_{k,c} \right)^2  \\
\ge A_2\frac{1}{D K (T - 1)}\frac{1}{\max_{1\le k \le K} D_k}\left(\sum\nolimits_{t = 0}^{T - 2} A_3^{T - 2 - t}\sum\nolimits_{k = 1}^K {\alpha _{k,t}}\sum\nolimits_{c = 1}^C D_{k,c} \right)^2 \\
= A_2\frac{1}{D K (T - 1)}\frac{1}{\max_{1\le k \le K} D_k}{\left(\sum\nolimits_{t = 0}^{T - 2} A_3^{T - 2 - t}\sum\nolimits_{k = 1}^K {\alpha_{k,t}}{D_k} \right)^2},
\end{multline}
where (a) follows Jensen's inequality.
Substituting \eqref{eq:oneR_im_t1} and \eqref{eq:oneR_im_t2} into \eqref{eq:oneR_temp}, the proof is completed.

\subsection{Proof of Proposition \rm{\ref{prop:ly_performance}} }\label{app:five}
For the ease of presentation, we define the Lyapunov function as $\mathcal{V}(t) = \sum\nolimits_{k = 1}^K {\frac{1}{2}} q_{k,t}^2$, the Lyapunov drift of round $t$ as ${\Delta_1}(t) = \mathcal{V}(t + 1) - \mathcal{V}(t)$. According to the evolution of the virtual queue defined in \eqref{eq:lyapunov_queue}, we have $q_{k,t + 1}^2 \le {\left( {{q_{k,t}} + \alpha_{k,t}E_{k,t} - \frac{E_k}{T}} \right)^2}$. For ${\Delta_1}(t)$, we have
\begin{multline}\label{eq:delta_1}
\Delta_1(t) = \frac{1}{2} \sum\nolimits_{k = 1}^K (q_{k,t + 1}^2 - q_{k,t}^2)
\le \sum\nolimits_{k = 1}^K  \Big( \frac{1}{2} ( q_{k,t} + \alpha_{k,t}E_{k,t} - \frac{E_k}{T} )^2 - \frac{1}{2}q_{k,t}^2 \Big) \\[-0.2cm]
\le \zeta_0 + \sum\nolimits_{k = 1}^K  q_{k,t} \Big{(} \alpha_{k,t} E_{k,t} - \frac{E_k}{T}\Big{)},
\end{multline}
where $\zeta_0 = \frac{1}{2}\sum\nolimits_{k = 1}^K \zeta_k^2$, ${\zeta_k} = {\max_t}\left\{ \left| \alpha_{k,t}E_{k,t} - \frac{E_k}{T} \right| \right\}$.
By adding $- V{\gamma_t}\sum\nolimits_{k = 1}^K  \alpha_{k,t}D_k$ on both sides of \eqref{eq:delta_1}, an upper bound of the one-round drift-plus-penalty function is given by
\begin{align}
\Delta_1(t) - V{\gamma_t}\sum\nolimits_{k = 1}^K  \alpha_{k,t}D_k \le \zeta_0 + \sum\nolimits_{k = 1}^K  {q_{k,t}}(\alpha_{k,t}E_{k,t} - \frac{E_k}{T}) - V{\gamma_t}\sum\nolimits_{k = 1}^K  \alpha_{k,t}D_k.
\end{align}
The drift-plus-penalty algorithm of Lyapunov optimization aims to minimize the upper bound of $\Delta_1(t) - {\gamma_t}V\sum\nolimits_{k = 1}^K  \alpha_{k,t}D_k$.
Define the $T$-round drift as $\Delta_T = \mathcal{V}(T - 1) - \mathcal{V}(0) = \sum\nolimits_{k = 1}^K \frac{1}{2}q_{k,T - 1}^2$. Then, the $T$-round drift-plus-penalty function can be bounded by
\begin{align}\label{eq:delta_T}
\Delta_T - V\sum\limits_{t = 0}^{T - 1}  \gamma_t\sum\limits_{k = 1}^K  \alpha_{k,t}D_k
\le T\zeta_0 + \sum\limits_{t = 0}^{T - 1}  \Big( \sum\limits_{k = 1}^K  q_{k,t}(\alpha_{k,t}E_{k,t} - \frac{E_k}{T}) - V\gamma_t\sum\limits_{k = 1}^K \alpha_{k,t}D_k \Big).
\end{align}
Based on the above analysis, we first prove the feasibility of the proposed algorithm. We use superscript * to denote the optimal offline solution of problem $\widetilde{\mathcal{P}}$, superscript $\dagger$ to represent the solution of the proposed drift-plus-penalty algorithm. For a feasible solution with ${\alpha_{k,t}} = 0$ and ${E_{k,t}} = 0$, we have
\begin{align}
\Delta_T = \sum\nolimits_{k = 1}^K  \frac{1}{2}q_{k,T - 1}^2 \le T\zeta_0 + V\sum\nolimits_{t = 0}^{T - 1} \gamma_t D.
\end{align}
Thus, we have
\begin{align}
\left( \sum\nolimits_{k = 1}^K  q_{k,T - 1} \right)^2 \le K\sum\nolimits_{k = 1}^K  q_{k,T - 1}^2 \le 2K\left( T\zeta_0 + V\sum\nolimits_{t = 0}^{T - 1} \gamma_t D \right),
\end{align}
where the first inequation comes from Jensen's inequality. According to the evolution of the virtual queue defined in \eqref{eq:lyapunov_queue}, we have $\alpha_{k,t}E_{k,t} - \frac{E_k}{T} \le {q_{k,t + 1}} - {q_{k,t}}$, summing this inequation over $T$ rounds, we have
\begin{align}
\sum\nolimits_{t = 0}^{T - 1} \! \sum\nolimits_{k = 1}^K \! \Big{(}\alpha_{k,t}E_{k,t} - \frac{E_k}{T} \Big{)}
\le \sum\limits_{t = 0}^{T - 1}  \sum\limits_{k = 1}^K ( q_{k,t + 1} - q_{k,t})
 \le \sqrt {2K\left( T\zeta_0 + V\sum\nolimits_{t = 0}^{T - 1}  \gamma_t D \right)}.
\end{align}
By rearranging the above inequation, the energy consumption bound in \eqref{eq:energy_bound} is derived. Below we analyze the optimality of the proposed drift-plus-penalty algorithm, which minimize the RHS in \eqref{eq:delta_T}. Since $\Delta_T$ is positive, based on \eqref{eq:delta_T}, we have
\begin{align}\label{eq:data_B_temp}
- V\sum\limits_{t = 0}^{T - 1}  {\gamma_t}\sum\limits_{k = 1}^K  \alpha_{k,t}^{\dagger} {D_{k,c}}
\le T\zeta_0 + \sum\limits_{t = 0}^{T - 1}  \sum\nolimits_{k = 1}^K  q_{k,t}\Big{(} \alpha_{k,t}^* E_{k,t} - \frac{E_k^*}{T} \Big{)} - V\sum\limits_{t = 0}^{T - 1} \gamma_t\sum\limits_{k = 1}^K \alpha_{k,t}^* D_{k,c},
\end{align}
Next, we bound the second term in the RHS of \eqref{eq:data_B_temp} as
\begin{multline}\label{eq:term_bnound}
\sum\limits_{t = 0}^{T - 1} \! \sum\limits_{k = 1}^K \!  q_{k,t} \! \Big{(} \alpha_{k,t}^*{E_{k,t}} \!-\! \frac{E_k}{T} \Big{)}
\!\!\!= \!\sum\limits_{t = 0}^{T - 1} \! \sum\limits_{k = 1}^K \!(q_{k,t} \!-\! q_{k,0})\Big{(} \alpha_{k,t}^* E_{k,t} \!-\! \frac{E_k}{T} \Big{)}
\!\!\le\! \frac{T(T - 1)}{2}\!\sum\limits_{k = 1}^K \! \zeta_k^2.
\end{multline}
Substituting \eqref{eq:term_bnound} into \eqref{eq:data_B_temp}, the inequation \eqref{eq:datavolume_bound} is derived, and the proof is completed.

\bibliographystyle{IEEEtran}
\bibliography{IEEEabrv,cited}

\end{titlepage}

\clearpage
\begin{titlepage}
\title{\Large{The Proofs and Additional Experimental Results in the Paper Titled ``Knowledge-aided Federated Learning for Energy-limited Wireless Networks"}\\
}
\author{Zhixiong~Chen,~\IEEEmembership{Student Member,~IEEE},
Wenqiang~Yi,~\IEEEmembership{Member,~IEEE},\\
Yuanwei~Liu,~\IEEEmembership{Senior Member,~IEEE},
and Arumugam~Nallanathan,~\IEEEmembership{Fellow,~IEEE}
\thanks{Zhixiong Chen, Wenqiang Yi, Yuanwei Liu, and Arumugam Nallanathan are with the School of Electronic Engineering and Computer Science, Queen Mary University of London, London, U.K. (emails: \{zhixiong.chen, w.yi, yuanwei.liu, a.nallanathan\}@qmul.ac.uk).}
}

\maketitle
\begin{abstract}
The conventional model/gradient aggregation-based federated learning (FL) approaches require all local models to be of the same architecture and thus may be inapplicable for many practical scenarios. Moreover, the frequent model/gradient exchange is costly for resource-limited wireless networks since modern deep neural networks usually have over-million parameters.
To tackle these challenges, we first devise a novel FL framework that aggregates light high-level data features, namely knowledge, in the per-round learning process. This design allows devices to design their machine models independently and remarkably reduces the communication overhead in the training process.
We then theoretically analyze the convergence bound of the framework under a non-convex loss function setting, revealing that scheduling more data volumes in each round helps improve the learning performance.  In addition, more scheduled data volumes should be biased towards the early rounds if the total data volumes during the entire learning course are fixed.
Inspired by this, we formulate an optimization problem to maximize the weighted scheduled data volumes for global loss minimization under the energy constraints of devices through device scheduling, bandwidth allocation and power control.
This paper provides the proof and additional experimental results of the journal version, namely ``Knowledge-aided Federated Learning for Energy-limited Wireless Networks". This paper
provides the proofs of Proposition 1, Lemma 3, and additional experimental results based on another heterogeneous data distribution setting. The other proposition, lemmas, and experimental results have been provided in the journal version or similar to the provided proofs.
\end{abstract}

\begin{IEEEkeywords}
Device scheduling, federated Learning, Lyapunov optimization, resource allocation
\end{IEEEkeywords}

\clearpage

\section{Proof A: The Proof of Proposition 1}
\renewcommand{\theequation}{A.\arabic{equation}}
\setcounter{equation}{0}
The first-order and second-order derivatives of the objective function (23) with respect to $\mathcal{T}_{k,t}^{\rm{U}}$ are
\begin{align}
\frac{\partial \left( \sum\nolimits_{k \in \bm{S}_t}  q_k(t)E_{k,t} \right)}{\partial \mathcal{T}_{k,t}^{\rm{U}}} = q_k(t)\frac{\theta_{k,t}B N_0 \mathcal{T}_{k,t}^{\rm{U}} - N_0 Qq\ln 2}{h_{k,t}\mathcal{T}_{k,t}^{\rm{U}}}{2^{\frac{Qq}{\theta_{k,t}B\mathcal{T}_{k,t}^{\rm{U}}}}} - \frac{q_k(t)\theta_{k,t}B N_0}{h_{k,t}},
\end{align}
and
\begin{align}
\frac{{{\partial ^2}\left(\sum\nolimits_{k \in \bm{S}_t}  q_k(t)E_{k,t} \right)}}{\partial (\mathcal{T}_{k,t}^{\rm{U}})^2} = \frac{q_k(t)Q^2 q^2 N_0{(\ln 2)^2}}{\theta_{k,t}Bh_{k,t}{\left( {\mathcal{T}_{k,t}^{\rm{U}}} \right)^3}}{2^{\frac{Qq}{\theta_{k,t}B\mathcal{T}_{k,t}^{\rm{U}}}}} \ge 0.
\end{align}
Thus, $\frac{\partial \left( \sum\nolimits_{k \in \bm{S}_t}  q_k(t)E_{k,t} \right)}{\partial \mathcal{T}_{k,t}^{\rm{U}}}$ is an increasing function with respect to $\mathcal{T}_{k,t}^{\rm{U}}$. Since $\mathop {\lim}\nolimits_{\mathcal{T}_{k,t}^{\rm{U}} \to \infty } \frac{dE_{k,t}^{\rm{U}}}{d\mathcal{T}_{k,t}^{\rm{U}}} = 0$, we have $\frac{\partial \left( \sum\nolimits_{k \in \bm{S}_t}  q_k(t)E_{k,t} \right)}{\partial \mathcal{T}_{k,t}^{\rm{U}}} \le 0$. That is, the objective function (23) is an non-increasing function with respect to the communication time $\mathcal{T}_{k,t}^{\rm{U}}$. The optimal completion time of device $k$ is $\mathcal{T}_{k,t}^{\rm{U}} = \mathcal{T}_{\max}-\mathcal{T}_k^{\rm{L}}$. Thus, the optimal transmit power of device $k$ satisfy (24).

\section{Proof B: Proof of Lemma 3 }
\renewcommand{\theequation}{B.\arabic{equation}}
\setcounter{equation}{0}
Problem $\mathcal{P}_2$ is a typical convex optimization problem, its proof is similar to the proof of Proposition 1, and thus omitted for brevity. By using KKT conditions, the Lagrange function of problem $\mathcal{P}_2$ is
\begin{align}
\mathcal{L}(\theta_t,\mu) = \sum\limits_{k \in \bm{S}_t}  q_k(t)\frac{\theta_{k,t}BN_0(\mathcal{T}_{\max} - \mathcal{T}_k^{\rm{L}})}{h_{k,t}} \mathcal{I}(\theta_{k,t}) + \mu \left( \sum\nolimits_{k = 1}^K \theta_{k,t} - 1 \right),
\end{align}
where $\mu$ is the Lagrange multiplier associated with constraint (14c). The first-order derivative of $\mathcal{L}(\theta_t,\mu)$ is
\begin{align}
\frac{\partial \mathcal{L}(\theta_t,\mu)}{\partial \theta_{k,t}} = \frac{B N_0 q_k(t)(\mathcal{T}_{\max} - \mathcal{T}_k^{\rm{L}})}{h_{k,t}}\left(\mathcal{I}(\theta_{k,t}) + \theta_{k,t}\mathcal{I}'(\theta_{k,t}) \right) + \mu.
\end{align}
Let $\frac{\partial \mathcal{L}(\theta_t,\mu)}{\partial \theta_{k,t}} = 0$, we have
\begin{align}
\mathcal{I}(\theta_{k,t}) + \theta_{k,t}\mathcal{I}'(\theta_{k,t}) = \frac{-\mu h_{k,t}}{{B N_0q_k(t)(\mathcal{T}_{\max} - \mathcal{T}_k^{\rm{L}})}}.
\end{align}
Its inverse function is shown to be (28). Given constraint (25a), the optimal bandwidth allocation policy is given as (27). In addition, similar to the proof of Proposition 1, one can prove that the objective function (25) is a decreasing function of $\theta_{k,t}$. Thus, $\sum\nolimits_{k = 1}^K \theta_{k,t}^* = 1$ always holds for the optimal solution.

\section{Additional Numerical Results}
In this section, we present the additional experiments based on the data heterogeneity setting of $m=3$, which shows a similar result to the setting of $m=2$.

Fig. \ref{fig:acc_sameM_S3} shows the learning performance of the proposed FL algorithm and two benchmarks on MNIST and CIFAR-10 datasets, where all the devices are equipped with homogeneous machine learning models.
Fig. \ref{fig:mnist_sameM_S3} presents the test accuracy on MNIST dataset. Compared to the baselines, the proposed algorithm achieves a 1.54\% accuracy improvement when 50 devices participate in each round learning process and obtains a 1.28\% accuracy gain when scheduling 10 devices in each round.
Fig. \ref{fig:cifar10_sameM_S3} presents the learning performance of these algorithms on the CIFAR-10 dataset, which also indicates that the proposed algorithm outperforms the benchmarks.
Fig. \ref{fig:schedule_ORsameM_S3} verifies the correctness of Remark 1, indicating that more scheduled data samples should be biased to the earlier rounds when the total scheduled data volumes in the entire learning course are fixed.
\begin{figure*}[ht]
\centering
\subfigure[]{\label{fig:mnist_sameM_S3}
\includegraphics[width=0.32\linewidth]{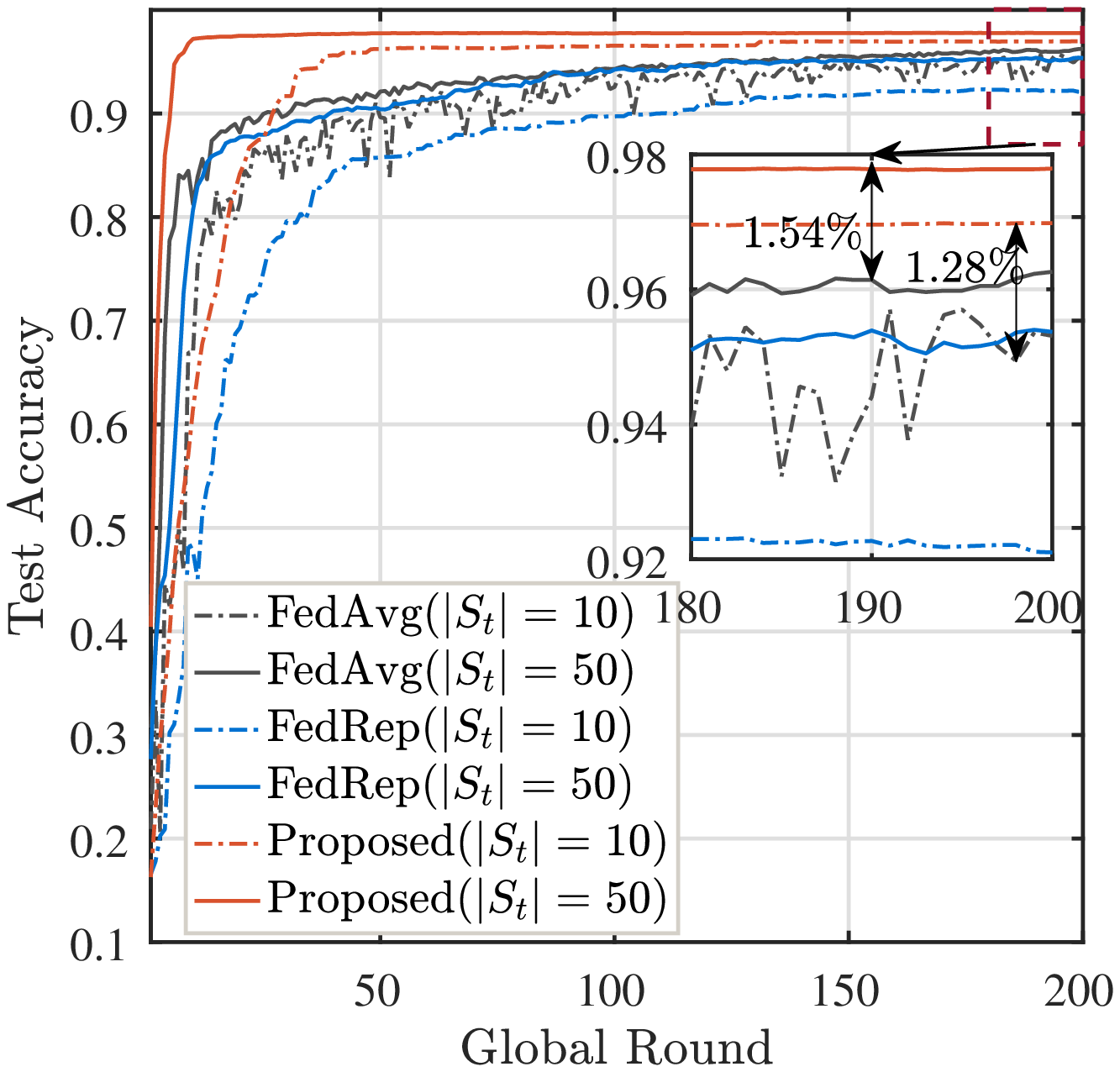}}
\hspace{-0.02\linewidth}
\subfigure[]{\label{fig:cifar10_sameM_S3}
\includegraphics[width=0.32\linewidth]{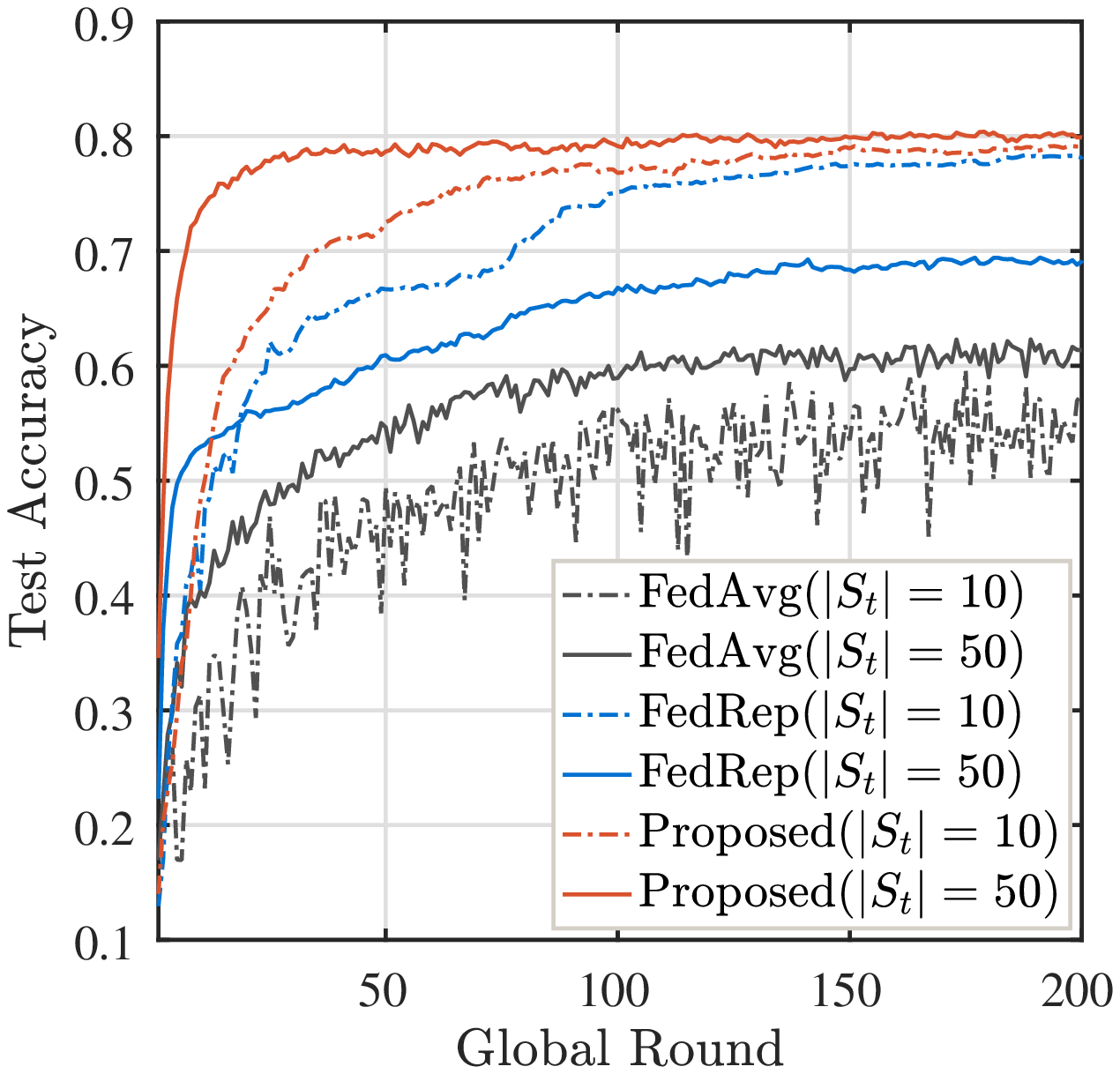}}
\hspace{-0.02\linewidth}
\subfigure[]{\label{fig:schedule_ORsameM_S3}
\includegraphics[width=0.32\linewidth]{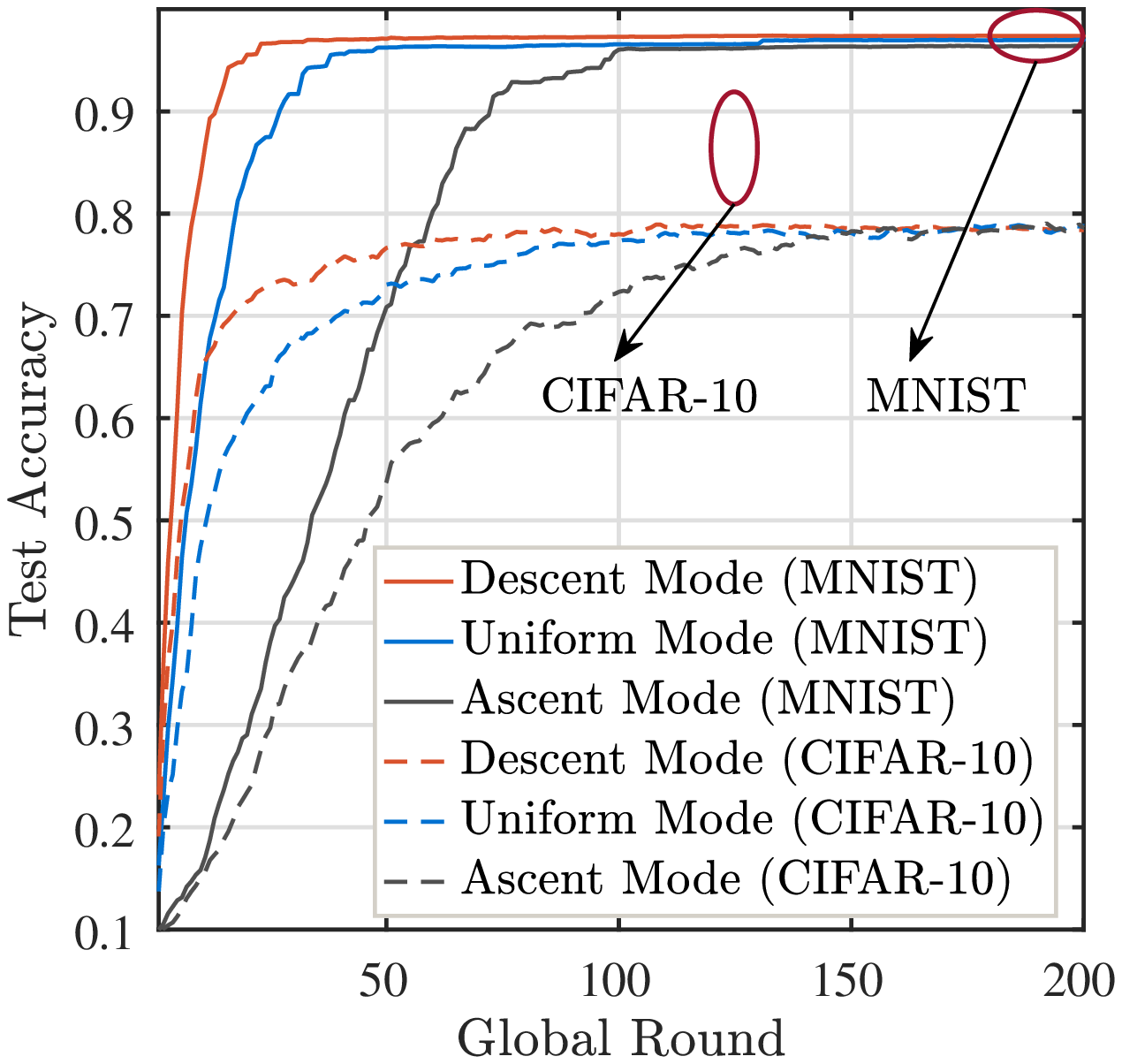}}
\caption{Comparison of learning performance under homogeneous models (a) different algorithms on the MNIST dataset; (b) different algorithms on the CIFAR-10 dataset; (c) different scheduling patterns on MNIST and CIFAR-10 datasets.}
\label{fig:acc_sameM_S3}
\end{figure*}

\begin{figure*}[ht]
\centering
\subfigure[]{\label{fig:mnist_diffM_S3}
\includegraphics[width=0.32\linewidth]{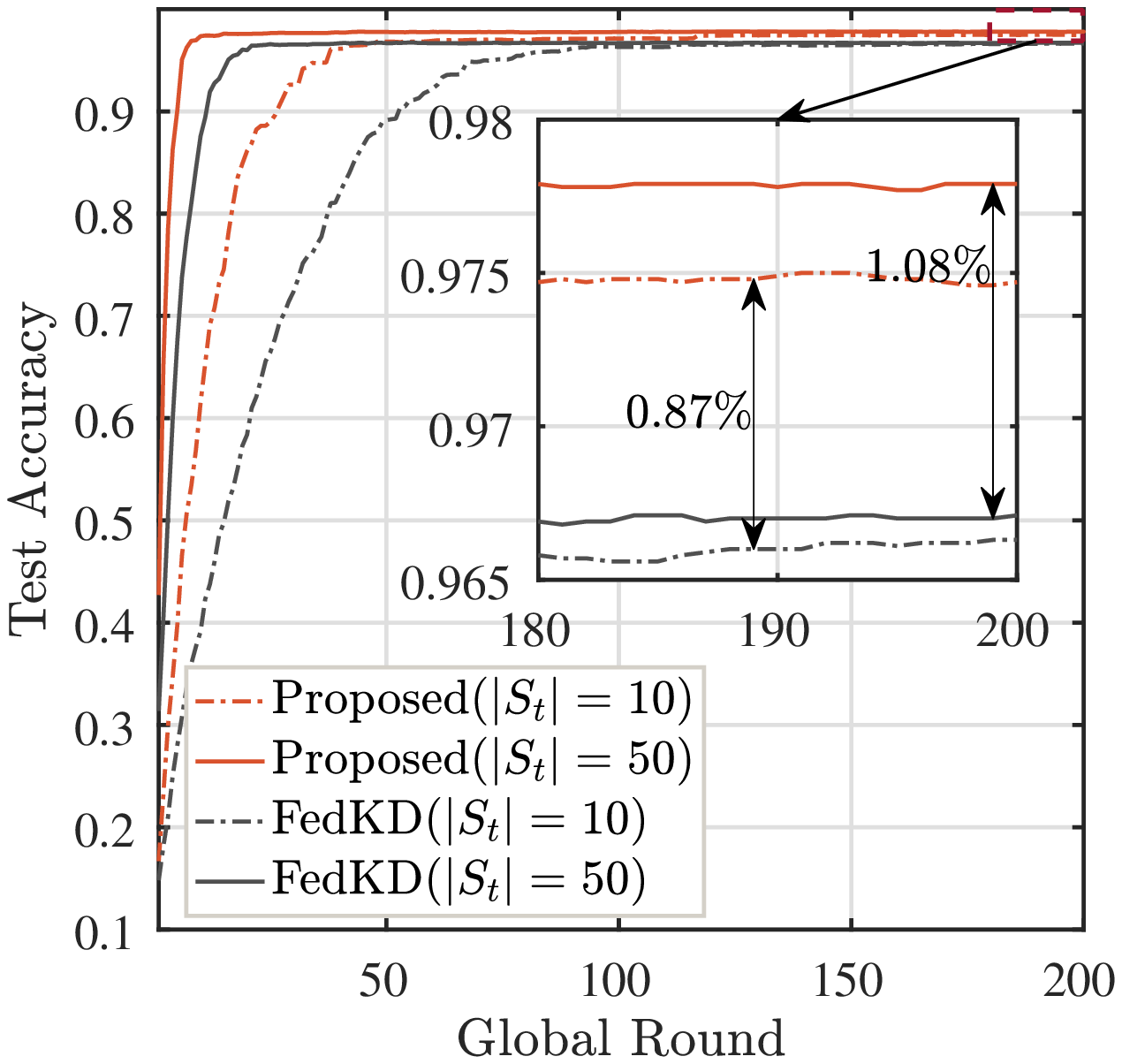}}
\hspace{-0.02\linewidth}
\subfigure[]{\label{fig:cifar10_diffM_S3}
\includegraphics[width=0.32\linewidth]{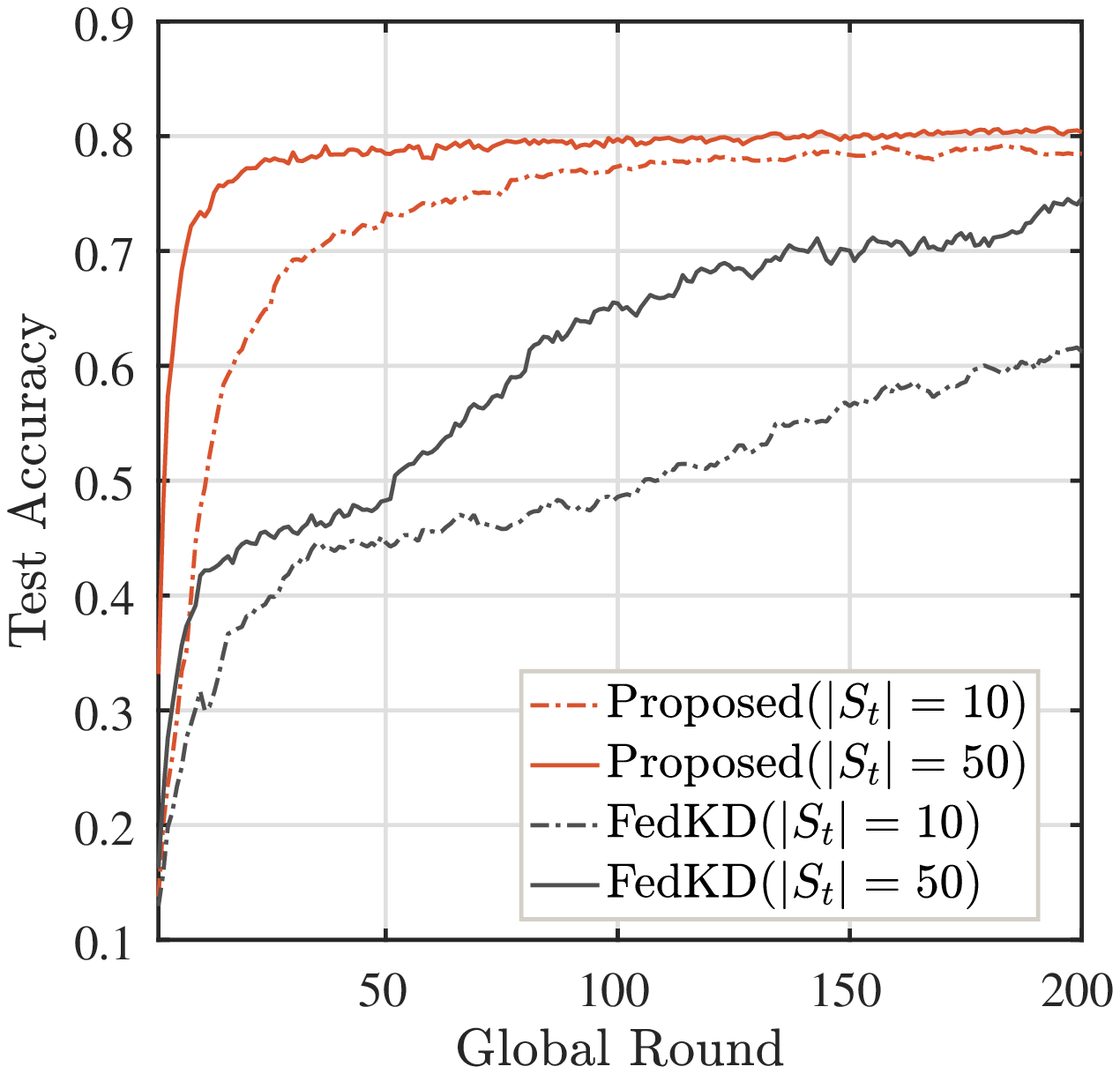}}
\hspace{-0.02\linewidth}
\subfigure[]{\label{fig:schedule_ORdiffM_S3}
\includegraphics[width=0.32\linewidth]{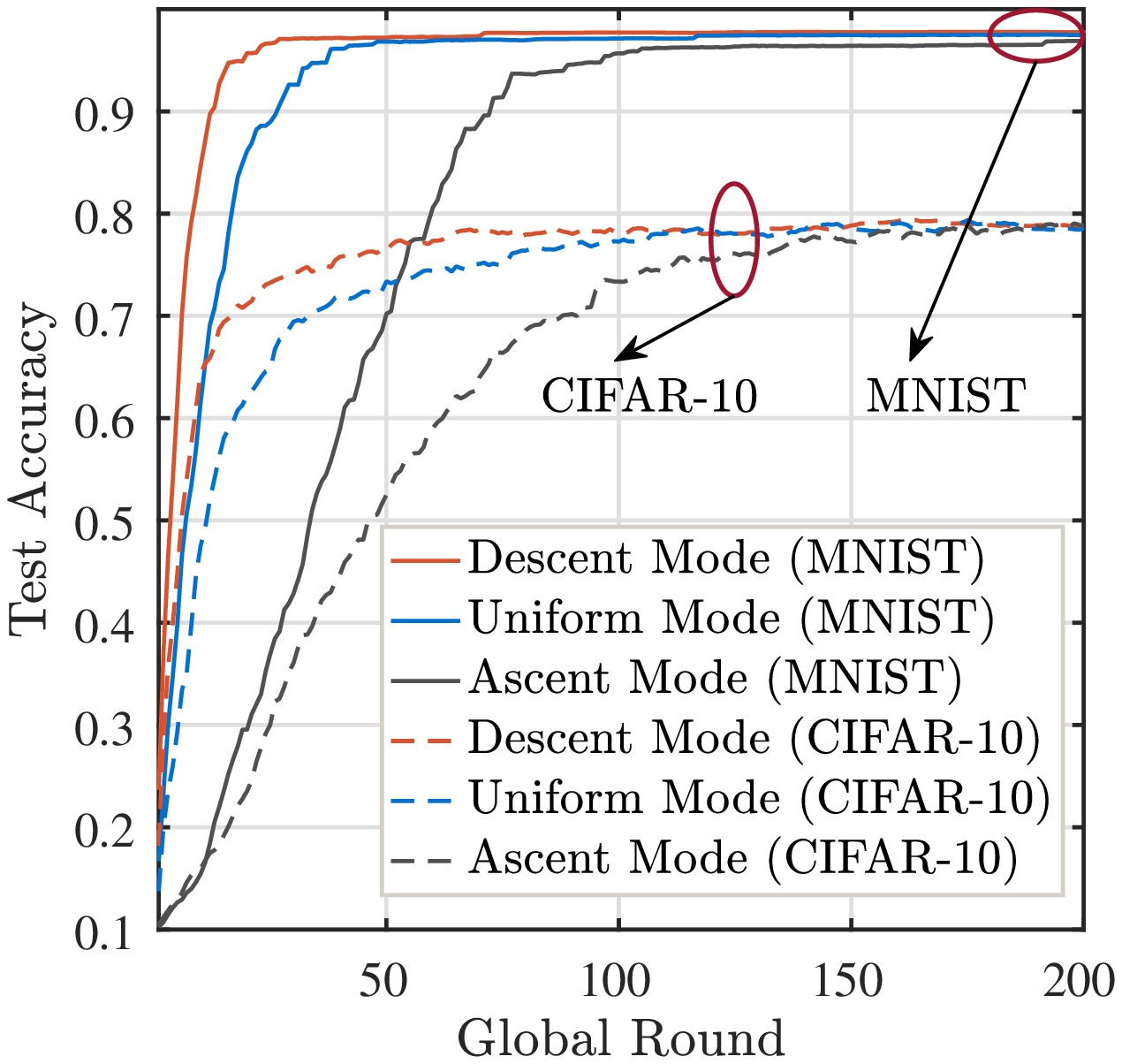}}
\vspace{-0.3cm}
\caption{Comparison of learning performance under heterogeneous models (a) different algorithms on the MNIST dataset; (b) different algorithms on the CIFAR-10 dataset; (c) different scheduling patterns on MNIST and CIFAR-10 datasets.}
\label{fig:acc_diffM_S3}
\end{figure*}

Fig. \ref{fig:acc_diffM_S3} compares the learning performance of the proposed knowledge aggregation-based FL algorithm and FedKD, where devices are equipped with heterogeneous models.
Fig. \ref{fig:cifar10_diffM_S3} presents the results on MNIST dataset. The proposed algorithm obtains 1.08\% and 0.87\% accuracy improvement when 50 and 10 devices participate in the learning process.
Fig. \ref{fig:cifar10_diffM_S3} also shows the proposed algorithm achieves better learning performance than the FedKD on the CIFAR-10 dataset.
Fig. \ref{fig:schedule_ORdiffM_S3} further verifies our theoretical results in Remark 1.

Fig. \ref{fig:scheduleE_S3} compare the test accuracy and cumulative energy usage of the scheduling algorithms on the MNIST and CIFAR-10 dataset.
Fig. \ref{fig:mnistE_accdiffM__S3} and Fig. \ref{fig:mnistE_enerdiffsameM__S3} presents the results on the MNIST dataset, where $\bar{E}=0.1$J and $T_{\max}=1$s. We can see that the proposed online device scheduling algorithm obtains a faster convergence speed and higher test accuracy than the benchmarks.
In particular, the proposed algorithm with $V=0.01$ has the same energy usage as the Adaptive Myopic algorithm, yet achieves better learning performance.
Fig. \ref{fig:cifarE_accdiffM_S3} and Fig. \ref{fig:cifarE_enerdiffM_S3} present the results on CIFAR-10 dataset, where $\bar{E}=0.5$J and $T_{\max}=2$s. It is also observed that the proposed online device scheduling algorithm outperforms the baselines in accuracy and convergence speed.

\begin{figure*}[ht]
\centering
\subfigure[]{\label{fig:mnistE_accdiffM__S3}
\includegraphics[width=0.25\linewidth]{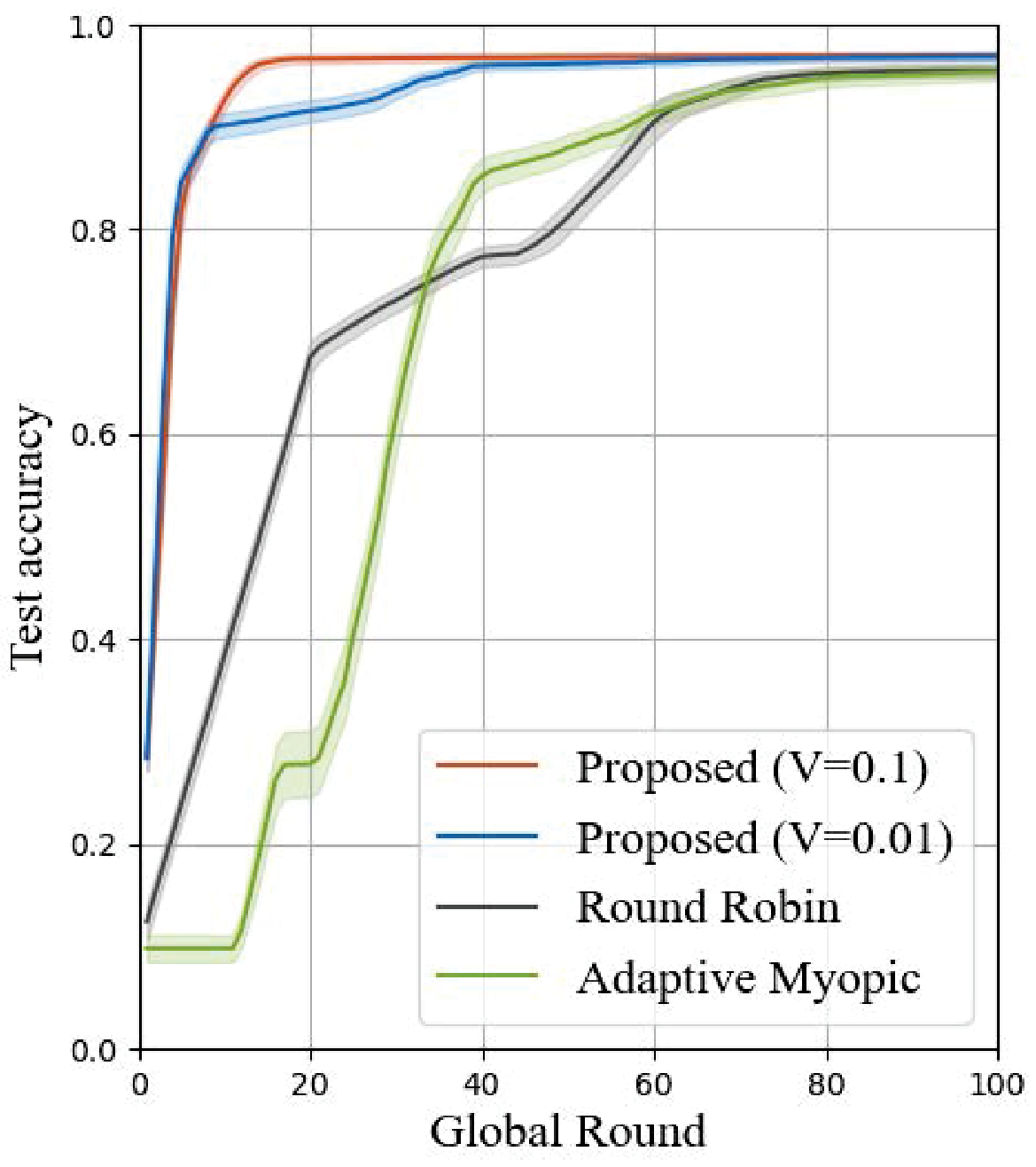}}
\hspace{-0.03\linewidth}
\subfigure[]{\label{fig:mnistE_enerdiffsameM__S3}
\includegraphics[width=0.25\linewidth]{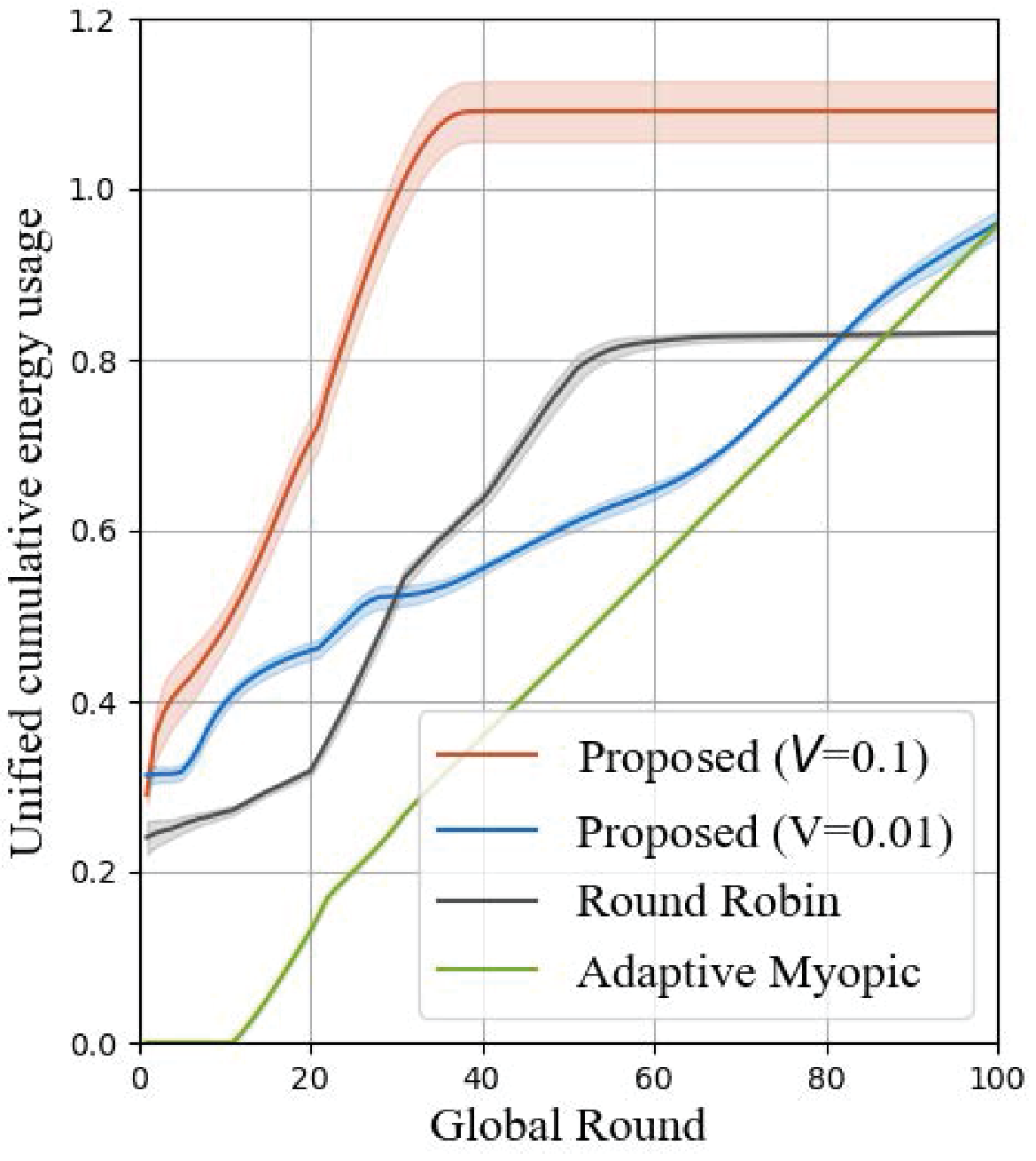}}
\hspace{-0.03\linewidth}
\subfigure[]{\label{fig:cifarE_accdiffM_S3}
\includegraphics[width=0.25\linewidth]{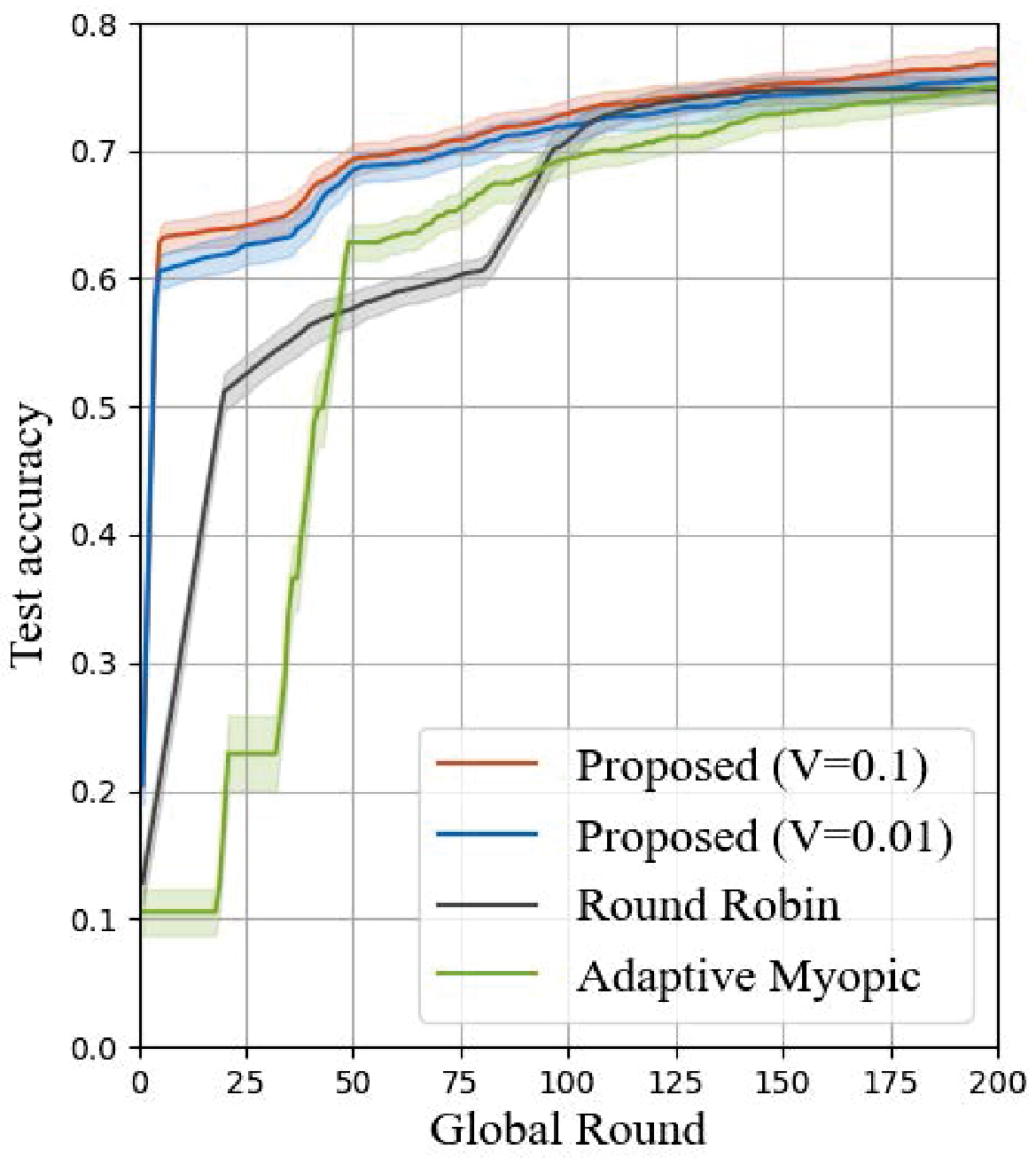}}
\hspace{-0.03\linewidth}
\subfigure[]{\label{fig:cifarE_enerdiffM_S3}
\includegraphics[width=0.25\linewidth]{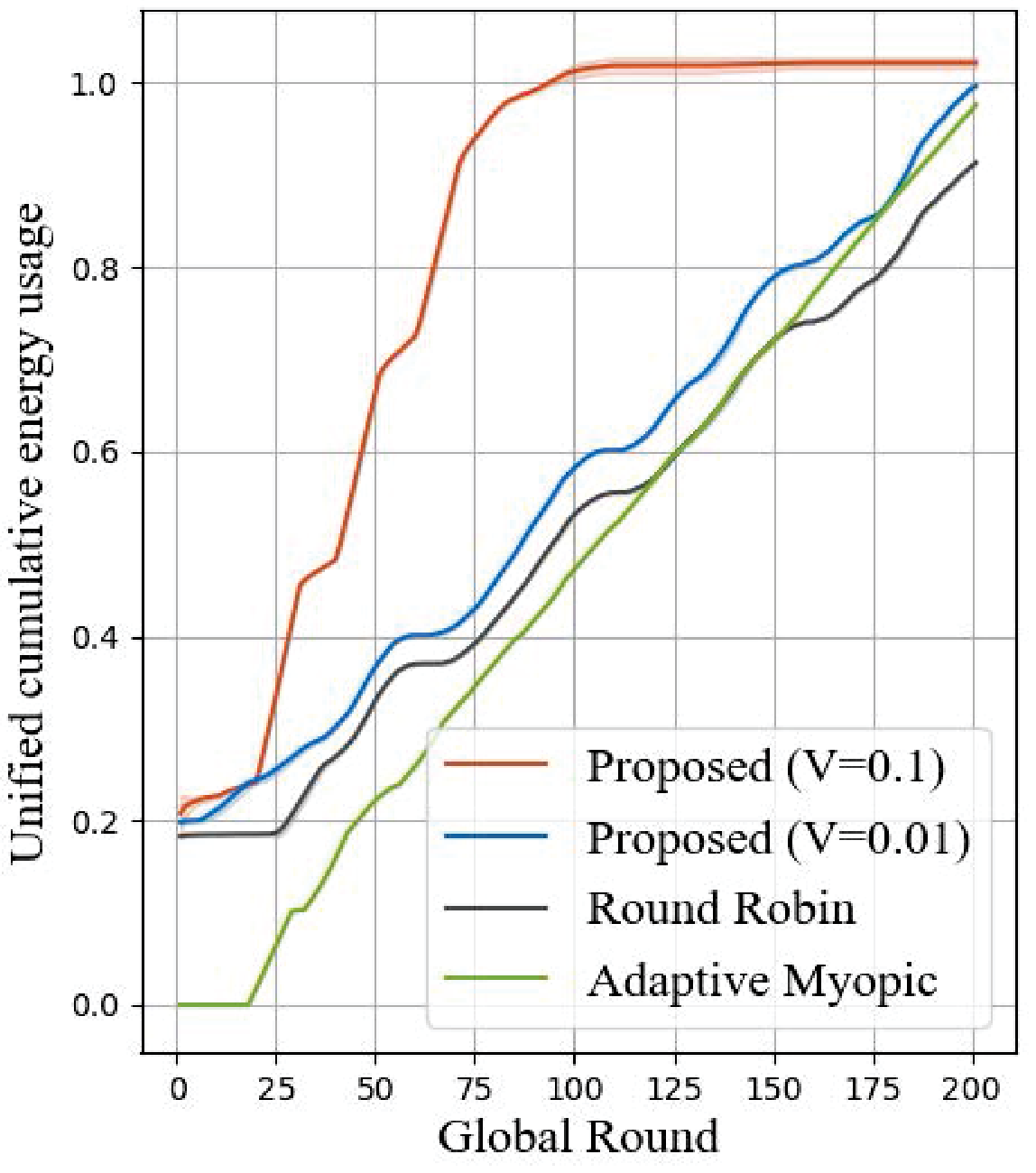}}
\vspace{-0.3cm}
\caption{Comparison of learning performance in different device scheduling algorithms on the MNIST and CIFAR-10 datasets.}
\label{fig:scheduleE_S3}
\end{figure*}

\end{titlepage}

\end{document}